\documentclass[10pt, conference, compsocconf]{IEEEtran}
\ifCLASSINFOpdf
\else
\fi
\usepackage{amsmath}
\usepackage[utf8]{inputenc} 
\usepackage{amssymb} 
\usepackage{amsmath}
\usepackage{listings} 
\usepackage{mathpartir} 
\usepackage{microtype} 
\usepackage{stmaryrd} 
\usepackage{graphicx} 
\usepackage[table,dvipsnames]{xcolor} 
\usepackage{url}
\usepackage{rtyps}
\usepackage{array} 
\usepackage{hhline}
\usepackage{calrsfs}
\usepackage{prettyref}
\usepackage{amsthm}
\usepackage{multirow}

\newcommand{\tstack}[0]{\ensuremath{\Sigma}}
\newcommand{\reststack}[0]{\ensuremath{\Sigma^r}}
\newcommand{\sectstack}[0]{\ensuremath{\Sigma^s}}

\newcommand{\type}[1]{\ensuremath{\text{#1}}}

\newcommand{\vid}[0]{\ensuremath{\text{VID}}}

\newcommand{\word}[1]{\ensuremath{\mathop{\,\mathsf{#1}\,}}}

\newcommand{\code}[1]{\text{\sl #1}}

\newcommand{\fid}[0]{\text{FID} }

\newcommand{\val}[0]{\ensuremath{\mathit{Val}}}

\newcommand{\dom}[1]{\text{dom}(#1)}

\newcommand{\share}[0]{\ensuremath{\curlyveedownarrow\!}}

\newcommand{\subtype}[0]{\ensuremath{<:}}

\newcommand{\bdatatypes}[0]{\ensuremath{\mathcal{T}}}

\newcommand{\datatypes}[0]{\ensuremath{\mathcal{A}}}
\newcommand{\fotypes}[0]{\ensuremath{\mathcal{F}}}

\newcommand{\secdatatypes}[0]{\ensuremath{\mathcal{S}}}
\newcommand{\secfotypes}[0]{\ensuremath{\mathcal{F}^s}}

\renewcommand{\int}[0]{\ensuremath{\mathit{int}}}

\newcommand{\tree}[2][XX]{%
   \ifthenelse{\equal{#1}{XX}}%
      {\ensuremath{T(#2)}}%
            {\ensuremath{T^{#1}(#2)}}}

\newcommand{\tr}[2][XX]{%
   \ifthenelse{\equal{#1}{XX}}%
      {\ensuremath{T(#2)}}%
            {\ensuremath{T^{#1}(#2)}}}

\newlength{\rWidth}

\newcommand{\funold}[3][XX]{%
   \ifthenelse{\equal{#1}{XX}}%
      {\ensuremath{#2 {\rightarrow} #3}}%
      { \settowidth{\rWidth}{\ensuremath{#1}}%
        \ensuremath{ #2 {\xrightarrow{\hspace{\rWidth}}\hspace{-0.84\rWidth}}\!\!\!^%
         {#1}
         \hspace{0.2\rWidth}\,\, #3}}}

\newcommand{\RuleToplabel}[4][]{\ensuremath{\inferrule[(#2)]{#3}{#4}}}
\newcommand{\RuleNolabel}[3][]{\ensuremath{\inferrule*[#1]{#2}{#3}}}

\newcommand{\Qplus}[0]{\ensuremath{\mathbb{Q}_{\geq 0}}}

\newcommand{\define}[0]{\ensuremath{::=}}
\newcommand{\poten}[1]{\ensuremath{\Phi_{#1}}}
\newcommand{\esim}[1]{\ensuremath{\approx_{\!#1}}}
\newcommand{\secesim}[1]{\ensuremath{\equiv_{#1}}}

\newcommand{\context}[0]{\ensuremath{\Gamma}}
\newcommand{\rescontext}[0]{\ensuremath{\Gamma^{r}}}
\newcommand{\seccontext}[0]{\ensuremath{\Gamma^{s}}}
\newcommand{\uppervars}[2]{[#1]_{\blacktriangleleft #2}}
\newcommand{\lowervars}[2]{_{#2\triangleleft}[#1]}
\newcommand{\hiddenvars}[2]{[#1]_{\not\blacktriangleleft #2}}
\newcommand{\midvars}[3]{_{#2\triangleleft}[#1]_{\blacktriangleleft#3}}

\newcommand{\ev}[0]{\ensuremath{E}}
\newcommand{\sizeval}[1]{\ensuremath{|#1|}}

\newcommand{\consume}[2]{\ensuremath{\text{consume}_{#1}{(#2)}}}

\newcommand{\transform}[0]{\ensuremath{\hookrightarrow}}
\newcommand{\llow}[0]{\ensuremath{\ell}}
\newcommand{\lhigh}[0]{\ensuremath{h}}

\newcommand{\atype}[1]{\ensuremath{\widehat{#1}}}

\newcommand{\qleakage}[0]{\ensuremath{C}}
\newcommand{\runcertain}[0]{\ensuremath{R}}
\newcommand{\envs}[1]{\ensuremath{\mathbf{E}^{#1}}}
\newcommand{\sizes}[0]{\ensuremath{N}}
\newcommand{\usage}[0]{\ensuremath{U}}

\newcommand{\lbound}[0]{\ensuremath{l}}
\newcommand{\ubound}[0]{\ensuremath{u}}

\newtheorem{lemma}{Lemma}

\newtheorem*{example*}{Example}

\newtheorem{theorem}{Theorem}

\newtheorem{definition}{Definition}
\newtheorem*{definition*}{Definition}
\newtheorem*{remark*}{Remark}

\newcommand{\wrt}[0]{w.r.t.}

\begin{document}
%
\title{Verifying and Synthesizing Constant-Resource Implementations with Types \\
\normalsize \emph{(Extended Version of the IEEE S\&P Oakland Paper)}}


\author{\IEEEauthorblockN{Van Chan Ngo \qquad  Mario Dehesa-Azuara \qquad Matthew Fredrikson \qquad Jan Hoffmann}
\IEEEauthorblockA{Carnegie Mellon University, Pittsburgh, Pennsylvania 15213\\
Email: channgo@cmu.edu, mdehazu@gmail.com, mfredrik@cs.cmu.edu, jhoffmann@cmu.edu}
}


%


\maketitle

\begin{abstract}
Side channel attacks have been used to extract critical data such as encryption keys and
confidential user data in a variety of adversarial settings. In practice, this
threat is addressed by adhering to a constant-time programming discipline,
which imposes strict constraints on the way in which programs are written.
This introduces an additional hurdle for programmers faced with the already
difficult task of writing secure code, highlighting the need for solutions
that give the same source-level guarantees while supporting more natural programming
models.

We propose a novel type system for verifying that programs correctly implement
constant-resource behavior. Our type system extends recent work on
\emph{automatic amortized resource analysis} (AARA), a set of techniques that
automatically derive provable upper bounds on the resource consumption of programs. 
We devise new techniques that build on the
potential method to achieve compositionality, precision, and automation.

A strict global requirement that a program always maintains constant resource
usage is too restrictive for most practical applications. 
It is
sufficient to require that the program's resource behavior remain constant
with respect to an attacker who is only allowed to observe part of the
program's state and behavior. To account for this, our type system
incorporates information flow tracking into its resource analysis. This allows
our system to certify programs that need to violate the constant-time
requirement in certain cases, as long as doing so does not leak confidential 
information to attackers. We formalize this guarantee by defining a new
notion of \emph{resource-aware noninterference}, and prove that our system
enforces it.

Finally, we show how our type inference algorithm can be used to synthesize a constant-time
implementation from one that cannot be verified as secure, effectively
repairing insecure programs automatically. We also show how a second novel
AARA system that computes \emph{lower bounds} on resource usage can be used to derive quantitative
bounds on the amount of information that a program leaks through its resource
use. We implemented each of these systems in Resource Aware ML, and show that
it can be applied to verify constant-time behavior in a number of applications
including encryption and decryption routines, database queries, and other
resource-aware functionality.
\end{abstract}

\begin{IEEEkeywords}
Language-based security; timing channels; information flow; resource analysis; static analysis
\end{IEEEkeywords}


%
\IEEEpeerreviewmaketitle

\section{Introduction}
\label{sec:introduction}
Side-channel attacks extract sensitive information about a program's state
through its observable use of resources such as time, network, and memory.
These attacks pose a realistic threat to the security of systems in a range 
of settings, in which the attacker has local
access to the native host~\cite{Kocher96}, through multi-tenant virtualized
environments~\cite{Reiter2015,Ristenpart2009}, or remotely over the
network~\cite{BrumleyB03}. Side channels have revealed highly-sensitive
data such as cryptographic
keys~\cite{Kocher96,BrumleyB03,CanvelHVV2003,AlFardanP13,GullaschBK11} and
private user
data~\cite{HaeberlenPA11,AndryscoKMJLS15,FeltenS00,BortzB07,ZhangJRR14}.

These attacks are mounted by taking repeated measurements of a program's
resource behavior, and comparing the resulting observations against a model
that relates the program's secret state to its resource usage. Unlike direct
information flow channels that operate over the input/output semantics of a
program, the conditions that give rise to side channels are oftentimes subtle
and therefore difficult for programmers to identify and mitigate. This also
poses a challenge for automated tool support aimed at addressing such
problems---whereas direct information flow can be described in terms of
standard program semantics, a similar precise treatment of side channels
requires incorporating the corresponding resource into the semantics and
applying quantitative reasoning.

This difficulty has led previous work 
in the area 
to treat resource use
indirectly, by reasoning about the flow of secret information into branching
control flow or other operations that might affect resource
use~\cite{AlmeidaBBDE16,RodriguesQFA16,BartheBCLP14,MolnarPSW}. These
approaches can limit the expressiveness of secure programs and further
complicate the development. For example, by requiring
programmers to write code using a ``constant-time discipline'' that forbids
the use of variables influenced by secret state in statements that could affect
the program's control path~\cite{AlmeidaBBDE16}. 

\paragraph{Verifiable constant-resource language}
In this paper, we present a novel type system that gives developers the ability to certify that
their code is secure against resource side-channel attacks \wrt{} a high-level attack 
model, in which the resource consumption of each language construct is modeled by a constant. 
Our approach reduces constraints on the expressiveness of programs
that can be verified, and does not introduce general stylistic guidelines that must be
followed in order to ensure constant-resource behavior. Programmers
write code in typical functional style and annotate variables with
standard types. Thus, it does not degrade the readability of the code. 
At compile time, our verifier performs a quantitative
analysis to infer additional type information that
characterizes the
resource usage. From this, constant-resource behavior \wrt{} the high-level 
model on all executions of the program is determined automatically. 

The granularity with which our resource guarantees hold against an attacker
who can measure the total quantity of consumed resources is roughly equivalent
to what can be obtained by adhering to a strict constant-time
programming discipline. The certified constant-resource programs prevent side-channels 
that are inherent in implementing algorithms \wrt{} the provided high-level attack model. 
For example, if the resource under consideration is
execution time, measured by the number of language constructs executed by the 
program (e.g., the total number of arithmetic operations, function calls, etc.),
then our system provides a defense against attackers that can observe the same resource measure.
To have a stronger guarantee, e.g., against cache side channels, our resource model could in principle incorporate memory-access patterns and instruction caches. Other types of side channels arising from low-level behaviors, such as branch prediction or instructions whose resource usage is influenced by argument values, require corresponding changes to the resource model.
Our technique does not currently model such timing differences, so is not a defense against
such attacks.

In general, requiring that a program always consumes constant resources is too
restrictive. In most settings, it is sufficient to ensure that the resource
behavior of a program does not depend on selected confidential parts of the
program's state. To account for this, our type system tracks information flow
using standard techniques, and uses this information to reason about an
adversary who can observe and manipulate public state as well as resource
usage through public outputs. Intuitively,
\emph{resource-aware noninterference}---the guarantee enforced by this type
system---requires that the parts of the
program that are both affected by secret data and can influence public
outputs, can only make constant use of resources.

To accomplish this without limiting expressiveness or imposing
stylistic requirements, the type system must be allowed to freely switch
between local and global reasoning. One extreme would be to ignore the
information flow of the secret values and prove that the whole program has
global constant resource consumption. The other extreme would be to ensure
that every conditional that branches on a secret value (\emph{critical}
conditionals) uses a constant amount of resources. However, there are
constant-resource programs in which individual conditionals are not locally
constant-resource (see Section~\ref{sec:securitytype}). As a result, we allow
different levels of global and local reasoning in the type system to ensure
that every critical conditional occurs in a constant-resource block.

Finally, we show that our type-inference algorithm can be used to
automatically repair programs that make inappropriate non-constant use of
resources, by \emph{synthesizing} constant-resource ones whose input/output
behavior is equivalent. To this end, we introduce a \code{consume} expression
that performs resource padding. The amount of resource padding that is needed
is automatically determined by the type system and is parametric in the values
held by program variables. An advantage of this technique over prior
approaches~\cite{AskarovZM10,KopfD09} is that it does not change the
worst-case resource behavior of many programs. Of course, it would be possible
to perform this transformation by padding resource usage dynamically at the
end of the program execution, but this would require instrumenting the program
to track at runtime the actual resource usage of the program.
\paragraph{Novel resource type systems} In order to verify constant
resource usage, as well as to produce quantitative upper and
lower-bounds on information leakage via resource behavior, this work
extends the theory behind automatic amortized resource analysis (AARA)~\cite{Jost03,HoffmannAH11,HoffmannW15}
to automatically derive lower-bound and constant-resource proofs.

Previous AARA techniques are limited to deriving upper bounds. 
To this end, the resource potential is used as an \emph{affine}
quantity: it must be available to cover the cost of the execution, but excess
potential is simply discarded. We show that if potential is treated as a
\emph{linear} resource, then corresponding type derivations prove that programs have
constant resource consumption, i.e., the resource consumption is
independent of the execution path. Intuitively, this amounts to requiring that
\emph{all} potential must be used to cover the cost and that excess potential
is not wasted. Furthermore, we show that if potential is treated as a
\emph{relevant} resource then we can derive lower bounds on the resource
usage. Following a similar intuition, this requires that all potential is
used, but the available potential does not need to be sufficient to cover the
remaining cost of the execution.

We implemented these type systems
in Resource Aware ML (RAML)~\cite{HoffmannW15}, a language that supports
user-defined data types, higher-order functions, and other features common to
functional languages. Our type inference uses efficient LP solving
to characterize resource usage for general-purpose programs in this language.
We formalized soundness proofs for these type systems, as well as the one of
classic linear AARA~\cite{Jost03}, in the proof assistant Agda. The soundness
is proved \wrt{} an operational cost semantics and, like the type systems
themselves, is parametric in the resource of interest.
\paragraph{Contributions} We make the following contributions:
\begin{itemize}

\item A security type system that incorporates
  our novel lower-bound and constant-time type systems to prevent and quantify
  leakage of secrets through resource side channels, as well as an
  LP-based method that automatically transforms programs into constant-resource versions.

\item An implementation of these systems that extends RAML. We evaluate the
  implementation on several examples, including encryption routines and data
  processing programs that were previously studied in the context of timing
  leaks in differentially-private systems~\cite{HaeberlenPA11}.

\item A mechanization of the soundness proofs the two new type systems
  and classic AARA for upper bounds in Agda. To the best our knowledge, this
  is also the first formalization of the soundness of linear AARA for
  worst-case bounds.

\end{itemize}
Technical details including the complete proofs and inference rules can be
found on the RAML website~\cite{HoffmannW15}.


\section{Language-level constant-resource programs}
In this section, we define our notion of a constant-resource program. We start
with an illustrative example: a login with a username and password. During the
login process, the secret password with a high security level is compared with
the low-security user input, and the result is sent back to the user. As
a result, the pure noninterference property~\cite{Goguen82,McLean90} is
violated because data flows from high to low. Nevertheless, such a program is
often considered secure because it satisfies the \emph{relaxed
noninterference} property~\cite{Myers98,Myers00,Li05}.

Fig.~\ref{fig:compare} shows an implementation of the login process in a 
monomorphically-typed purely-functional language. The arguments
$h$ and $l$ are lists of integers that are the bytes of the password and the user
input (characters of the hashes). The function returns \code{true} if
the input is valid and \code{false} otherwise.

This implementation is vulnerable against an attacker who measures the
execution time of the login function. Because the function returns
\code{false} immediately on finding a mismatched pair of bytes, the resource
usage depends on the size of the longest matching prefix. Based on this
observation, the attacker can mount an efficient attack to recover the correct
password byte-by-byte. For example, if we assume that there is no noise
in the measurements, it requires at most $256 = 2^8$ calls to the function to
reveal one byte of the secret password. Thus, at most $256*N$ runs are needed
to recover a secret password of $N$ bytes.
If noise is added to the measurements then the number of necessary guesses
is increased but the attack remains feasible~\cite{Kocher96,BrumleyB03}.

One method to prevent this sort of attack is to develop a
\emph{constant-resource} implementation of the \code{compare} function
that minimizes the information that an attacker can learn from the
resource-usage information. Ideally, the resource usage should not be
dependent on the content of the secret password, which means it is
constant for fixed sizes of all public parameters.
%
\paragraph{Syntax and semantics} We use the purely-functional first-order language defined in
Fig.~\ref{fig:language} to formally define the notion of a
language-level constant-time implementation. The grammar is written
using abstract binding trees~\cite{Foundationpl}. However, equivalent
expressions in OCaml syntax are used for examples. The expressions are
in \emph{let normal form}, meaning that they are formed from variables
whenever it is possible. It makes the typing rules and semantics
simpler without losing expressivity. The syntactic form \emph{share}
has to be use to introduce multiple occurrences of a variable in an
expression.  A value is a boolean constant, an integer value $n$,
the empty list $\word{nil}$,
a list of values $[v_1,...,v_n]$, or a pair of values $(v_1,v_2)$.
\begin{figure}[!t]
\begin{lstlisting}
let rec compare(h,l) = match h with
| [] -> match l with | [] -> true 
                     | y::ys -> false
| x::xs -> match l with
  | [] -> false
  | y::ys -> if (x = y) then compare(xs,ys)
             else false
\end{lstlisting}
\vspace{-.2cm}
\caption{The list comparison function \code{compare} is not constant resource
\wrt{} $h$ and $l$. This implementation is insecure against an attacker who
measures its resource usage.}
\label{fig:compare}
\end{figure}
\begin{figure}[!thb]
\vspace{-.2cm}
\begin{displaymath}
\begin{array}{rl}
T & ::= \type{unit} \mid \type{bool} \mid \type{int} \mid L(T) \mid T * T \\
G & ::= T \rightarrow T \\
e & ::= \word{()} \mid \word{true} \mid \word{false} \mid n \mid x \mid \text{op}_{\diamond}(x_1,x_2) \mid \text{app}(f,x) \\
& \mid \text{if}(x,e_t,e_f) \mid \text{let}(x,e_1,x.e_2) \mid \text{pair}(x_1,x_2) \mid \text{nil} \\
& \mid \text{match}(x,(x_1,x_2).e) \mid \text{cons}(x_1,x_2) \\
& \mid \text{match}(x,e_1,(x_1,x_2).e_2) \mid \text{share}(x,(x_1,x_2).e) \\
v & ::= \word{()} \mid \word{true} \mid \word{false} \mid n \mid \word{nil} \mid [v_1,...,v_n] \mid (v_1,v_2) \\
\diamond & \in \{{+},{-},{*},\word{div},\word{mod},{=},{<>},{>},{<},\word{and},\word{or}\}
\end{array}
\end{displaymath}
\vspace{-.2cm}
\caption{Syntax of the language}
\label{fig:language}
\end{figure}
To reason about
the resource consumption of programs, we first define the operational 
cost semantics of the language. It is standard big-step semantics instrumented with
a non-negative resource counter
that is incremented or decremented by a constant at every step.
The semantics is parametric in the cost that is used at
each step and we call a particular set of such cost parameters a
\emph{cost model}. The constants can be used to indicate the costs of
storing or loading a value in the memory, evaluating a primitive
operation, binding of a value in the environment, or branching on a
Boolean value.
It is possible to further parameterize some constants to obtain a more
precise cost model. For example, the cost of calling a function may
vary according to the number of the arguments. In the following, we
will show that the soundness of type systems does not rely on
any specific values for these constants. In the examples, we use a
cost model in which the constants are $0$
for all steps except for calls to the \emph{tick} function where
\code{tick(q)} means that we have resource usage
$q \in \mathbb{Q}$.
A negative number specifies that resources (such as stack space)
become available.

The cost semantics is formulated
using an \emph{environment} $\ev: \vid \rightarrow \val$ that is a
finite mapping from a set of variable identifiers to values.
%
Evaluation judgements are of the form
$\ev \rtyps{q}{q'} e \Downarrow v$ where $q, q' \in \mathbb{Q}^+_0$.
The intuitive meaning is that under the environment $\ev$ and $q$
available resources, $e$ evaluates to the value $v$ without running
out of resources and $q'$ resources are available after the
evaluation. The evaluation consumes $\delta = q - q'$ resource
units. Fig.~\ref{fig:bosrules} presents some selected evaluation
rules. In the rule \textsc{E:Fun} for function applications, $e_g$ is an
expression defining the function's body and $x^g$ is the argument.

\begin{figure*}[!t]
\vspace{-.2cm}
\begin{mathpar}
\RuleToplabel{E:Bin}
{
	v = \ev(x_1) \diamond \ev(x_2)
}
{
	\ev \rtyps{q + K^{\word{op}}}{q}  \text{op}_{\diamond}(x_1,x_2) \Downarrow v
}

\RuleToplabel{E:Fun}
{
	E[x^g \mapsto \ev(x)] \rtyps{q}{q'} e_g \Downarrow v
}
{
	\ev \rtyps{q+K^{\word{app}}}{q'} \text{app}(g,x) \Downarrow v
}

\RuleToplabel{E:Let}
{
	\ev \rtyps{q - K^{\word{let}}}{q'_1} e_1 \Downarrow v_1 \\
	\ev[x \mapsto v_1] \rtyps{q'_1}{q'} e_2 \Downarrow v}
{
	\ev \rtyps{q}{q'} \text{let}(x,e_1,x.e_2) \Downarrow v
}

\RuleToplabel{E:Var}
{
	x \in \dom{\ev}
}
{
	\ev \rtyps{q + K^{\word{var}}}{q} x \Downarrow \ev(x)
}
\hfill
\RuleToplabel{E:If-True}
{
	\ev(x) = \word{true} \\
	\ev \rtyps{q - K^{\word{cond}}}{q'} e_t \Downarrow v
}
{
	\ev \rtyps{q}{q'} \text{if}(x,e_t,e_f) \Downarrow v
}
\hfill
\RuleToplabel{E:Match-L}
{
	\ev(x) = [v_1,...,v_n] \\\\
	\ev[x_h \mapsto v_1,x_t \mapsto [v_2,...,v_n]] \rtyps{q - K^{\word{matchL}}}{q'} e_2 \Downarrow v
}
{
	\ev \rtyps{q}{q'} \text{match}(x,e_1,(x_h,x_t).e_2) \Downarrow v
}
\end{mathpar}
\vspace{-.2cm}
\caption{Selected evaluation rules of the operational cost semantics}
\label{fig:bosrules}
\end{figure*}

\begin{figure}[!t]
\vspace{-.2cm}
\begin{lstlisting}
let rec p_compare(h,l) =
let rec aux(r,h,l) = match h with
| [] -> match l with | [] -> tick(1.0); r
        | y::ys -> tick(1.0); false
| x::xs -> match l with
  | [] -> tick(1.0); false
  | y::ys -> if (x = y) then
      tick(5.0); aux(r,xs,ys)
    else tick(5.0); aux(false,xs,ys)
in aux(true,h,l)
\end{lstlisting}
\vspace{-.2cm}
\caption{The manually padded function \code{p\_compare} is constant resource
\wrt{} $h$ and $l$. However, it is not constant resource \wrt{} only $h$.}
\label{fig:pcompare}
\end{figure}

\paragraph{Constant-resource programs} Let $\context: \vid \rightarrow \bdatatypes$ be a context that maps variable
identifiers to base types $T$. We write $\models v: T$
to denote that $v$ is a well-formed value of type $T$.
The typing rules for values are standard~\cite{Jost03,HoffmannAH11,sidechanTR}
and we omit them here. An environment $\ev$
is \emph{well-formed} \wrt{} $\context$,
denoted $\models \ev:\context$, if $\forall x \in \dom{\context}. \models \ev(x): \context(x)$. 
Below we define the notation of \emph{size equivalence}, written
$\sizeval{v} \esim{} \sizeval{u}$, which is a binary relation relating two values of the same type.
{
\def \MathparLineskip {\lineskip=0.1cm}
\begin{mathpar}
\RuleNolabel
{ \models v: T \\ \models u: T \\
	T \in \{\type{unit}, \type{bool},\type{int}\}
}
{
	\sizeval{v} \esim{} \sizeval{u}
}

\RuleNolabel
{
	\sizeval{v_1} \esim{} \sizeval{u_1} \\
	\sizeval{v_2} \esim{} \sizeval{u_2}
}
{
	\sizeval{(v_1,v_2)} \esim{} \sizeval{(u_1,u_2)}
}

\RuleNolabel
{
	m = n \\
	\sizeval{v_i} \esim{} \sizeval{u_i}
}
{
	\sizeval{[v_1,...v_n]} \esim{} \sizeval{[u_1,...u_m]}
}
\end{mathpar}}

Informally, a program is \emph{constant resource} if it has the same
quantitative resource consumption under all environments in which
values have the same sizes.
Let $X \subseteq \dom{\context}$ be a set of variables and $\ev_1$, $\ev_2$
be two well-formed environments. Then $\ev_1$ and $\ev_2$ are
\emph{size-equivalent} \wrt{} $X$,
denoted $\ev_1 \esim{X} \ev_2$,
when they agree on the sizes of the variables in $X$,
that is,
$\forall x \in X. \sizeval{\ev_1(x)} \esim{} \sizeval{\ev_2(x)}$.
\begin{definition}
\label{def:constantresource}
An expression $e$ is constant resource \wrt{} $X \subseteq \dom{\context}$, written $\type{const}_X(e)$, if for all well-formed environments $\ev_1$ and $\ev_2$ such that $\ev_1 \esim{X} \ev_2$, the following statement holds.
\begin{displaymath}
	\text{If } \ev_1 \rtyps{p_1}{p_1'} e \Downarrow v_1 \text{ and } \ev_2 \rtyps{p_2}{p_2'} e \Downarrow v_2 \text{ then } p_1 - p_1' = p_2 - p_2'
\end{displaymath}
\end{definition}

We say that a function $g(x_1,\ldots,x_n)$
is constant resource \wrt{} $X$ if $\type{const}_X(e_g)$ where $e_g$ is the expression
defining the function body. We have the following lemma.
\begin{lemma}
\label{lemm:equalpotential}
	For all $e$, $X$, and $Y \subseteq X$, if $\type{const}_Y(e)$ then $\type{const}_X(e)$.
\end{lemma}
\begin{example*}
The function \code{p\_compare} in Fig.~\ref{fig:pcompare} is a manually padded version of \code{compare},
in which the cost model is
defined using \code{tick} annotations. It is constant resource \wrt{} $h$ and
$l$. However, it is not constant resource \wrt{} $h$. For instance,
\code{p\_compare([1;2;3],[0;1;2])} has cost $16$ but
\code{p\_compare([1;2;1],[0;1])} has cost $12 \neq 16$. 
If the nil case of the second match on $l$ is padded with \code{tick(5.0); aux(false,xs,[])} then the function is constant resource \wrt{} $h$.
\vspace{-.2cm}
\end{example*}
Intuitively, this implementation is constant \wrt{} the given cost
model for fixed sizes 
of all public
parameters, e.g., the lengths of argument lists.  However, it might be
not constant resource at a lower level, e.g., machine code
on modern hardware, because the cost model does not precisely
capture the resource consumption of the instructions executed on the
hardware. Moreover, the compilation process can interfere with the
resource behavior. It may introduce a different type of leakage that
could reveal the secret data on the lower level. For instance, memory
accesses would allow an attacker with access to the full trace of
memory addresses accessed to infer the content of the password. This
leakage can be exploited via cache-timing
attacks~\cite{timingAES,Percival05}. In addition, in
some modern processors, execution time of arithmetic operations may
vary depending on the values of their operands and the execution time
of conditionals is affected by branch prediction.


\section{A resource-aware security type system}
\label{sec:securitytype}
In this section we introduce a new type system that enforces
\emph{resource-aware noninterference} to prevent the leakage of information in
\emph{high-security} variables through \emph{low-security} channels. In
addition to preventing leakage over the usual input/output information flow
channels, our system incorporates the constant-resource type system
discussed in Section~\ref{sec:resourcetype} to ensure that leakage does not
occur over resource side channels.

The notion of security addressed by our type system considers an attacker who
wishes to learn information about secret data by
making observations of the program's public outputs and resource usage. We
assume an attacker who is able to control the value of any variable she is
capable of observing, and thus to influence the program's behavior and resource
consumption. However, in our model the attacker can only observe the program's
total resource usage upon termination, and cannot distinguish between
intermediate states or between terminating and non-terminating executions.
\subsection{Security types}
To distinguish parts of the program under the attacker's control from those
that remain secret, we annotate types with labels ranging over a lattice
$(\mathcal{L}, \sqsubseteq, \sqcup, \perp)$. The elements of $\mathcal{L}$
correspond to security levels partially-ordered by $\sqsubseteq$ with a unique
bottom element $\perp$. The corresponding basic security types take the
form:
\begin{displaymath}
\begin{array}{rll}
k & \in & \mathcal{L} \\
S & \define & (\type{unit},k) \mid (\type{bool},k) \mid (\type{int},k) \mid (L(S),k) \mid S * S
\end{array}
\end{displaymath}
A security context $\seccontext$ is a partial mapping from variable
identifiers and the program counter $\type{pc}$ to security types. The context
assigns a type $(\type{unit},k)$ to $\type{pc}$ to track information
that may propagate through control flow as a result of branching statements.
The security type for lists contains a label $L(S)$ for the elements, as well
as a label $k$ for the list's length.

As in other information flow type systems, the partial order $k \sqsubseteq
k'$ indicates that the class $k'$ is at least as restrictive as $k$, i.e., $k$
is allowed to flow to $k'$. We assume a non-trivial security lattice that
contains at least two labels: $\llow$ (low security) and $\lhigh$ (high
security), with $\llow \sqsubseteq \lhigh$. Following the convention defined
in FlowCaml~\cite{flowCaml}, we also make use of a \emph{guard relation} $k
\triangleleft S$, which denotes that all of the labels appearing in $S$ are at
least as restrictive as $k$. The definition is given in Figure~\ref{fig:secsubtype}
along with its dual notion $S \blacktriangleleft k$, called the
\emph{collecting relation}, and the standard subtyping relation $S_1 \leq S_2$.

To refer to sets of variables by security class, we write
$\uppervars{\seccontext}{k}$ to denote the set of variable identifiers $x$ in
the domain of $\seccontext$ such that $\seccontext(x) \blacktriangleleft k$,
and define $\lowervars{\seccontext}{k}$ similarly. This gives us the set of
variables upper- and lower-bounded by $k$, respectively. Conversely, we define
$\hiddenvars{\seccontext}{k} = \{x \in \dom{\seccontext} : \seccontext(x)
\not\blacktriangleleft k\}$, the set of variables more restrictive than $k$.
To refer to the set of variables strictly bounded below by $k_1$ and above by
$k_2$, we write $\midvars{\seccontext}{k_1}{k_2}$. Given two well-formed
environments $\ev_1$ and $\ev_2$, we say that they are
\emph{k-equivalent} w.r.t $\seccontext$
if they agree on all variables with label at most $k$:
\begin{displaymath}
  \ev_1 \secesim{k} \ev_2 \Leftrightarrow
  \forall x \in \uppervars{\seccontext}{k}. \ev_1(x) = \ev_2(x)
\end{displaymath}
This relation captures the attacker's \emph{observational equivalence} between the two environments.
%
\begin{figure}[!th]
  \def \MathparLineskip {\lineskip=0.45cm}
\begin{mathpar}
\RuleNolabel
{
  k \sqsubseteq k'\\
        \!\!\!T \in \text{Atoms}
}
{
  k \triangleleft (T,k')
}
\hfill
\RuleNolabel
{
k \sqsubseteq k' \\
\!\!\! k \triangleleft S
}
{
k \triangleleft (L(S),k')
}
\hfill
\RuleNolabel
{
k \triangleleft S_1 \\
\!\!\! k \triangleleft S_2
}
{
k \triangleleft S_1 * S_2
}

\RuleNolabel
{
  k' \sqsubseteq k\\
        \!\!\!T \in \text{Atoms}
}
{
  (T,k') \blacktriangleleft k
}
\hfill
\RuleNolabel
{
k' \sqsubseteq k \\
\!\!\! S \blacktriangleleft k
}
{
(L(S),k') \blacktriangleleft k
}
\hfill
\RuleNolabel
{
S_1 \blacktriangleleft k \\
\!\!\! S_2 \blacktriangleleft k
}
{
S_1 * S_2 \blacktriangleleft k
}

\RuleNolabel
{
{k}{\sqsubseteq}{ k'}  \\
\!\!\!\! {T}{\in}{\text{Atoms}}
}
{
  (\type{T},k) \leq (\type{T},k')
}
\hfill
\RuleNolabel
{
{k}{\sqsubseteq}{ k'} \\
\!\!\!\! {S}{\leq}{S'}
}
{
{(L(S),k)}{\leq}{(L(S'),k')}
}
\hfill
\RuleNolabel
{
{S_1}{\leq}{ S_1'} \\
\!\!\!\! {S_2}{\leq}{S_2'}
}
{
{S_1}{ *}{ S_2}{\leq}{ S_1'}{*}{ S_2'}
}
\end{mathpar}
\caption{Guards, collecting security labels, and subtyping
  ($\text{Atoms} = \{\type{unit}, \type{int}, \type{bool}\}$)}
\label{fig:secsubtype}
\end{figure}
The first-order security types take the following form. 
The annotation $\type{pc}$ indicates the security level of the program
counter, i.e., a lower-bound on the label of any observer who is allowed to
learn that a given function has been invoked.
The $\type{const}$ annotation denotes that the function body respects
\emph{resource-aware noninterference}. 
\begin{displaymath}
\begin{array}{lcr}
\type{pc} \in \mathcal{L}  &&
F^s \define S_1 \xrightarrow{\type{pc}/\type{const}} S_2  \mid
S_1 \xrightarrow{\type{pc}} S_2
\end{array}
\end{displaymath}
A \emph{security
signature} $\sectstack: \fid \rightarrow \wp(\secfotypes) \setminus
\{\emptyset\}$ is a finite partial mapping from a set of function identifiers
to a \emph{non-empty sets} of first-order security types.

\subsection{Resource-aware noninterference}
We consider an adversary associated with label $k_1 \in \mathcal{L}$, who can
observe and control variables in $\uppervars{\seccontext}{k_1}$. Intuitively,
we say that a program $P$ satisfies resource-aware noninterference at level
$(k_1, k_2)$ w.r.t $\seccontext$, where $k_1 \sqsubseteq k_2$, if
\emph{1)} the behavior of $P$ does not leak any information about the contents
of variables more sensitive than $k_1$, and
\emph{2)} does not leak any information about the contents \emph{or sizes} of
variables more sensitive than $k_2$. The definition follows.

\vspace{-.1cm}
\begin{definition}
\label{def:secureprogram}

Let $\ev_1$ and $\ev_2$ be two well-formed environments and $\seccontext$ be a
security context sharing their domain. An expression $e$ satisfies
resource-aware noninterference at level $(k_1,k_2)$ for $k_1 \sqsubseteq k_2$,
if whenever $\ev_1$ and $\ev_2$ are:
\begin{enumerate}
\item observationally equivalent at $k_1$: $\ev_1 \secesim{k_1} \ev_2$,
\item size equivalent \wrt{} $\midvars{\seccontext}{k_1}{k_2}$: $\ev_1 \esim{\midvars{\seccontext}{k_1}{k_2}} \ev_2$
\end{enumerate}
then it follows from
$\ev_1 \rtyps{p_1}{p_1'} e \Downarrow v_1$ and $\ev_2 \rtyps{p_2}{p_2'} e \Downarrow v_2$ that
${v_1} = {v_2}$ and ${p_1}-{p_1'} = {p_2} - {p_2'}$.
\end{definition}
\vspace{-.1cm}

\noindent
The final condition in Defintion~\ref{def:secureprogram} ensures two properties. First, requiring that $v_1 = v_2$ provides noninterference~\cite{Goguen82}, given that $E_1$ and $E_2$ are observationally equivalent. Second, the requirement $p_1 - p_1' = p_2 - p_2'$ ensures that the program's resource consumption will remain constant w.r.t changes in variables from the set $\hiddenvars{\seccontext}{k_1}.$ This establishes noninterference w.r.t the program's final resource consumption, and thus prevents the leakage of secret information.

Before moving on, we point out an important subtlety in this definition. We require that all variables in $\midvars{\seccontext}{k_1}{k_2}$ begin with equivalent sizes, but not those in $\lowervars{\seccontext}{k_2}$. By fixing this quantity in the initial environments, we assume that an attacker is able to control and observe it, so it is not protected by the definition. This effectively establishes three classes of variables, i.e., those whose size and content are observable to the $k_1$-adversary, those whose size (but \emph{not} content) is observable, and those whose size and content remain secret. In the remainder of the text, we will simplify the technical development by assuming that the third and most-restrictive class is empty, and that all of the secret variables reside in $\midvars{\seccontext}{k_1}{k_2}.$

\subsection{Proving resource-aware noninterference}
There are two extreme ways of proving resource-aware noninterference.
Assume we already have established classic noninterference, the first way is to
additionally prove constant-resource usage \emph{globally} by
forgetting the security labels and showing that the program has
constant-resource usage. This is a sound approach but it requires us
to reason about parts of the programs that are not affected by secret
data. It would therefore result in the rejection of programs that have the
resource-aware noninterference property but are not constant resource.
The second way is to prove constant resource usage
\emph{locally} by ensuring that every conditional that branches
on secret values is constant time. However, this local
approach is problematic because it is not compositional. Consider the
following examples where \code{rev} is the reverse function.
\vspace{-.1cm}
\begin{lstlisting}[label=lst:pccond]
let f1(b,x) =
  let z = if b then x else [] in rev z

let f2(b,x,y) =
  let z = if b then let _ = rev y in x
          else let _ = rev x in y in rev z
\end{lstlisting}
\vspace{-.1cm}
If we assume a cost model in which we count the number of function
calls then the cost of \code{rev(x)} is $|x|$.
So \code{rev} is constant resource \wrt{} its argument. Moreover, the
expression \code{if b then x else []} is constant resource. However,
\code{f1} is not constant resource. In contrast, the conditional in 
\code{f2} is not constant resource. But \code{f2} is
a constant-resource function. The function \code{f2} can be
automatically analyzed with the constant-resource type system from
Section~\ref{sec:resourcetype} while \code{f1} is correctly rejected.

The idea of our type system for resource-aware noninterference is to
allow both global and local reasoning about resource consumption as
well as arbitrary intermediate levels. \emph{We ensure that every
  expression that is typed in a high security context is part of a
  constant-resource expression.}  In this way, we get the benefits
of local reasoning without losing compositionality.

\subsection{Typing rules and soundness}
We combine our type system for constant resource
usage with a standard information flow type system which based on
FlowCaml~\cite{Pottier02}. The interface between the two type systems
is relatively light and the idea is applicable to other cost-analysis methods 
as well as other security type systems.

In the type judgement, an expression is typed under a type context
$\seccontext$ and a label \type{pc}. The \type{pc} label can be
considered an upper bound on the security labels of all values that
affect the control flow of the expression and a lower bound on the
labels of the function's effects~\cite{Pottier02}.
As mentioned earlier, we will simplify the technical development by assuming
that the third and most-restrictive class is empty, 
that is, the typing rules here
guarantee that well-typed expressions provably satisfy the resource-aware
noninterference property \wrt{} changes in variables from the set $\hiddenvars{\seccontext}{k_1}$, say $X$.
We define two type judgements of the form
\begin{displaymath}
\begin{array}{lcl}
\type{pc};\sectstack;\seccontext \rtyps{\type{const}}{} e: S & \text{and} &
\type{pc};\sectstack;\seccontext \rtyps{}{} e: S \;.
\end{array}
\end{displaymath}
The judgement with the $\type{const}$
annotation states that under a security configuration given by
$\seccontext$
and the label $\type{pc}$,
$e$
has type $S$
and it satisfies resource-ware noninterference \wrt{} changes in variables from $X$.
The second judgement indicates that $e$
satisfies the noninterference property but does not make any
guarantees about resource-based side channels.
Selected typing rules are given in Fig. \ref{fig:securityrules}. We
implicitly assume that the security types and the resource-annotated
counterparts have the same base types.
We write $[\type{const}]$ to denote it is optional, and $\type{const}_{X}(e)$
if $e$ is
well-typed in the constant-resource type system \wrt{} $X$ (i.e., $\reststack;\rescontext \rtyps{q}{q'} e:A$, $\share(A \mid A,A)$, and $\forall x \in \dom{\rescontext} \setminus X. \share(\rescontext(x)
\mid \rescontext(x),\rescontext(x))$). We will discuss the
constant-resource type system in Section~\ref{sec:resourcetype}.

\begin{figure*}[!htb]
\begin{mathpar}
\small
%
\RuleToplabel{SR:Gen}
{
\type{pc};\sectstack;\seccontext \rtyps{\type{const}}{} e: S
}
{
\type{pc};\sectstack;\seccontext \vdash e: S
}

\RuleToplabel{SR:C-Gen}
{
\type{pc};\sectstack;\seccontext \vdash e: S \\
\type{const}_{X}(e)
}
{
\type{pc};\sectstack;\seccontext \rtyps{\type{const}}{} e: S
}

\RuleToplabel{SR:SubTyping}
{
\type{pc};\sectstack;\seccontext \rtyps{[\type{const}]}{} e: S \\
S \leq S'
}
{
\type{pc};\sectstack;\seccontext \rtyps{[\type{const}]}{} e: S'
}

\RuleToplabel{SR:Fun}
{
\sectstack(f) = S_1 \xrightarrow{\type{pc}'} S_2 \\\\
x: S_1 \in \seccontext \\
\type{pc} \sqsubseteq \type{pc}'
}
{
\type{pc};\sectstack;\seccontext \vdash \text{app}(f,x): S_2
}

\RuleToplabel{SR:L-Arg}
{
x: S_1 \in \seccontext \\
\sectstack(f) = S_1 \xrightarrow{\type{pc}'} S_2 \\\\
\type{pc} \sqsubseteq \type{pc}' \\
S_1 \blacktriangleleft k_1
}
{
\type{pc};\sectstack;\seccontext \rtyps{\type{const}}{} \text{app}(f,x): S_2
}

\RuleToplabel{SR:C-Fun}
{
\sectstack(f) = S_1 \xrightarrow{\type{pc}'/\type{const}} S_2 \\\\
x: S_1 \in \seccontext \\
\type{pc} \sqsubseteq \type{pc}'
}
{
\type{pc};\sectstack;\seccontext \rtyps{\type{const}}{} \text{app}(f,x): S_2
}

\RuleToplabel{SR:If}
{
x:(\type{bool},k_x) \in \seccontext \\
\type{pc} \sqcup k_x;\sectstack;\seccontext \vdash e_t: S  \\\\
\type{pc} \sqcup k_x;\sectstack;\seccontext \vdash e_f: S \\
\type{pc} \sqcup k_x \triangleleft S
}
{\type{pc};\sectstack;\seccontext \vdash \text{if}(x,e_t,e_f): S}

\RuleToplabel{SR:L-If}
{
\type{pc} \sqcup k_x;\sectstack;\seccontext \rtyps{\type{const}}{} e_t: S  \\
\type{pc} \sqcup k_x;\sectstack;\seccontext \rtyps{\type{const}}{} e_f: S \\\\
x:(\type{bool},k_x) \in \seccontext \\
\type{pc} \sqcup k_x \triangleleft S \\
k_x \sqsubseteq k_1
}
{\type{pc};\sectstack;\seccontext \rtyps{\type{const}}{} \text{if}(x,e_t,e_f): S}
%

\RuleToplabel{SR:Let}
{
\type{pc};\sectstack;\seccontext \vdash e_1: S_1 \\
\type{pc};\sectstack;\seccontext,x: S_1 \vdash e_2: S_2
}
{
\type{pc};\sectstack;\seccontext \vdash \text{let}(x,e_1,x.e_2): S_2
}

\RuleToplabel{SR:L-Let}
{
\type{pc};\sectstack;\seccontext \rtyps{\type{const}}{} e_1: S_1 \\
S_1 \blacktriangleleft k_1 \\\\
\type{pc};\sectstack;\seccontext,x: S_1 \rtyps{\type{const}}{} e_2: S_2
}
{
\type{pc};\sectstack;\seccontext \rtyps{\type{const}}{} \text{let}(x,e_1,x.e_2): S_2
}

\RuleToplabel{SR:C-Match-L}
{
x:(L(S),k_x) \in \seccontext \\
\type{pc} \sqcup k_x;\sectstack;\seccontext \rtyps{\type{const}}{} e_1: S_1 \\
k_x \sqsubseteq k_1 \\\\
\type{pc} \sqcup k_x;\sectstack;\seccontext,x_h: S,x_t:(L(S),k_x) \rtyps{\type{const}}{} e_2: S_1 \\
\type{pc} \sqcup k_x \triangleleft S_1
}
{
\type{pc};\sectstack;\seccontext \rtyps{\type{const}}{} \text{match}(x,e_1,(x_h,x_t).e_2): S_1
}
\end{mathpar}
\caption{Selected security typing rules}
\vspace{-.3cm}
\label{fig:securityrules}
\end{figure*}

Note that the standard
information flow typing rules~\cite{Heintze98,Pottier02} can be
obtained by removing the \type{const} annotations from all judgements.
Consider for instance the rule \textsc{SR:If} for conditional
expressions. By
executing the true or false branches, an adversary could gain
information about the conditional value whose security label is $k_x$.
Therefore, the
conditional expression must be type-checked under a security
assumption at least as restrictive as \type{pc} and $k_x$. This is a
standard requirement in any information flow type system.
In the following, we will focus on explaining how the rules restrict the
observable resource usage instead of these classic noninterference
aspects.

The most interesting rules are \textsc{SR:C-Gen} and the rules for 
\code{let} and \code{if} expressions, which block leakage over resource
usage when branching on high security data. \textsc{SR:C-Gen} allows
us to \emph{globally} reason about constant resource usage for an
arbitrary subexpression that has the noninterference property. For
example, we can apply \textsc{SR:If}, the standard 
rule for conditionals, first and then \textsc{SR:C-Gen} to prove that 
its super-expression is constant resource. Alternatively, we can use
rules such as \textsc{SR:L-If} and \textsc{SR:L-Let} to \emph{locally} reason
about resource use.
The rule \textsc{SR:L-Let} reflects the fact that if both $e_1$
and $e_2$
have the resource-aware noninterference property and the size of $x$
only depends on low security data then $\text{let}(x,e_1,x.e_2)$
respects resource-aware noninterference. The reasoning is
similar for \textsc{SR:L-If} where we require that the variable
$x$ does not depend on high security data.

Leaf expressions such as $\text{op}_{\diamond}(x_1,x_2)$
and $\text{cons}(x_h,x_t)$
have constant resource usage. Thus their judgements are always
associated with the qualifier $\type{const}$
as shown in the rule \textsc{SR:B-Op}.
The rule \textsc{SR:C-Fun} states that if a function's body has
the resource-aware noninterference property then the function application
has the resource-aware noninterference property too. If the argument's
label is low security data, bounded below by $k_1$, then the function application
has the resource-aware noninterference property since the value of the
argument is always the same under any $k$-equivalent environments.
It is reflected by rule \textsc{SR:L-Arg}.

%
\vspace{-.1cm}
\begin{example*}
Recall the functions \code{compare} and \code{p\_compare} in Fig.~\ref{fig:compare}.
Suppose the content of the first list is secret and the length is public. Thus it has type
$(L(\type{int},h),\ell)$. While the second list controlled by adversaries is public, hence it has type
$(L(\type{int},\ell),\ell)$. Assume that the \type{pc} label is $\ell$ and $\hiddenvars{\seccontext}{k_1}$ = $\hiddenvars{\seccontext}{\ell}$.
The return value's label depends on the content of the elements of the first list whose label is $h$.
Thus it must be assigned the label $h$ to make the functions well-typed.
\begin{displaymath}
\begin{array}{@{}r@{\;}l}
\word{compare}: & ((L(\type{int},h),\ell),(L(\type{int},\ell),\ell)) \xrightarrow{\ell} (\type{bool},h) \\
\word{p\_compare}: & ((L(\type{int},h),\ell),(L(\type{int},\ell),\ell)) \xrightarrow{\ell/\type{const}} (\type{bool},h)
\end{array}
\end{displaymath}
Here, both functions satisfy the noninterference property at security label $\ell$.
However, only \code{p\_compare} is a resource-aware noninterference function
\wrt{} $\hiddenvars{\seccontext}{\ell}$, or the secret list.
\end{example*}

\begin{example*}
Consider the following function \code{cond\_rev} in which \code{rev} is the standard reverse function.
\begin{lstlisting}
let cond_rev(l1,l2,b1,b2) = if b1 then
  let r = if b2 then rev l1; l2
  else rev l2; l1 in rev r; () else ()
\end{lstlisting}
Assume that $l_1$, $l_2$, $b_1$ and $b_2$ have types $(L(\type{int},h),\ell)$, $(L(\type{int},h),\ell)$, $(\type{bool},\ell)$, and
$(\type{bool},h)$, respectively. Given the \code{rev} function is constant \wrt{} the argument, the inner conditional
does not satisfy resource-aware noninterference.
However, the \code{let} expression satisfies resource-aware noninterference 
\wrt{} $\hiddenvars{\seccontext}{\ell}$ = $\{l_1,l_2,b_2\}$.
We can derive this by applying the rule \textsc{SR:C-Gen}.
By the rule \textsc{SR:L-If}, the outer conditional on low security data 
satisfies resource-aware noninterference \wrt{} $\{l_1,l_2,b_2\}$ at level $\ell$. We derive the following type.
\begin{displaymath}
\begin{array}{@{}r@{\;}l}
\word{cond\_rev}: & ((L(\type{int},h),\ell),(L(\type{int},h),\ell),(\type{bool},\ell),(\type{bool},h)) \\
						& \xrightarrow{\ell/\type{const}} (\type{unit},\ell)
\end{array}
\end{displaymath}
\end{example*}
\vspace{-.1cm}
We now prove the soundness of the type system \wrt{} the resource-aware noninterference property. 
It states that if $e$
is a well-typed expression with the \type{const} annotation then $e$ is a
resource-aware noninterference expression at level $k_1$.

The following two lemmas are needed in the soundness proof. The first lemma
states that the type system satisfies the standard \emph{simple
  security} property~\cite{Volpano96} and the second shows that the
type system prove classic noninterference.
\begin{lemma}
\label{lem:simplesecurity}
Let $\type{pc};\sectstack;\seccontext \vdash e: S$ or $\type{pc};\sectstack;\seccontext \rtyps{\type{const}}{} e: S$. For all variables $x$ in $e$, if $S \blacktriangleleft k_1$ then $\seccontext(x) \blacktriangleleft k_1$.
\end{lemma}
\begin{lemma}
\label{lem:noninterference}
Let $\type{pc};\sectstack;\seccontext \vdash e: S$ or $\type{pc};\sectstack;\seccontext \rtyps{\type{const}}{} e: S$, $\ev_1 \vdash e \Downarrow v_1$, $\ev_2 \vdash e \Downarrow v_2$, and $\ev_1 \secesim{k_1} \ev_2$. Then $v_1 = v_2$ if $S \blacktriangleleft k_1$.
\end{lemma}
\begin{theorem}
\label{theo:securitysoundness}
If $\models \ev: \seccontext$, $\ev \vdash e \Downarrow v$, and $\type{pc};\sectstack;\seccontext \rtyps{\type{const}}{} e: S$
then $e$ is a resource-aware noninterference expression at $k_1$.
\end{theorem}
\begin{proof}
The proof is done by induction on the structure of the typing derivation and the evaluation derivation. Let $X$ be the set of variables $\hiddenvars{\seccontext}{k_1}$. For all environments $\ev_1$, $\ev_2$ such that $\ev_1 \esim{X} \ev_2$ and $\ev_1 \secesim{k_1} \ev_2$, if $\ev_1 \rtyps{p_1}{p_1'} e \Downarrow v_1$ and $\ev_2 \rtyps{p_2}{p_2'} e \Downarrow v_2$. We then show that $p_1 - p_1' = p_2 - p_2'$ and $v_1 = v_2$ if $S \blacktriangleleft k_1$. We illustrate one case of the conditional expression. Suppose $e$ is of the form $\text{if}(x,e_t,e_f)$, thus the typing derivation ends with an application of either the rule \textsc{SR:L-If} or \textsc{SR:C-Gen}. By Lemma \ref{lem:noninterference}, if $S \blacktriangleleft k_1$ then $v_1 = v_2$.
\vspace{-.1cm}
\begin{itemize}
	\item Case \textsc{SR:L-If}. By the hypothesis we have $\ev_1(x) = \ev_2(x)$. Assume that $\ev_1(x) = \ev_2(x) = \word{true}$, by the evaluation rule \textsc{E:If-True}, $\ev_1 \rtyps{p_1-K^{\word{cond}}}{p_1'} e_t \Downarrow v_1$ and $\ev_2 \rtyps{p_2-K^{\word{cond}}}{p_2'} e_t \Downarrow v_2$. By induction for $e_t$ we have $p_1 - p_1' = p_2 - p_2'$. It is similar for $\ev_1(x) = \ev_2(x) = \word{false}$.
	\item Case \textsc{SR:C-Gen}. Since $\ev_1 \esim{X} \ev_2$ \wrt{} $\seccontext$, we have $\ev_1 \esim{X} \ev_2$ \wrt{} $\rescontext$. By the hypothesis we have $\type{const}_{X}(e)$. Thus by Theorem \ref{theo:constant}, it follows $p_1 - p_1' = p_2 - p_2'$.
\end{itemize}
\vspace{-.45cm}
\end{proof}

\section{Type systems for lower bounds and constant resource usage}
\label{sec:resourcetype}
We now discuss how to automatically and statically verify constant
resource usage, upper bounds, and lower bounds. For upper bounds we
rely on existing work on automatic amortized resource
analysis~\cite{Jost03,HoffmannW15}. This technique is based on an
affine type system. For constant resource usage and lower bounds we
introduce two new sub-structural resource-annotated type systems: The
type system for constant resource usage is linear and the one for
lower bounds is relevant.
\subsection{Background}
\paragraph{Amortized analysis}
To statically analyze a program with the potential method~\cite{tarjan85}, a mapping from program points to potentials must be established. One has to show that the potential at every program point suffices to cover the cost of any possible evaluation step and the 
potential of the next program point. The initial potential is then an upper bound on the resource usage of the program.

\paragraph{Linear potential for upper bounds}
To automate amortized analysis, we fix a format of the potential
functions and use LP solving to find the optimal coefficients.
To infer linear potential functions, inductive data types are
annotated with a non-negative rational numbers $q$~\cite{Jost03}. For
example, the type $L^q(\word{bool})$ of Boolean lists with potential
$q$ defines potential $q{\cdot}n$, where $n$ is the number of list's elements.

This idea is best explained by example. Consider the function
\code{filter\_succ} below that filters out positive numbers and
increments non-positive numbers. As in RAML, we use OCaml syntax and
\code{tick} commands to specify resource usage. If we filter out a
number then we have a high cost ($8$ resource units) since \code{x}
is, e.g., sent to an external device. If \code{x} is incremented we
have a lower cost of $3$ resource units. As a result, the worst-case
resource consumption of \code{filter\_succ($\ell$)} is $8|\ell| + 1$
(where $1$ is for the cost that occurs in the nil case of the match).
The function \code{fs\_twice($\ell$)} applies \code{filter\_succ}
twice, to $\ell$ and to the result of \code{filter\_succ($\ell$)}. The
worst-case behavior appears if no list element is filtered out in the
first call and all elements are filtered out in the second call. The
worst-case behavior is thus $11|\ell| + 2$.
\begin{figure}[!t]
\vspace{-.2cm}
\begin{lstlisting}
let rec filter_succ(l) =
match l with
| [] -> tick(1.0); []
| x::xs -> if x > 0 then
    tick(8.0); filter_succ(xs)
  else tick(3.0); (x+1)::filter_succ(xs)

let fs_twice(l) = 
  filter_succ(filter_succ(l))
\end{lstlisting}
\vspace{-.2cm}
  \caption{Two OCaml functions with linear resource usage. The
    \emph{worst-case} number of \emph{ticks} executed by
    \code{fitler\_succ($\ell$)} and \code{fs\_twice($\ell$)} is
    $8|\ell|+1$ and $11|\ell|+2$ respectively. In the \emph{best-case}
    the functions execute $3|\ell|+1$ and $6|\ell|+2$ ticks,
    respectively. The resource consumption is not constant.}
\label{fig:filter}
\vspace{-.5cm}
\end{figure}
These upper bounds can be expressed with the following annotated
function types, which can be derived using local type rules in Fig.~\ref{fig:constantrules}.
\begin{displaymath}
\begin{array}{rl}
\word{filter\_succ}: & L^{8}(\type{int}) \xrightarrow{1/0} L^{0}(\type{int)} \\
\word{fs\_twice}: & L^{11}(\type{int}) \xrightarrow{2/0} L^{0}(\type{int)}
\end{array}
\end{displaymath}
Intuitively, the first function type states that an initial potential
of $8|\ell| + 1$ is sufficient to cover the cost of
\code{filter\_succ($\ell$)} and there is $0|\ell'| + 0$ potential left
where $\ell'$ is the result of the computation. This is just one possible
potential annotation of many. The right choice of the potential
annotation depends on the use of the function result. For example, for
the inner call of \code{filter\_succ} in \code{fs\_twice} we need the
following annotation.
\begin{displaymath}
\begin{array}{rl}
\word{filter\_succ}: & L^{11}(\type{int}) \xrightarrow{2/1} L^{8}(\type{int)}
\end{array}
\end{displaymath}
It states that the initial potential of $11|\ell| + 2$ is sufficient
to cover the cost of \code{filter\_succ($\ell$)} and there is
$8|\ell'| + 1$ potential left to be assigned to the returned list
$\ell'$. The potential of the result can then be used with the
previous type of \code{filter\_succ} to pay for the cost of the outer
call.
\begin{displaymath}
\begin{array}{rll}
\word{filter\_succ}: & L^{p}(\type{int}) \xrightarrow{q/q'} L^{r}(\type{int)}
\mid&
q \geq q' {+} 1 \, \land  \,
p \geq 8 \, \\ && \land \,
p \geq 3 {+} r
\end{array}
\end{displaymath}
We can summarize all possible types of \code{filter\_succ} with a
linear constraint system. In the type inference, we generate such
a constraint system and solve it with an off-the-shelf LP solver. To obtain tight bounds, we perform a
whole-program analysis and minimize the coefficients in the input
potential.

Surprisingly, this approach---as well as the new concepts we introduce
here---can be extended to polynomial bounds~\cite{HoffmannAH12},
higher-order functions~\cite{Jost10,HoffmannW15}, polymorphism~\cite{Jost09}, and
user-defined inductive types~\cite{Jost09,HoffmannW15}.

\subsection{Resource annotations}

The resource-annotated types are base types in which the inductive
data types are annotated with non-negative rational numbers, called
\emph{resource annotations}. 
\begin{displaymath}
\begin{array}{ll}
A \define \type{unit} \mid \type{bool} \mid \type{int} \mid L^p(A) \mid A * A & (\text{for } p \in \mathbb{Q}^{+}_{0})
\end{array}
\end{displaymath}
A type context,
$\rescontext: \vid \rightarrow \datatypes$,
is a partial mapping from variable identifiers to resource-annotated
types. The \emph{underlying} base type and context denoted
by $\atype{A}$,
and $\atype{\rescontext}$
respectively can be obtained by removing the annotations.  We extend
all definitions such as $\sizeval{v}$,
$\models \ev:\context$
and $\esim{}$
for base data types to resource-annotated data types by ignoring the
annotations.

We now formally define the notation of \emph{potential}. The potential of a value $v$ of type $A$, written $\poten{}(v:A)$, is defined by the function $\poten{}: \val \rightarrow \mathbb{Q}^+_0$ as follows.
\begin{displaymath}
\begin{array}{l}
\poten{}(\word{()}:\type{unit}) = \poten{}(b:\type{bool}) = \poten{}(n:\type{int}) = 0 \\
\poten{}((v_1,v_2):A_1 * A_2) = \poten{}(v_1:A_1) + \poten{}(v_2:A_2) \\
\poten{}([v_1, \cdots ,v_n]:L^{p}(A)) = n {\cdot} p + \Sigma^{n}_{i = 1}\poten{}(v_i:A)
\end{array}
\end{displaymath}
\begin{example*}
The potential of a list $v = [b_1,\cdots,b_n]$
of type $L^{p}(\type{bool})$
is $n{\cdot}p$.
Similarly, a list of lists of Booleans $v = [v_1,\cdots,v_n]$
of type $L^{p}(L^{q}(\type{bool}))$,
where $v_i = [b_{i1}, \cdots ,b_{im_i}]$,
has the potential $n{\cdot}p + (m_1 + \cdots + m_n) {\cdot}q$.
\end{example*}

Let $\rescontext$
be a context and $\ev$
be a well-formed environment \wrt{} $\rescontext$.
The potential of $X \subseteq \dom{\rescontext}$
under $\ev$
is defined as
$\poten{\ev}(X:\rescontext) = \Sigma_{x \in
  X}\poten{}(\ev(x):\rescontext(x))$. The potential of
$\rescontext$
is $\poten{\ev}(\rescontext) = \poten{\ev}(\dom{\rescontext}:\rescontext)$.
Note that if $x \not \in X$
then $\poten{\ev}(X:\rescontext) = \poten{\ev[x \mapsto v]}(X:\rescontext)$.
The following lemma states that the potential is the same
under two well-formed size-equivalent environments.
\begin{lemma}
\label{lem:environmentpotential}
If $\ev_1 \esim{X} \ev_2$ then $\Phi_{\ev_1}(X:\rescontext) = \Phi_{\ev_2}(X:\rescontext)$.
\end{lemma}

Annotated first-order data types are given as follows, where $q$ and $q'$ are rational numbers.
\begin{displaymath}
	F ::= A_1 \xrightarrow{q/q'} A_2
\end{displaymath}

A resource-annotated signature
$\reststack: \fid \rightarrow \wp(\fotypes) \setminus \{\emptyset\}$
is a partial mapping from function identifiers to a non-empty sets of
annotated first-order types. That means a function can have different
resource annotations depending on the context. The \emph{underlying}
base types are denoted by
$\atype{F}.$ 
And the underlying base signature is denoted by $\atype{\reststack}$
where $\atype{\reststack}(f) = \atype{\reststack(f)}$.

\begin{figure*}[!htb]
\begin{mathpar}
\small
\RuleToplabel{A:Var}
{
}
{
	\reststack;x: A \rtyps{K^{\word{var}}}{0} x: A
}
\and
\RuleToplabel{A:B-Op}
{
	\diamond \in \{\word{and},\word{or}\}
}
{
	\reststack;x_1:\type{bool},x_2:\type{bool} \rtyps{K^{\word{op}}}{0} \text{op}_{\diamond}(x_1,x_2):\type{bool}
}
\and
\RuleToplabel{A:Fun}
{
	\reststack(f) = A_1 \xrightarrow{q/q'} A_2
}
{
	\reststack;x: A_1 \rtyps{q {+} K^{\word{app}}}{q'} \text{app}(f,x): A_2
}
\and
\RuleToplabel{A:Let}
{
	\!\!\! \reststack;\rescontext_1 \rtyps{q {-} K^{\word{let}}}{q'_1} e_1: A_1 \\
	\reststack;\rescontext_2,x: A_1 \rtyps{q'_1}{q'} e_2: A_2
}
{
	\!\! \reststack;\rescontext_1,\rescontext_2 \rtyps{q}{q'} \text{let}(x,e_1,x.e_2): A_2
}
\and
\RuleToplabel{A:If}
{
	\reststack;\rescontext \rtyps{q - K^{\word{cond}}}{q'} e_t: A \!\! \\
	\reststack;\rescontext \rtyps{q - K^{\word{cond}}}{q'} e_f: A \!\!\!
}
{
	\reststack;\rescontext,x:\type{bool} \rtyps{q}{q'} \text{if}(x,e_t,e_f): A
}
\and
\RuleToplabel{A:Match-L}
{
	\reststack;\rescontext \rtyps{q - K^{\word{matchN}}}{q'} e_1: A_1 \\
	\reststack;\rescontext,x_h:A,x_t:L^{p}(A) \rtyps{q + p - K^{\word{matchL}}}{q'} e_2:A_1
}
{
	\reststack;\rescontext,x:L^{p}(A) \rtyps{q}{q'} \text{match}(x,e_1,(x_h,x_t).e_2):A_1
}
\and
\RuleToplabel{A:Cons}
{
}
{
	\reststack;x_h:A,x_t:L^{p}(A) \rtyps{p + K^{\word{cons}}}{0} \text{cons}(x_h,x_t):L^{p}(A)
}
\and
\RuleToplabel{A:Share}
{
	\reststack;\rescontext,x_1:A_1,x_2:A_2 \rtyps{q}{q'} e:B \\
	\share(A \mid A_1,A_2)
}
{
	\reststack;\rescontext,x:A \rtyps{q}{q'} \text{share}(x,(x_1,x_2).e):B
}
\end{mathpar}
\caption{Selected syntax-directed rules of the resource type systems. They are shared among all type systems.}
\vspace{-.1cm}
\label{fig:constantrules}
\end{figure*}
\subsection{Type system for constant resource consumption}
The typing rules of the constant-resource type system define judgements
of the form:
\begin{displaymath}
	\reststack;\rescontext \rtyps{q}{q'} e: A
\end{displaymath}
where $e$
is an expression and $q, q' \in \mathbb{Q}^+_0$.
The intended meaning is that in the environment $\ev$,
$q + \poten{\ev}(\rescontext)$
resource units are sufficient to evaluate $e$
to a value $v$
with type $A$
and there are \emph{exactly} $q' + \poten{}(v: A)$
resource units left over.

The typing rules form a \emph{linear} type system. It ensures that
every variable is used exactly once by allowing exchange but
not weakening or contraction~\cite{Walker02}. The rules can
be organized into syntax directed and structural rules.
\paragraph{Syntax-directed rules} The syntax-directed rules are shared
among all type systems and selected rules are listed in
Fig.~\ref{fig:constantrules}. Rules like \textsc{A:Var} and
\textsc{A:B-Op} for leaf expressions (e.g., variable, binary
operations, pairs) have fixed costs as specified by the constants
$K^x$.
Note that we require all available potential to be spent.
The cost of the function call is represented by the constant
$K^{\word{app}}$ in the rule \textsc{A:Fun} and the argument
carries the potential to pay for the function execution.
%
In the rule \textsc{A:Let}, the cost of binding is represented by the
constant $K^{\word{let}}$.
The potentials carried by the contexts $\rescontext_1$
and $\rescontext_2$
are passed sequentially through the sub derivations. Note that the
contexts are disjoint since our type system is linear. Multiple uses
of variables must be introduced through the rule \textsc{A:Share}. Thus, the 
context split is deterministic.
The rule \textsc{A:If} is the key rule for ensuring
constant resource usage. By using the same context $\rescontext$
for typing both $e_t$
and $e_f$,
we ensure that the conditional expression has the same resource usage
in size-equivalent environments independent of the value of the
Boolean variable $x$.
The rules for inductive data types are crucial for the interaction of
the linear potential annotations with the constant potential, in which 
\textsc{A:Cons} shows how constant potential can be associated with a
new data structure. The dual is \textsc{A:Match-L}, which
shows how potential associated with data can be released. It is
important that these transitions are made in a linear fashion:
potential is neither lost or gained.
\paragraph{Sharing relation}
{\def \MathparLineskip {\lineskip=0.1cm}
\small
\begin{mathpar}
	\RuleNolabel
	{
		A \in \{\type{unit},\type{bool},\type{int}\}
	}
	{
		\share(A \mid A,A)
	}

  \RuleNolabel
  {
    \share(A \mid A_1,A_2) \\
    \share(B \mid B_1,B_2)
  }
  {
    \share(A * B \mid A_1 * B_1,A_2 * B_2)
  }

	\RuleNolabel
	{
		\share(A \mid A_1,A_2) \\
		p = p_1 + p_2
	}
	{
		\share(L^{p}(A) \mid L^{p_1}(A_1),L^{p_2}(A_2))
	}
\end{mathpar}}
The \emph{share expression} makes multiple uses of a variable
explicit. While multiple uses of a variable seem to be in conflict
with the linear type discipline, the \emph{sharing relation}
$\share(A \mid A_1,A_2)$
ensures that potential is treated in a linear way. It apportions
potential to ensure that the total potential associated with all uses
is equal to the potential initially associated with the variable. This
relation is only defined for structurally-identical types which differ
in at most the resource annotations.

\paragraph{Structural rules}
To allow more programs to be typed we add two structural rules to the
type system which can be applied to every expression. These rules are
specific to the the constant-resource type system.
\vspace{-.1cm}
\begin{mathpar}
\small
\RuleToplabel{C:Weakening}
{
	\reststack;\rescontext \rtyps{q}{q'} e:B \\
	\share(A \mid A,A)
}
{
	\reststack;\rescontext,x:A \rtyps{q}{q'} e:B
}
\hfill
\RuleToplabel{C:Relax}
{
  \reststack;\rescontext \rtyps{p}{p'} e:A \\\\
  q \geq p \\
  q - p = q' - p'
}
{
	\reststack;\rescontext \rtyps{q}{q'} e:A
}
\end{mathpar}
\vspace{-.1cm}

The rule \textsc{C:Relax} reflects the fact that if it is sufficient to
evaluate $e$
with $p$
available resource units and there are $p'$
resource units left over then $e$ can be evaluated with $p + c$
resource units and there are exactly $p' + c$
resource units left over, where $c \in \mathbb{Q}^+_0$.
Rule \textsc{C:Weakening} states that an extra variable can be added
into the given context if its potential is zero. The condition is
enforced by $\share(A \mid A,A)$
since $\poten{}(v:A) = \poten{}(v:A) + \poten{}(v:A)$
or $\poten{}(v:A) = 0$.
These rules can be used in branchings such as the conditional or the
pattern match to ensure that subexpressions are typed using the same
contexts and potential annotations.
\vspace{-.1cm}
\begin{example*}
  Consider again the function \code{p\_compare} in Fig.~\ref{fig:pcompare} in which
  the nil case of the second matching on $l$ is padded with \code{tick(5.0); aux(false,xs,[])}
  and the resource consumption is defined using \emph{tick} annotations. The
  resource usage of \code{p\_compare($h,\ell$)} is constant \wrt{} $h$, that is, it is exactly
  $5|h| + 1$. This is reflected by
  the following type.
\begin{displaymath}
\begin{array}{rl}
\word{p\_compare}: & (L^{5}(\type{int}),L^{0}(\type{int})) \xrightarrow{1/0} \type{bool}
\end{array}
\end{displaymath}
It can be understood as follows. If the input list $h$ carries $5$
potential units per element then it is sufficient to cover the cost
of \code{p\_compare($h,\ell$)}, no potential is wasted, and $0$ potential is left.
\end{example*}

\paragraph{Soundness}
That soundness theorem states that if $e$
is well-typed in the resource type system and evaluates to a value
$v$
then the difference between the initial and the final potential is the
net resources usage. Moreover, if the potential annotations of the
return value and all variables not belonging to a set
$X \subseteq \dom{\rescontext}$ are zero then $e$ is constant-resource \wrt{} $X$.
\vspace{-.1cm}
\begin{theorem}
\label{theo:soundness}
If $\models \ev : \rescontext$, $\ev \vdash e \Downarrow v$, and $\reststack;\rescontext \rtyps{q}{q'} e:A$, then for all $p, r \in \mathbb{Q}^+_0$ such that $p = q + \poten{\ev}(\rescontext) + r$, there exists $p' \in \mathbb{Q}^+_0$ satisfying $\ev \rtyps{p}{p'} e \Downarrow v$ and $p' = q' + \poten{}(v:A) + r$.
\end{theorem}
\vspace{-.2cm}
\begin{proof}
  The proof proceeds by a nested induction on the derivation of
  the evaluation judgement and the typing judgement, in which the
  derivation of the evaluation judgement takes priority over the typing
  derivation. We need to induct on both, evaluation
  and typing derivation. An induction on only the
  typing derivation would fail for the case of function application,
  which increases the size of the typing derivation, while the
  size of the evaluation derivation does not increase. An induction on only the
  evaluation judgement would fail because of structural rules such as \textsc{C:Weakening}.
  If such a rule is the final step in the
  derivation then the size of typing derivation decreases while
  the length of evaluation derivation is unchanged. The additional
  constant $r$
  is needed to make the induction case for the let rule
  work. \vspace{-.2cm}



\end{proof}
\begin{theorem}
\label{theo:constant}
If $\models \ev:\rescontext$,
$\ev \vdash e \Downarrow v$,
$\reststack;\rescontext \rtyps{q}{q'} e:A$,
$\share(A \mid A,A)$,
and
$\forall x \in \dom{\rescontext} \setminus X. \share(\rescontext(x)
\mid \rescontext(x),\rescontext(x))$ then $e$
is constant resource \wrt{} $X \subseteq \dom{\rescontext}$.
\end{theorem}
\subsection{Type system for upper bounds}
If we treat potential as an \emph{affine} resource then we arrive at
the original amortized analysis for upper bounds~\cite{Jost03}. To
this end, we allow unrestricted weakening and a relax rule in which we
can waste potential.
\begin{mathpar}
\small
\RuleToplabel{U:Relax}
{
	\reststack;\rescontext \rtyps{p}{p'} e:A \\
	q \geq p \\\\
	q - p \geq q' - p'
}
{
	\reststack;\rescontext \rtyps{q}{q'} e:A
}

\RuleToplabel{U:Weakening}
{
	\reststack;\rescontext \rtyps{q}{q'} e:B
}
{
	\reststack;\rescontext,x:A \rtyps{q}{q'} e:B
}
\end{mathpar}
Additionally, we can use subtyping to waste linear
potential~\cite{Jost03}. (See the dual definition for subtyping for
lower bounds below.) Similarly to Theorem~\ref{theo:soundness}, we can
prove the following theorem.

\vspace{-.1cm}
\begin{theorem}
\label{theo:soundness_upper}
If $\models \ev : \rescontext$, $\ev \vdash e \Downarrow v$, and $\reststack;\rescontext \rtyps{q}{q'} e:A$, then for all $p, r \in \mathbb{Q}^+_0$ such that $p \geq q + \poten{\ev}(\rescontext) + r$, there exists $p' \in \mathbb{Q}^+_0$ satisfying $\ev \rtyps{p}{p'} e \Downarrow v$ and $p' \geq q' + \poten{}(v:A) + r$.
\end{theorem}
\vspace{-.2cm}
\subsection{Type system for lower bounds}
The type judgements for lower bounds have the same form and data types
as the type judgements for constant resource usage and upper bounds.
However, the intended meaning of the judgement
$\reststack;\rescontext \rtyps{q}{q'} e:A$ is the following. Under
given environment $\ev$, less than $q + \poten{\ev}(\context)$
resource units are not sufficient to evaluate $e$ to a value $v$ so
that more than $q' + \poten{}(v: A)$ resource units are left over.

The syntax-directed typing rules are the same as the rules in the
constant-resource type system as given in
Fig.~\ref{fig:constantrules}. In addition, we have the structural
rules in Fig.~\ref{fig:lowerboundrules}. The rule \textsc{L:Relax} is
dual to \textsc{U:Relax}. In \textsc{L:Relax}, potential is treated as
a \emph{relevant} resource: We are not allowed to waste potential but
we can create potential out of the blue if we ensure that we either
use it or pass it to the result. The same idea is formalized for the
linear potential with the sub-typing rules \textsc{L:Subtype} and
\textsc{L:Supertype}. The sub-typing relation is defined as follows.
\begin{mathpar}
\small
\RuleNolabel
{
	A  {\in} \{\type{unit},\type{bool},\type{int}\}
}
{
	A \subtype A
}
\hfill
\RuleNolabel
{
	A_1 {\subtype} A_2 \!\!\! \\
	p_1 {\leq} p_2
}
{
	L^{p_1}\!(A_1) \subtype L^{p_2}\!(A_2)
}
\hfill
\RuleNolabel
{
	A_1 {\subtype} A_2 \!\!\!\!\!\! \\
	B_1 {\subtype} B_2 \!
}
{
	A_1 {*} A_2 \subtype B_1 {*} B_2
}
\end{mathpar}
It holds that if $A \subtype B$ then $\atype{A} = \atype{B}$ and
$\poten{}(v:A) \leq \poten{}(v:B)$.
Suppose that it is not sufficient to evaluate $e$ with $p$ available
resource units to get $p'$ resource units left over.
\textsc{L:Subtype} reflects the fact that we also cannot evaluate $e$
with $p$ resources get more than $p'$ resource units after the
evaluation. \textsc{L:Supertype} says that we also cannot evaluate $e$
with less than $p$ and get $p'$ resource units
afterwards. 
\vspace{-.1cm}
\begin{example*}
  Consider again the functions \code{filter\_succ} and
  \code{fs\_twice} given in Fig.~\ref{fig:filter} in which the
  resource consumption is defined using \emph{tick} annotations.  The
  best-case resource usage of \code{filter\_succ($\ell$)} is
  $3|\ell| + 1$ and best-case resource usage of
  \code{fs\_twice($\ell$)} is $6|\ell| + 2$. This is reflected by
  the following function types for lower bounds.
\begin{displaymath}
\begin{array}{rl}
\word{filter\_succ}: & L^{3}(\type{int}) \xrightarrow{1/0} L^{0}(\type{int)} \\
\word{fs\_twice}: & L^{6}(\type{int}) \xrightarrow{2/0} L^{0}(\type{int)}
\end{array}
\end{displaymath}
To derive the lower bound for \code{fs\_twice}, we need the same
compositional reasoning as for the derivation of the upper bound.
For the inner call of \code{filter\_succ} we use the type
\begin{displaymath}
\begin{array}{rl}
\word{filter\_succ}: & L^{6}(\type{int}) \xrightarrow{2/1} L^{3}(\type{int)} \; .
\end{array}
\end{displaymath}
It can be understood as follows. If the input list carries $6$
potential units per element then, for each element, we can either use
all $6$ (\code{if} case) or we can use $3$ and assign $3$ to
the output (\code{else} case).
\end{example*}
\vspace{-.1cm}
The type system for lower bounds is a \emph{relevant} type
system~\cite{Walker02}. That means every variable is used at least
once by allowing \emph{exchange} and \emph{contraction properties},
but not \emph{weakening}. However, as in the constant-time type system we allow a restricted form of weakening if the potential
annotations are zero using the rule \textsc{L:Weakening}. The
following lemma states formally the contraction property which is
derived in Fig.~\ref{fig:contraction}.
\begin{lemma}
If $\reststack;\rescontext,x_1:A,x_2:A \rtyps{q}{q'} e:B$ then $\reststack;\rescontext,x:A \rtyps{q}{q'} \text{share}(x,(x_1,x_2).e):B$
\end{lemma}
\begin{figure}[!t]
\begin{mathpar}
\small
\RuleToplabel{L:Relax}
{
	\reststack;\rescontext \rtyps{p}{p'} e:A \\\\
	q \geq p\!\!\! \\
	q {-} p \leq q' {-} p'
}
{
	\reststack;\rescontext \rtyps{q}{q'} e:A
}
\and
\RuleToplabel{L:Weakening}
{
	\reststack;\rescontext \rtyps{q}{q'} e:B \\
	\share(A \mid A,A)
}
{
	\reststack;\rescontext,x:A \rtyps{q}{q'} e:B
}
\and
\RuleToplabel{L:Subtype}
{
	\reststack;\rescontext \! \rtyps{q}{q'} \! e{:}A \!\!\!\!\\
	A {\subtype} B
}
{
	\reststack;\rescontext \rtyps{q}{q'} e:B
}
\hfill
\RuleToplabel{L:Supertype}
{
	\reststack;\rescontext,x{:}B \! \rtyps{q}{q'} e:C \!\!\!\! \\
	A {\subtype} B
}
{
	\reststack;\rescontext,x:A \rtyps{q}{q'} e:C
}
\end{mathpar}
\caption{Structural rules for lower bounds.}
\label{fig:lowerboundrules}
\end{figure}

\begin{figure}
\begin{mathpar}
\small
\vspace{-.5cm}
\raisebox{-.5cm}[0cm][0cm]{\textsc{\hspace{6.0cm}\small{(L:Contr)}}}\\
\mbox{}
\hspace{-1cm}
\RuleNolabel
{
  \RuleNolabel
  {
    \RuleNolabel[rightskip=1.2cm]
    {
      \reststack;\rescontext,x_1:A,x_2:A \rtyps{q}{q'} e:B
      \\  A_2 \subtype A
    }
    {
      \reststack;\rescontext,x_1:A,x_2:A_2 \rtyps{q}{q'} e:B \\
      A_1 \subtype A
    }
  }
  {
    \reststack;\rescontext,x_1:A_1,x_2:A_2 \rtyps{q}{q'} e:B
  }
  \share(A \mid A_1,A_2)
}
{
  \reststack;\rescontext,x:A \rtyps{q}{q'} \text{share}(x,(x_1,x_2).e):B
}
\end{mathpar}
\caption{Derivation of the contraction rule for lower-bounds.}
\label{fig:contraction}
\end{figure}

The following theorems establish the soundness of the analysis. The
proofs can be found in the TR~\cite{sidechanTR}.
Theorem~\ref{theo:lowerbound} is proved by induction and
Theorem~\ref{theo:lowersoundness} follows by contradiction.
\vspace{-.1cm}
\begin{theorem}
\label{theo:lowersoundness}
Let $\models \ev :\rescontext$, $\ev \vdash e \Downarrow v$, and
$\reststack;\rescontext \rtyps{q}{q'} e:A$. Then for all
$p, r \in \mathbb{Q}^+_0$ such that
$p < q + \poten{\ev}(\rescontext) + r$, there exists no
$p' \in \mathbb{Q}^+_0$ satisfying $\ev \rtyps{p}{p'} e \Downarrow v$
and $p' \geq q' + \poten{}(v:A) + r$.
\end{theorem}
\vspace{-.3cm}
\begin{theorem}
\label{theo:lowerbound}
Let $\models \ev :\rescontext$, $\ev \vdash e \Downarrow v$, and
$\reststack;\rescontext \rtyps{q}{q'} e:A$. Then for all
$p, p' \in \mathbb{Q}^+_0$ such that
$\ev \rtyps{p}{p'} e \Downarrow v$ we have
$q + \poten{\ev}(\rescontext) - (q' + \poten{}(v:A)) \leq p - p'$.
\end{theorem}

\subsection{Mechanization}
We mechanized the soundness proofs for both the two new type systems
as well as the classic AARA type system using the proof assistant
Agda. The development is roughly 4000 lines of code, which includes
the inference rules, the operational cost semantics,
a proof of type preservation, and the soundness
theorems for each type system. 


One notable difference is our implementation of the typing contexts. In Agda our
contexts are implemented as lists of pairs of variables and their types. Moreover,
in our typing rules whenever a variable is added to the context we require the variable is fresh with respect to the existing context. This
requirement is important as it allows us to preserve the invariant that the
context is well formed with respect to the environment as we induct over typing
and evaluation judgements in our soundness proofs. Furthermore, as our typing
contexts are ordered lists we added an \emph{exchange} rule to our typing
rules.

Another important detail is in the implementation of potential. Potential
$\poten{}(v: A)$ for a value only is defined for well formed
inputs. Inputs such as $\poten{}(\word{nil}: \type{bool})$ are not defined. Agda
is total language and as such prohibits users from implementing partial functions.
Thus we require in our Agda implementation that when
calculating the potential of a value of a given type the user provide a derivation
that the value is well formed with respect to that type. Similarly when calculating
the potential of a context, $\poten{\ev}(\rescontext)$, with respect to an
environment we require that the user provide a derivation that the context is
well formed with respect to that environment. 

Lastly, whereas the type systems and proofs presented here used positive
rational numbers, in the Agda implementation we use natural numbers. This
deviation was simply due to the lacking support for rationals in the Agda
standard library. By replacing a number of trivial lemmas, mostly related to
associativity and commutativity, the proofs and embeddings could be transformed to
use rational numbers instead.



\section{Quantifying and transforming out leakages}
\label{sec:trans_quantify}
We present techniques to quantify the amount of information leakage
through resource usage and transform leaky programs into constant
resource programs. The quantification relies on the lower and upper
bounds inferred by our resource type systems. The transformation pads
the programs with dummy computations so that the evaluations consume
the same amount of resource usage and the outputs are identical with
the original programs. In the current implementation, these dummy
computations are added into programs by users and the padding
parameters are automatically added by our analyzer to obtain the
optimal values. It would be straightforward to make the process fully
automatic but the interactive flavor of our approach helps to get
a better understanding of the system.

\subsection{Quantification}

Recall from Section~\ref{sec:securitytype} that we assume an adversary at
level $k_1$ who is always able to observe \emph{1)} the values of variables in
$\uppervars{\seccontext}{k_1}$, and \emph{2)} the final resource consumption
of the program. For many programs, it may be the case that changes to the
secret variables $\hiddenvars{\seccontext}{k_1}$ effect observable differences
in the program's final resource consumption, but only allow the attacker to
learn partial information about the corresponding secrets. In this section, we
show that the upper and lower-bound information provided by our type systems
allow us to derive bounds on the amount of partial information that is leaked.

To quantify the amount of leaked information, we measure the number of
distinct environments that the attacker could deduce as having produced a
given resource consumption observation. However, because there may be an
unbounded number of such environments, we parameterize this quantity on the
size of the values contained in each environment. Let $\envs{\sizes}$ denote
the space of environments with values of size characterized by $\sizes$.
Given an environment
$\ev$ and expression $e$, define $\usage(\ev,e) = p_\delta$ such that $\ev
\rtyps{p}{p'} e \Downarrow v$ and $p_\delta = p - p'$. Then for an expression
$e$ and resource observation $p_\delta$, we define the set
$\runcertain_{\sizes}(e, p_\delta)$ which captures the attacker's uncertainty
about the environment which produced $p_\delta$:
\[
\runcertain_{\sizes}(e, p_\delta) =
\{
\ev' \in \envs{\sizes} : \usage(\ev,e) = p_\delta
\}
\]
Notice that when $|\runcertain_{\sizes}(e, p)| = 1$, the attacker can deduce
exactly which environment was used, whereas when this quantity is large little
additional information is learned from $p_\delta$. This gives us a natural
definition of leakage, which is obtained by aggregating the inverse of the
cardinality of $\runcertain_{\sizes}$ over the possible initial environments
of $e$:
\[
\qleakage_{\sizes}(e) =
\left(
\sum_{\ev \in \envs{\sizes}} \frac{1}{|\runcertain_{\sizes}(e, \usage(\ev,e))|}
\right)
-1
\]
$\qleakage_{\sizes}(e)$ corresponds to our intuition about leakage. When $e$
leaks no information through resource consumption, then each term in the
summation will be $1/|\envs{sizes}|$ giving $\qleakage_{\sizes}(e) = 0$,
whereas if $e$ leaks perfect information about its starting environment then
each term will be $1$, leading to $\qleakage_{\sizes}(e) = |\envs{\sizes}| -
1$.
\vspace{-.1cm}
\begin{theorem}
\label{thm:leakage}
Let $P^e_{\sizes}$ be the complete set of resource observations producible by
expression $e$ under environments of size $\sizes$, i.e.,
\[
P^e_{\sizes} = \{p : \exists \ev \in \envs{\sizes} . \usage(\ev,e) = p \}
\]
Then $|P^e_{\sizes}| = \qleakage_{\sizes}(e) + 1$.
\end{theorem}
\vspace{-.1cm}
\begin{lemma}
\label{lem:leakagebound}
Let $\lbound_e(\sizes)$ and $\ubound_e(\sizes)$ be lower and upper-bounds on
the resource consumption of $e$ for inputs of size $\sizes$. If $\usage(\ev,e)
\in \mathbb{Z}$ for all environments $\ev$, then $\qleakage_{\sizes}(e) \le
\ubound_e(\sizes) - \lbound_e(\sizes)$.
\end{lemma}
\vspace{-.1cm}
\begin{lemma}
\label{lem:leakageentropy}
Assume that environments are sampled uniform-randomly from $\envs{\sizes}$.
Then the Shannon entropy of $P^e_{\sizes}$ is given by $\qleakage_{\sizes}(e)$:
$H(P^e_{\sizes}) \le \log_2(\qleakage_{\sizes}(e)+1)$

\end{lemma}
\vspace{-.1cm}
\noindent
Lemma~\ref{lem:leakagebound} leverages Theorem~\ref{thm:leakage} to derive an
upper-bound on leakage from upper and lower resource bounds. This
result only holds when the resource observations of $e$ are integral,
which ensures the interval $[\lbound_e(\sizes), \ubound_e(\sizes)]
\supseteq P^e_{\sizes}$ is finite. Lemma~\ref{lem:leakageentropy} relates
$\qleakage_{\sizes}(e)$ to Shannon entropy, which is commonly used to
characterize information leakage~\cite{Zhang12,KopfMO12,KopfB2007}.

\subsection{Transformation}

To transform programs into constant resource programs we extend the type
system for constant resource use from Section~\ref{sec:resourcetype}. Recall
that the type system treats potential in a linear fashion to ensure that
potential is not wasted. We will now add \emph{sinks} for potential which will
be able to absorb excess potential. At runtime the sinks will consume the
exact amount of resources that have been statically-absorbed to ensure that
potential is still treated in a linear way. The advantage of this approach is
that the worst-case resource consumption is often not affected by the
transformation. Additionally, we do not need to keep track of resource usage
at runtime to pad the resource usage at the sinks, because the amount of
resource that must be discarded is statically-determined by the type system.
Finally, we automatically obtain a type derivation that serves as a proof that
the transformation is constant resource.

More precisely, the sinks are represented by the syntactic form:
$\consume{(\type{A},p)}{x}$. Here, $\type{A}$ is a resource-annotated type and
$p \in \Qplus$ is a non-negative rational number. The idea is that $\type{A}$
and $p$ define the resource consumption of the expression.  In the
implementation, the user only has to write $\consume{}{x}$, and the
annotations are added via automatic syntax elaboration during the resource
type inference.

Let $\ev$ be a well-formed environment \wrt{} $\rescontext$. For every $x \in
\dom{\context}$ with $\rescontext(x) = A$, the expression
$\consume{(\type{A},p)}{x}$ consumes $\poten{}(\ev(x):A) + p$ resource units
and evaluates to $\word{()}$. The evaluation and typing rules for sinks are:
\begin{mathpar}
%
\RuleToplabel{E:Consume}
{
   q = q' + \poten{}(\ev(x):A) + p
}
{
  \ev \rtyps{q}{q'} \consume{(\type{A},p)}{x} \Downarrow ()
}
\and
\RuleToplabel{A:Consume}
{
}
{
  \reststack;x{:} A \rtyps{p}{0} \consume{(A,p)}{x}: \type{unit}
}
\hfill
%
\end{mathpar}
The extension of the proof of Theorem~\ref{theo:soundness} to consume
expressions is straightforward.

\paragraph{Adding consume expressions}
Let $e_i$
be a subexpression of $e$
and let $e_i'$
be the expression
$\text{let}(z,\consume{}{x_1,\cdots,x_n},z.e_i)$
for some variables $x_i$.
Let $e'$ be the expression
obtained from $e$
by replacing $e_i$
with $e_i'$. We write $e \transform e'$ for such a transformation.
Note that additional \code{share} and \code{let} expressions have to be added
to convert $e_i'$ into share-let normal form.
\vspace{-.1cm}
\begin{lemma}
\label{lem:transformpreservation}
If $\tstack;\context \vdash e: T$, $\ev \vdash e \Downarrow v$, and $e \transform e'$ then $\tstack;\context \vdash e': T$ and $\ev \vdash e' \Downarrow v$.
\end{lemma}
\vspace{-.1cm}
To transform an expression $e$
into a constant resource expressions we perform multiple
transformations $e \transform e'$
which do not affect the type and semantics of $e$.
This can be done automatically but in our implementation it works in
an interactive fashion, meaning that users are responsible for the
locations where consume expressions are put. The analyzer will infer
the annotations $A$
and constants $p$
of the given consume expressions during type inference. If the
inference is successful then we have $\type{const}_{X}(e')$ for the transformed program $e'$.

\vspace{-.1cm}
\begin{example*}
  Recall the function \code{compare} form Fig.~\ref{fig:compare}. To
  turn \code{compare} into a constant resource function. We insert
  \code{consume} expressions as shown below. Users can insert many
  \code{consume} expressions and the analyzer will determine which
  \code{consumes} are actually needed.
\vspace{-.1cm}
\begin{lstlisting}
let rec c_compare(h,l) = match h with
| [] -> match l with
  | [] -> tick(1.0); true
  | y::ys -> tick(1.0); false
| x::xs -> match l with
  | [] -> tick(1.0); consume(xs); false
  | y::ys -> if (x = y) then
      tick(5.0); c_compare(xs,ys)
    else tick(5.0); consume(xs); false
\end{lstlisting}
\vspace{-.1cm}
  We automatically obtain the following typing of the transformed
  function and the consume expressions:
\begin{displaymath}
\begin{array}{rll}
\word{c\_compare} : & (L^{5}(\type{int}),L^{0}(\type{int})) \xrightarrow{1/0} \type{bool} & \\
\word{consume} : & L^{5}(\type{int}) \xrightarrow{5/0} \type{unit} & (\text{at line } 6) \\
\word{consume} : & L^{5}(\type{int}) \xrightarrow{1/0} \type{unit} & (\text{at line } 9) \\
\end{array}
\end{displaymath}
The worst-case resource consumption of the unmodified function \code{c\_compare} is
$1 + 5|h|$.
Thus the consumption of the first \code{consume} must be
$5 + 5(|h| - 1 - |\ell|)$
when $h$
is longer than $l$.
Otherwise, the consumption is zero. The second one consumes
$1 + 5(|h_1| - 1)$,
where $h_1$
is the sub-list of $h$
from the first node which is different from the corresponding node in
$l$.
\end{example*}

\begin{table*}
\begin{center}
{\footnotesize
\begin{tabular*}{\textwidth}{@{\extracolsep{\fill}} l l l l l}
\hline
\noalign{\vskip 1mm}
Constant Function	&	LOC & Metric & Resource Usage & Time\\
\noalign{\vskip 1mm}
\hline
\noalign{\vskip 1mm}
$\word{cond\_rev}: (L(\type{int}),L(\type{int}),\type{bool}) \rightarrow \type{unit}$
& 20
& steps
& $13n {+} 13x {+} 35$
& 0.03s\\
$\word{trunc\_rev}: (L(\type{int}),\type{int}) \rightarrow L(\type{int})$
& 28
& function calls
& $1n$
& 0.06s\\
$\word{ipquery}: L(\type{logline}) \rightarrow (L(\type{int}),L(\type{int}))$
& 86
& steps
& $86n {+} 99$
& 0.86s\\
$\word{kmeans}: L(\type{float},\type{float}) \rightarrow L(\type{float},\type{float})$
& 170
& steps
& $1246n {+} 3784$
& 8.18s\\
$\word{tea\_enc}: (L(\type{int}),L(\type{int}),\type{nat}) \rightarrow L(\type{int})$
& 306
& ticks
& $128n^2z {+} 32nxz {+} 1184nz {+} 96n {+} 128z {+} 96$
& 13.73s\\
$\word{tea\_dec}: (L(\type{int}),L(\type{int}),\type{nat}) \rightarrow L(\type{int})$
& 306
& ticks
& $128n^2z {+} 32nxz {+} 1184nz {+} 96n {+} 96z {+} 96$
& 14.34s\\
\end{tabular*}}
\end{center}
\vspace{-1mm}
\end{table*}
\begin{table*}
\begin{center}
{\footnotesize
\begin{tabular*}{\textwidth}{@{\extracolsep{\fill}} l l l l l l l}
\hline
\noalign{\vskip 1mm}
Function	&	LOC	&  Metric & 	Lower Bound & 	Time & Upper Bound & Time \\
\noalign{\vskip 1mm}
\hline
\noalign{\vskip 1mm}
$\word{compare}: (L(\type{int}),L(\type{int})) \rightarrow \type{bool}$
& 60
& steps
& $7$
& 0.05s
& $16n {+} 7$
& 0.09s \\
$\word{find}: (L(\type{int}),\type{int}) \rightarrow \type{bool}$
& 40
& steps
& $5$
& 0.04s
& $14n {+} 5$
& 0.02s \\
$\word{rsa}: (L(\type{bool}),\type{int},\type{int}) \rightarrow \type{int}$
& 42
& multiplications
& $1n$
& 0.07s
& $2n$
& 0.05s\\
$\word{filter}: L(\type{int}) \rightarrow L(\type{int})$
& 30
& steps
& $13n {+} 5$
& 0.05s
& $20n {+} 5$
& 0.04s\\
$\word{isortlist}: L(L(\type{int})) \rightarrow L(L(\type{int}))$
& 60
& steps
& $21n {+} 5$
& 0.13s
& $12n^2 {+} 9n {+} 10n^2m {-} 10nm {+} 5$
& 0.43s\\
$\word{bfs\_tree}: (\type{btree},\type{int}) \rightarrow \type{btree option}$
& 116
& steps
& $15$
& 0.30s
& $92n {+} 24$
& 0.32s\\
\end{tabular*}}
\vspace*{-1.25em}
\end{center}
\caption{Automatically-Derived Bounds with Resource Aware ML}
\vspace{-.4cm}
\label{tab:evalation}
\end{table*}
\section{Implementation and Evaluation}
\label{sec:evaluation}
\paragraph{Type Inference}
Type inference for the constant resource and lower bound systems are
implemented in RAML~\cite{HoffmannW15}. RAML is integrated into Inria's OCaml
compiler and supports polynomial bounds, user-defined inductive types,
higher-order functions, polymorphism, arrays, and references. All features are
implemented for the new type systems, as they are straightforward extensions
of the simplified rules presented in this paper. The implementation is
publicly available in an easy-to-use web interface~\cite{ramlWebAnon}.

The type inference is technically similar to the inference of upper
bounds~\cite{Jost03}. We first integrate the structural rules of the
respective type system in the syntax directed rules. For example,
weakening and relaxation is applied at branching points such as
conditional. We then compute a type derivation in
which all resource annotations are replace by (yet unknown) variables.
For each type rule we produce a set of linear constraints that specify
the properties of valid annotations. These linear constraints are then
solved by the LP solver CLP to obtain a type derivation in which the
annotations are rational numbers.

An interesting challenge lies in finding a solution for the linear
constraints that leads to the best bound for a given function. For
upper bounds, we simply disregard the potential of the result type and
provide an objective function that minimizes the annotations of the
arguments. The same strategy works the constant-time type
systems. An interesting property is that the solution to the linear
program is unique if we require that the potential of the result type
is zero. To obtain the optimal lower bound we want to maximize the
potential of the arguments and minimize the potential of the result.
We currently simply maximize the potential of the arguments while
requiring the potential of the result to be zero. Another approach
would be to first minimize the output potential and then maximize the
input potential.

\paragraph{Resource-aware noninterference}
We are currently integrating our constant-time type system with
FlowCaml~\cite{flowCaml}. The combined inference is based on the
typing rules in Fig.~\ref{fig:securityrules}. It is possible to derive
a set of type inference rules in the same way as for
FlowCaml~\cite{Sulzmann99,Pottier02}.
%
%
One of the challenges in the integration is interfacing FlowCaml's type
inference with our constant-time type system in rule
\textsc{SR:C-Gen}. In the implementation, we intend for each application
of \textsc{SR:C-Gen} to generate an intermediate representation of the
expression in RAML for the expression under consideration, in which
all types are annotated with fresh resource annotations along with the
set of variables $X$.  The expression is marked with the qualifier
$\type{const}$ if RAML can prove that it is constant time. The type
inference algorithm always tries to apply the syntax-directed rules
first before using \textsc{SR:C-Gen}.

\paragraph{Evaluation}
Table~\ref{tab:evalation} shows the verification and computation of
constant resource usage, lower, and upper bounds for different
functions, together with the lines of code (LOC) of the analyzed
function and the run time of the analysis in seconds. Note that lower
and upper bounds are identical when a function is constant. In the
computed bounds, $n$ is the size of the first argument,
$m = \text{max}_{1 \leq i \leq n} m_i$ where $m_i$ are the sizes of
the first argument's elements, $x$ is the size of the second argument,
and $z$ is the value of the third argument.

The cost models are specified by different cost metrics that are
appropriate for the respective application, e.g., number of evaluation
steps or number of multiplication operations. Note that the computed
upper bounds are also the resource usages of functions which are
padded using \code{consume} expressions. The experiments were run on a
machine with Intel Core i5 2.4 GHz processor and 8GB RAM under OS
X 10.11.5. The run time of the analysis varies from 0.02 to 14.34
seconds depending on the function's code complexity.  The example
programs that we analyzed consist of commonly-used primitives
(\word{cond\_rev}, \word{trunc\_rev}, \word{compare}, \word{find},
\word{filter}), functions related to cryptography (\word{tea\_enc},
\word{tea\_dec}, \word{rsa}), and examples taken from Haeberlan et
al.~\cite{HaeberlenPA11} (\word{ipquery}, \word{kmeans}). The full
source code of the examples can be found in the technical
report~\cite{sidechanTR}.

The encryption functions \word{tea\_enc} and \word{tea\_dec} correspond to the
encryption and decryption routines of the Corrected Block Tiny Encryption
Algorithm~\cite{Yarrkov10}, a block cipher presented by Needham and Wheeler in
an unpublished technical report in 1998. Our implementation correctly
identifies these operations as constant-time in the number of primitive
operations performed. We applied this cost model for the \word{tea} examples
due to the presence of bitwise operations in the original algorithm, which are
not currently supported in RAML. In order to derive a more meaningful bound,
we implemented bitwise operations in the example source and counted them as
single operations.

The two examples taken from Haeberlen et al.~\cite{HaeberlenPA11} were
originally created in a study of timing attacks in differentially-private data
processing systems. \word{ipquery} applies pattern matching to a database
derived from Apache server logs, counting the number of matches and
non-matches. \word{kmeans} implements the k-means clustering
algorithm~\cite{MacQueen67}, which partitions a set of geometric points into
$k$ clusters that minimize the total inter-cluster distance between points.
Haeberlen et al. demonstrated that when a query applied to a dataset
introduces attacker-observable timing variations, then the privacy guarantees
provided by differential privacy are negated. To address this, they proposed a
mitigation approach that enforces constant-time behavior by aborting or
padding the query's runtime. Our implementation is able to determine that the
queries, as we implemented them, were constant-time to begin with, and thus
did not need black-box mitigation.

\section{Related work}
\label{sec:relatedwork}
\paragraph{Resource bounds}
Our work builds on past research on automatic amortized resource
analysis (AARA). AARA has been introduced by Hofmann and Jost for a
strict first-order functional language with built-in data types to
derive linear heap-memory bounds~\cite{Jost03}. It has then been
extended to polynomial
bounds~\cite{HoffmannH10,HoffmannAH12,HoffmannS13,HofmannM14,HofmannM15} for strict and
higher-order~\cite{Jost10,HoffmannW15} functions.
AARA has also been used to derive linear bounds for lazy functional
programs~\cite{SimoesVFJH12,VasconcelosJFH15} and object-oriented
programs~\cite{Jost06,HofmannR13}. In another line of work, the
technique has been integrated into separation logic~\cite{Atkey10} to
derive bounds that depend on mutable data-structures, and into Hoare
logic to derive linear bounds that depend on
integers~\cite{veristack14,CarbonneauxHZ15}.
Amortized analysis has also been used to
manually verify the complexity of algorithms and data-structures using
proof assistants~\cite{Nipkow15,ChargueraudP15}.
In contrast to
our work, these techniques can only derive upper bounds and cannot prove
constant resource consumption.

The focus on upper bounds is shared with automatic resource analysis
techniques that are based on sized
types~\cite{Vasconcelos03,Vasconcelos08}, linear dependent
types~\cite{LagoG11,LagoP13}, and other type
systems~\cite{Crary00,Danielsson08,CicekGA15}. Similarly,
semiautomatic
analyses~\cite{Grobauer01,Benzinger04,DannerLR15,AvanziniLM15} focus
on upper bounds too.

Automatic resource bound analysis is also actively studied for
imperative languages using recurrence
relations~\cite{AlonsoG12,FloresH14,AlbertFR15} and abstract
interpretation~\cite{GulwaniMC09,BlancHHK10,Zuleger11,SinnZV14,CernyHKRZ15}.
While these techniques focus on worst-case bounds, it is possible to
use similar techniques for deriving lower
bounds~\cite{AlbertGM12}. The advantage of our method is that it is
compositional, deals well with amortization effects, and works for
language features such as user-defined data types and higher-order
functions. Another approach to (worst-case) bound analysis is based on
techniques from term
rewriting~\cite{AvanziniM13,NoschinskiEG13,BrockschmidtEFFG14}, which
mainly focus on upper bounds. One line of work~\cite{FrohnNHBG16} derives lower
bounds on the \emph{worst-case} behavior of programs which is
different from our lower bounds on the best-case behavior.

\paragraph{Side channels}
Analyzing and mitigating potential sources of side channel leakage is an
increasingly well-studied area. Several groups have proposed using type
systems or other program analyses to transform programs into constant-time
versions by padding 
with ``dummy''
commands~\cite{Agat00,HedinS05,CoppensVBK09,Zhang12,BartheRW06,MolnarPSW}. Because these
systems do not account for timing explicitly, as is the case for our work, this
approach will in nearly all cases introduce an unnecessary performance
penalty. The most recent system by Zhang et al.~\cite{Zhang12}
describes an approach for mitigating side channels using a combination of
security types, hardware assistance, and predictive
mitigation~\cite{ZhangAM11}. Unlike the type system in
Section~\ref{sec:securitytype}, they do not guarantee that information is
not leaked through timing. Rather, they show that the amount of this leakage
is bounded by the variation of the mitigation commands.

K\"opf and Basin~\cite{KopfB2007} presented an information-theoretic model for
adaptive side-channel attacks that occur over multiple runs of a program, and an automated analysis for measuring the corresponding leakage. Because
their analysis is doubly-exponential in the number of steps taken by the
attacker, they describe an approximate version based on a greedy heuristic.
Mardziel et al.~\cite{MardzielAHC14} later generalized this model to probabilistic
systems, secrets that change over time, and wait-adaptive
adversaries. Pasareanu et al.~\cite{PasareanuPM16} proposed a symbolic
approach for the multi-run setting based on MaxSAT and model counting. Doychev
et al.~\cite{DoychevFKMR13} and K\"opf et al.~\cite{KopfMO12} consider cache
side channels, and present analyses that over-approximate leakage using
model-counting techniques. While these analyses are sometimes able to derive
useful bounds on the leakage produced by binaries on real hardware, they do not
incorporate security labels to distinguish between different sources, and were
not applied to verifying constant-time behavior.

FlowTracker~\cite{RodriguesQFA16} and ct-verif~\cite{AlmeidaBBDE16} are both
constant-time analyses built on top of LLVM which reason about timing and
other side-channel behavior indirectly through control and address-dependence
on secret inputs. VirtualCert~\cite{BartheBCLP14} instruments CompCert with a
constant-time analysis based on similar reasoning about control and
address-dependence. These approaches are intended for code that has been
written in ``constant-time style'', and thus impose effective restrictions on
the expressiveness of the programs that they will work on. Because our
approach reasons about resources explicitly, it imposes no a priori
restrictions on program expressiveness.

\paragraph{Information flow}
A long line of prior work looks at preventing undesired information flows using type
systems. Sabelfeld and Myers~\cite{Sabelfeld03} present an excellent overview
of much of the early work in this area. The work most closely related to our
security type system is FlowCaml~\cite{Pottier02}, which provides a type
system that enforces noninterference for a core of ML with references,
exceptions, and let-polymorphism. The portion of our type system that applies
to traditional noninterference coincides with the rules used in FlowCaml. 
However, the rules in our type system are not only designed to track flows 
of information, but they are also used to incorporate the information flow 
and resource usage behavior such as the rules \textsc{SR:L-If} and 
\textsc{SR:L-Let}. Moreover, our type system constructs a flexible interface between 
FlowCaml and the resource type system, which means the rules can be easily adapted to integrate into 
any information-flow type system.   

The primary difference between our work and the prior work on information-flow
type systems is best summarized in terms of our attacker model. Whereas prior
work assumes an attacker that can manipulate low-security inputs and observe
low-security outputs, our type system enhances this attacker by granting the
ability to observe the program's final resource consumption. This broadens the
relevant class of attacks to include resource side channels, which we prevent
by extending a traditional information flow type system with explicit
reasoning about the resource behavior of the program. 


\section{Discussion}
\label{sec:discussion}
The definition of resource-aware noninterference given in
Definition~\ref{def:secureprogram} assumes an adversary whose observations of
resource consumption match the cost semantics with respect to the cost model
given in Section~\ref{sec:resourcetype}. Depending on how the costs are
parameterized, this may not match the actual resource use in a physical
environment on modern hardware. Architectural features such as caching and
variable-duration instructions need to be accounted for in the cost semantics,
or the guarantees might not hold in
practice~\cite{Brumley05,GullaschBK11,Osvik06}. Moreover, additional artifacts
of the compilation process can affect the constant-resource guarantees
established by the type system. Certain optimization passes and garbage
collectors might affect timing properties in ways that lead to
vulnerabilities if not accounted for by the cost semantics.


The cost semantics used in this work is conceptually straightforward, and
corresponds to the resource model encapsulated by the high-level programming
language. Accordingly, our verifier is oblivious to the machine instructions
and operand values that are eventually executed after the high-level code is
compiled. In particular, the fact that our cost model effectively counts the
total number of language primitives that are executed, and not the
corresponding processor instructions with caching and other micro-architectural
effects accounted for, means that compiled programs may not satisfy resource-aware noninterference in practice despite being provable within our type system.


Although architectural timing channels are nominally invisible at the
source-language level, it may be possible to incorporate these aspects into
the cost semantics with specific assumptions about the target platform and
compiler toolchain. Doing so with a high degree of precision is challenging,
as the semantics may need to track extensive state to accurately reflect the
timing behavior of the underlying platform. Another approach is to incorporate
dependence on these features indirectly, as in Zhang et al.~\cite{Zhang12}
where security labels are associated with hardware states to track information
flow dependencies throughout the hardware environment. This approach is
compatible with our resource-aware noninterference type system, but is less
flexible for the programmer as it is subject to the same types of imprecision
present in information-flow type systems. We leave as future work developing
more precise models that remain faithful to the resource-consumption
subtleties of hardware platforms.

Another limitation of this work follows from the imprecision of the
information-flow type system that is integrated with our constant-resource
type system to verify resource-aware noninterference. It is well-known that
such type systems are more conservative than the semantics of allowed noninterference~\cite{Myers98,Myers00,Li05}, and this applies to our work as well. In
particular, a variable conservatively identified as high-security could influence resource usage, leading our
verifier to conclude that a program which is constant-resource in practice is
not. 
Our approach mitigates this issue since imprecise information-flow
tracking does not directly lead to rejections of secure programs but
only increases the burden on constant-resource analysis. 
Another potential mitigation that applies in some cases is to simply prove
that the program is constant-resource with respect to all variables. Another
approach that we leave to future work is to incorporate declassification
mechanisms into our system.



\section{Conclusion}
\label{sec:conclusion}
We have introduced new sub-structural type systems for automatically
deriving lower bounds and proving constant resource usage. The
evaluation with the implementation in RAML shows that the technique
extends beyond the core language that we study in this paper and works
for realistic example programs. We have shown how the new type systems
can interact with information-flow type systems to prove
resource-aware noninterference. Moreover, the type system for constant
resource can be used to automatically remove side-channel
vulnerabilities from programs.

There are many interesting connections between security and (automatic)
quantitative resource analysis that we plan to study in the
future. Two concrete projects that we already started are the
integration of the type systems for upper and lower bounds with
information-flow type systems to precisely quantify the resource-based
information leakage at certain security levels. Another direction is
to more precisely characterize the amount of information that can be
obtained about secrets by making one particular resource-usage
observation.

\section*{Acknowledgments}

This article is based on research that has been supported, in part, by AFRL
under DARPA STAC award FA8750-15-C-0082 and DARPA Brandeis award
FA8750-15-2-028, by NSF under grant 1319671 (VeriQ), and by a Google Research
Award.
Any opinions, findings, and conclusions contained in this document are
those of the authors and do not necessarily reflect the views of the
sponsoring organizations.





%


\newpage
\clearpage
\appendix
\section*{The language semantics}
The equivalent expressions in OCaml syntax of the language are given as follows.
\begin{displaymath}
\begin{array}{rl}
e & ::= \word{()} \mid \word{true} \mid \word{false} \mid n \mid x\\
  & \mid x_1 \diamond x_2\\
  & \mid f(x)\\
  & \mid \text{let } x = e_1 \text{ in } e_2\\
  & \mid \text{if } x \text{ then } e_t \text{ else } e_f\\
  & \mid (x_1,x_2)\\
  & \mid \text{match } x \text{ with } (x_1,x_2) \rightarrow e\\
  & \mid []\\
  & \mid x_1::x_2\\
  & \mid \text{match } x \text{ with } \mid [] \rightarrow e_1 \mid x_1::x_2 \rightarrow e_2\\
  & \mid \text{share } x = (x_1,x_2) \text{ in } e \\
\end{array}
\end{displaymath}
\begin{displaymath}
\diamond \in \{+,-,*,\word{div},\word{mod},=,<>,>,<,<=,>=,\word{and},\word{or} \}
\end{displaymath}
Fig.~\ref{fig:valuetypes} and Fig.~\ref{fig:bosrules} represent the typing rules for values, the base typing and the evaluation rules for the language, respectively.
\begin{figure*}[!ht]
\begin{mathpar}
\small
\RuleToplabel{V:Unit}
{
v = \word{()}
}
{
\models v:\type{unit}
}

\RuleToplabel{V:Bool}
{
v \in \{\word{true},\word{false}\}
}
{
\models v:\type{bool}
}

\RuleToplabel{V:Int}
{
v \in \mathbb{Z}
}
{
\models v:\type{int}
}

\RuleToplabel{V:Pair}
{
\models v_1:T_1 \\
\models v_2:T_2
}
{
\models (v_1,v_2):T_1 * T_2
}

\RuleToplabel{V:Nil}
{
v = \word{nil}
}
{
\models v:L(T)
}

\RuleToplabel{V:List}
{
\models v_i:T \\
\forall i = 1,...,n
}
{
\models [v_1,...,v_n]:L(T)
}

\RuleToplabel{T:Unit}
{
}
{
\tstack;\emptyset \vdash ():\type{unit}
}

\RuleToplabel{T:Bool}
{
b \in \{\word{true},\word{false}\}
}
{
\tstack;\emptyset \vdash b:\type{bool}
}

\RuleToplabel{T:Int}
{
n \in \mathbb{Z}
}
{
\tstack;\emptyset \vdash n:\type{int}
}

\RuleToplabel{T:Var}
{
x \in \dom{\ev}
}
{
\tstack;x: T \vdash x: T
}

\RuleToplabel{T:B-Op}
{
\diamond \in \{\word{and},\word{or}\}
}
{
\tstack;x_1:\type{bool},x_2:\type{bool} \vdash \text{op}_{\diamond}(x_1,x_2):\type{bool}
}

\RuleToplabel{T:IB-Op}
{
\diamond \in \{=,<>,>,<,<=,>=\}
}
{
\tstack;x_1:\type{int},x_2:\type{int} \vdash \text{op}_{\diamond}(x_1,x_2):\type{bool}
}

\RuleToplabel{T:I-Op}
{
\diamond \in \{+,-,*,\word{div},\word{mod}\}
}
{
\tstack;x_1:\type{int},x_2:\type{int} \vdash \text{op}_{\diamond}(x_1,x_2):\type{int}
}

\RuleToplabel{T:Fun}
{
\tstack(g) = T_1 \rightarrow T_2
}
{
\tstack;x: T_1 \vdash \text{app}(g,x): T_2
}

\RuleToplabel{T:Let}
{
\tstack;\context_1 \vdash e_1: T_1 \\
\tstack;\context_2,x: T_1 \vdash e_2: T_2
}
{
\tstack;\context_1,\context_2 \vdash \text{let}(x,e_1,x.e_2): T_2
}

\RuleToplabel{T:If}
{
\tstack;\context \vdash e_t: T \\
\tstack;\context \vdash e_f: T
}
{
\tstack;\context,x:\type{bool} \vdash\text{if}(x,e_t,e_f): T
}

\RuleToplabel{T:Pair}
{
}
{
\tstack;x_1:T_1,x_2:T_2 \vdash \text{pair}(x_1,x_2):T_1 * T_2
}

\RuleToplabel{T:Match-P}
{
\tstack;\context,x_1:T_1,x_2:T_2 \vdash e: T
}
{
\tstack;\context,x:T_1 * T_2 \vdash \text{match}(x,(x_1,x_2).e): T
}

\RuleToplabel{T:Nil}
{
T \in \bdatatypes
}
{
\tstack;\emptyset \vdash \word{nil}:L(T)
}

\RuleToplabel{T:Cons}
{
}
{
\tstack;x_h:T,x_t:L(T) \vdash \text{cons}(x_h,x_t):L(T)
}

\RuleToplabel{T:Match-L}
{
\tstack;\context \vdash e_1: T_1 \\
\tstack;\context,x_h:T,x_t:L(T) \vdash e_2: T_1
}
{
\tstack;\context,x:L(T) \vdash \text{match}(x,e_1,(x_h,x_t).e_2): T_1
}

\RuleToplabel{T:Share}
{
\tstack;\context,x_1:T,x_2:T \vdash e: T_1
}
{
\tstack;\context,x:T \vdash \text{share}(x,(x_1,x_2).e): T_1
}

\RuleToplabel{T:Weakening}
{
\tstack;\context \vdash e: T_1
}
{
\tstack;\context,x:T \vdash e:T_1
}
\end{mathpar}
\caption{Typing rules for values and base types}
\label{fig:valuetypes}
\end{figure*}

\begin{figure*}[!ht]
\begin{mathpar}
\small
\RuleToplabel{E:Unit}
{
}
{
\ev \rtyps{q + K^{\word{unit}}}{q} \word{()} \Downarrow \word{()}
}

\RuleToplabel{E:Bool}
{
b \in \{\word{true},\word{false}\}
}
{
\ev \rtyps{q + K^\text{\word{bool}}}{q} b \Downarrow b
}

\RuleToplabel{E:Int}
{
n \in \mathbb{Z}
}
{
\ev \rtyps{q + K^{\word{int}}}{q} n \Downarrow n
}

\RuleToplabel{E:Var}
{
x \in \dom{\ev}
}
{
\ev \rtyps{q + K^{\word{var}}}{q} x \Downarrow \ev(x)
}
\hfill
\RuleToplabel{E:Bin}
{
v = \ev(x_1) \diamond \ev(x_2)
}
{
\ev \rtyps{q + K^{\word{op}}}{q}  \text{op}_{\diamond}(x_1,x_2) \Downarrow v
}

\RuleToplabel{E:Fun}
{
\tstack(g) = T_1 \rightarrow T_2 \\
E[x^g \mapsto \ev(x)] \rtyps{q}{q'} e_g \Downarrow v
}
{
\ev \rtyps{q+K^{\word{app}}}{q'} \text{app}(g,x) \Downarrow v
}

\RuleToplabel{E:Match-N}
{
\ev(x) = \word{nil} \\
\ev \rtyps{q - K^{\word{matchN}}}{q'} e_1 \Downarrow v
}
{
\ev \rtyps{q}{q'} \text{match}(x,e_1,(x_h,x_t).e_2) \Downarrow v
}

\RuleToplabel{E:Nil}
{
}
{
\ev \rtyps{q + K^{\word{nil}}}{q} \word{nil} \Downarrow \word{nil}
}

\RuleToplabel{E:If-True}
{
\ev(x) = \word{true} \\
\ev \rtyps{q - K^{\word{cond}}}{q'} e_t \Downarrow v
}
{
\ev \rtyps{q}{q'} \text{if}(x,e_t,e_f) \Downarrow v
}

\RuleToplabel{E:If-False}
{
\ev(x) = \word{false} \\
\ev \rtyps{q - K^{\word{cond}}}{q'} e_f \Downarrow v
}
{
\ev \rtyps{q}{q'} \text{if}(x,e_t,e_f) \Downarrow v
}

\RuleToplabel{E:Pair}
{
x_1,x_2 \in \dom{\ev} \\
v = (\ev(x_1),\ev(x_2))
}
{
\ev \rtyps{q + K^{\word{pair}}}{q} \text{pair}(x_1,x_2) \Downarrow v
}

\RuleToplabel{E:Match-P}
{
\ev(x) = (v_1,v_2) \\
\ev[x_1 \mapsto v_1,x_2 \mapsto v_2] \rtyps{q - K^{\word{matchP}}}{q'} e \Downarrow v
}
{
\ev \rtyps{q}{q'} \text{match}(x,(x_1,x_2).e) \Downarrow v
}

\RuleToplabel{E:Cons}
{
x_h,x_t \in \dom{\ev} \\
\ev(x_h) = v_1 \\
\ev(x_t) = [v_2,..,v_n]
}
{
\ev \rtyps{q + K^{\word{cons}}}{q} \text{cons}(x_h,x_t) \Downarrow [v_1,...,v_n]
}

\RuleToplabel{E:Let}
{
\ev \rtyps{q - K^{\word{let}}}{q'_1} e_1 \Downarrow v_1 \\
\ev[x \mapsto v_1] \rtyps{q'_1}{q'} e_2 \Downarrow v
}
{
\ev \rtyps{q}{q'} \text{let}(x,e_1,x.e_2) \Downarrow v
}

\RuleToplabel{E:Share}
{
\ev(x) = v_1 \\ \ev[x_1 \mapsto v_1,x_2 \mapsto v_1] \setminus \{x\} \rtyps{q}{q'} e \Downarrow v
}
{
\ev \rtyps{q}{q'} \text{share}(x,(x_1,x_2).e) \Downarrow v
}

\RuleToplabel{E:Match-L}
{
\ev(x) = [v_1,...,v_n] \\
\ev[x_h \mapsto v_1,x_t \mapsto [v_2,...,v_n]] \rtyps{q - K^{\word{matchL}}}{q'} e_2 \Downarrow v
}
{
\ev \rtyps{q}{q'} \text{match}(x,e_1,(x_h,x_t).e_2) \Downarrow v
}
\end{mathpar}
\caption{Evaluation typing rules}
\label{fig:bosrules}
\end{figure*}

\section*{Typing rules and proofs of Lemma \ref{lem:simplesecurity}, Lemma \ref{lem:noninterference}, and Theorem \ref{theo:securitysoundness}}

\subsection*{Typing rules}
The full typing rules of the resource-aware security type system are presented in Fig.~\ref{fig:securityrules}.
%
\begin{figure*}[!th]
\begin{mathpar}
\small
\RuleNolabel
{
	k \sqsubseteq k'\\
   T \in \{\type{unit}, \type{int}, \type{bool}\}
}
{
	k \triangleleft (T,k')
}

\RuleNolabel
{
k \sqsubseteq k' \\
\!\!\! k \triangleleft S
}
{
k \triangleleft (L(S),k')
}

\RuleNolabel
{
k \triangleleft S_1 \\
k \triangleleft S_2
}
{
k \triangleleft S_1 * S_2
}

\RuleNolabel
{
   k' \sqsubseteq k\\
   T \in \{\type{unit}, \type{int}, \type{bool}\}
}
{
	(T,k') \blacktriangleleft k
}

\RuleNolabel
{
k' \sqsubseteq k \\
S \blacktriangleleft k
}
{
(L(S),k') \blacktriangleleft k
}

\RuleNolabel
{
S_1 \blacktriangleleft k \\
S_2 \blacktriangleleft k
}
{
S_1 * S_2 \blacktriangleleft k
}

\RuleNolabel
{
	k \sqsubseteq k'  \\
	T {\in} \{\type{unit}, \type{int}, \type{bool}\}
}
{
	(\type{T},k) \leq (\type{T},k')
}

\RuleNolabel
{
k \sqsubseteq k' \\
S \leq S'
}
{
(L(S),k) {\leq} (L(S'),k')
}

\RuleNolabel
{
S_1 \leq S_1' \\
S_2 \leq S_2'
}
{
S_1 * S_2 \leq S_1' * S_2'
}
\end{mathpar}
\caption{Guards, collecting, and subtyping relations}
\label{fig:secsubtype}
\end{figure*}

\begin{figure*}[!ht]
\begin{mathpar}
\small
\RuleToplabel{SR:Unit}
{
}
{
\type{pc};\sectstack;\seccontext \rtyps{\type{const}}{} ():(\type{unit},\type{pc})
}

\RuleToplabel{SR:Bool}
{
b \in \{\word{true},\word{false}\}
}
{
\type{pc};\sectstack;\seccontext \rtyps{\type{const}}{} b:(\type{bool},\type{pc})
}

\RuleToplabel{SR:Int}
{
n \in \mathbb{Z}
}
{
\type{pc};\sectstack;\seccontext \rtyps{\type{const}}{} n:(\type{int},\type{pc})
}

\RuleToplabel{SR:Var}
{
x: S \in \seccontext \\
\type{pc} \triangleleft S
}
{\type{pc};\sectstack;\seccontext \rtyps{\type{const}}{} x: S}

\RuleToplabel{SR:B-Op}
{
x_1:(\type{bool},k_{x_1}) \in \seccontext \\
x_2:(\type{bool},k_{x_2}) \in \seccontext \\
\type{pc} \sqsubseteq k_{x_1} \sqcup k_{x_2} \\
\diamond \in \{\word{and},\word{or}\}
}
{
\type{pc};\sectstack;\seccontext \rtyps{\type{const}}{} \text{op}_{\diamond}(x_1,x_2):(\type{bool},k_{x_1} \sqcup k_{x_2})
}

\RuleToplabel{SR:C-Gen}
{
\type{pc};\sectstack;\seccontext \vdash e: S \\
\type{const}_{X}(e)
}
{
\type{pc};\sectstack;\seccontext \rtyps{\type{const}}{} e: S
}

\RuleToplabel{SR:Gen}
{
\type{pc};\sectstack;\seccontext \rtyps{\type{const}}{} e: S
}
{
\type{pc};\sectstack;\seccontext \vdash e: S
}

\RuleToplabel{SR:IB-Op}
{
x_1:(\type{int},k_{x_1}) \in \seccontext \\
x_2:(\type{int},k_{x_2}) \in \seccontext \\\\
\type{pc} \sqsubseteq k_{x_1} \sqcup k_{x_2} \\
\diamond \in \{=,<>,>,<,<=,>=\}
}
{
\type{pc};\sectstack; \seccontext \rtyps{\type{const}}{} \text{op}_{\diamond}(x_1,x_2):(\type{bool},k_{x_1} \sqcup k_{x_2})
}

\RuleToplabel{SR:Pair}
{
x_1: S_1 \in \seccontext \\
x_2: S_2 \in \seccontext \\
\type{pc} \triangleleft S_1 * S_2
}
{
\type{pc};\sectstack;\seccontext \rtyps{\type{const}}{} \text{pair}(x_1,x_2):S_1 * S_2
}

\RuleToplabel{SR:I-Op}
{
x_1:(\type{int},k_{x_1}) \in \seccontext \\
x_2:(\type{int},k_{x_2}) \in \seccontext \\\\
\type{pc} \sqsubseteq k_{x_1} \sqcup k_{x_2} \\
\diamond \in \{+,-,*,\word{div},\word{mod}\}
}
{\type{pc};\sectstack;\seccontext \rtyps{\type{const}}{} \text{op}_{\diamond}(x_1,x_2):(\type{int},k_{x_1} \sqcup k_{x_2})}

\RuleToplabel{SR:Cons}
{
x_h: S \in \seccontext \\
x_t: (L(S),k_x) \in \seccontext \\\\
\type{pc} \triangleleft (L(S),k_x)
}
{
\type{pc};\sectstack;\seccontext \rtyps{\type{const}}{} \text{cons}(x_h,x_t):(L(S),k_x)
}

\RuleToplabel{SR:SubTyping}
{
\type{pc};\sectstack;\seccontext \vdash e: S \\
S \leq S'
}
{
\type{pc};\sectstack;\seccontext \vdash e: S'
}

\RuleToplabel{SR:Fun}
{
x: S_1 \in \seccontext \\
\sectstack(f) = S_1 \xrightarrow{\type{pc}'} S_2 \\
\type{pc} \sqsubseteq \type{pc}'
}
{
\type{pc};\sectstack;\seccontext \vdash \text{app}(f,x): S_2
}

\RuleToplabel{SR:L-Arg}
{
x: S_1 \in \seccontext \\
\sectstack(f) = S_1 \xrightarrow{\type{pc}'} S_2 \\
\type{pc} \sqsubseteq \type{pc}' \\
S_1 \blacktriangleleft k_1
}
{
\type{pc};\sectstack;\seccontext \rtyps{\type{const}}{} \text{app}(f,x): S_2
}

\RuleToplabel{SR:C-Fun}
{
x: S_1 \in \seccontext \\
\sectstack(f) = S_1 \xrightarrow{\type{pc}'/\type{const}} S_2 \\
\type{pc} \sqsubseteq \type{pc}'
}
{
\type{pc};\sectstack;\seccontext \rtyps{\type{const}}{} \text{app}(f,x): S_2
}

\RuleToplabel{SR:Nil}
{
S \in \secdatatypes \\
\type{pc} \triangleleft S
}
{
\type{pc};\sectstack;\seccontext \rtyps{\type{const}}{} \word{nil}:(L(S),\type{pc})
}

\RuleToplabel{SR:C-SubTyping}
{
\type{pc};\sectstack;\seccontext \rtyps{\type{const}}{} e: S \\
S \leq S'
}
{
\type{pc};\sectstack;\seccontext \rtyps{\type{const}}{} e: S'
}

\RuleToplabel{SR:Let}
{
\type{pc};\sectstack;\seccontext \vdash e_1: S_1 \\
\type{pc};\sectstack;\seccontext,x: S_1 \vdash e_2: S_2
}
{
\type{pc};\sectstack;\seccontext \vdash \text{let}(x,e_1,x.e_2): S_2
}

\RuleToplabel{SR:L-Let}
{
\type{pc};\sectstack;\seccontext \rtyps{\type{const}}{} e_1: S_1 \\
\type{pc};\sectstack;\seccontext,x: S_1 \rtyps{\type{const}}{} e_2: S_2 \\
S_1 \blacktriangleleft k_1
}
{
\type{pc};\sectstack;\seccontext \rtyps{\type{const}}{} \text{let}(x,e_1,x.e_2): S_2
}

\RuleToplabel{SR:If}
{
x:(\type{bool},k_x) \in \seccontext \\
\type{pc} \sqcup k_x;\sectstack;\seccontext \vdash e_t: S  \\
\type{pc} \sqcup k_x;\sectstack;\seccontext \vdash e_f: S \\
\type{pc} \sqcup k_x \triangleleft S
}
{\type{pc};\sectstack;\seccontext \vdash \text{if}(x,e_t,e_f): S}

\RuleToplabel{SR:L-If}
{
x:(\type{bool},k_x) \in \seccontext \\
\type{pc} \sqcup k_x;\sectstack;\seccontext \rtyps{\type{const}}{} e_t: S  \\
\type{pc} \sqcup k_x;\sectstack;\seccontext \rtyps{\type{const}}{} e_f: S \\
\type{pc} \sqcup k_x \triangleleft S \\
k_x \sqsubseteq k_1
}
{\type{pc};\sectstack;\seccontext \rtyps{\type{const}}{} \text{if}(x,e_t,e_f): S}

\RuleToplabel{SR:Match-P}
{
x: S_1 * S_2 \in \seccontext \\
\type{pc};\sectstack;\seccontext,x_1:S_1,x_2:S_2 \vdash e: S
}
{
\type{pc};\sectstack;\seccontext \vdash \text{match}(x,(x_1,x_2).e): S
}

\RuleToplabel{SR:C-Match-P}
{
x: S_1 * S_2 \in \seccontext \\
\type{pc};\sectstack;\seccontext,x_1:S_1,x_2:S_2 \rtyps{\type{const}}{} e: S
}
{
\type{pc};\sectstack;\seccontext \rtyps{\type{const}}{} \text{match}(x,(x_1,x_2).e): S
}

\RuleToplabel{SR:Match-L}
{
x:(L(S),k_x) \in \seccontext \\
\type{pc} \sqcup k_x;\sectstack;\seccontext \vdash e_1: S_1 \\
\type{pc} \sqcup k_x;\sectstack;\seccontext,x_h: S,x_t:(L(S),k_x) \vdash e_2: S_1 \\
\type{pc} \sqcup k_x \triangleleft S_1
}
{
\type{pc};\sectstack;\seccontext \vdash \text{match}(x,e_1,(x_h,x_t).e_2): S_1
}

\RuleToplabel{SR:C-Match-L}
{
x:(L(S),k_x) \in \seccontext \\
\type{pc} \sqcup k_x;\sectstack;\seccontext \rtyps{\type{const}}{} e_1: S_1 \\
\type{pc} \sqcup k_x;\sectstack;\seccontext,x_h: S,x_t:(L(S),k_x) \rtyps{\type{const}}{} e_2: S_1 \\
\type{pc} \sqcup k_x \triangleleft S_1 \\
k_x \sqsubseteq k_1
}
{
\type{pc};\sectstack;\seccontext \rtyps{\type{const}}{} \text{match}(x,e_1,(x_h,x_t).e_2): S_1
}
\end{mathpar}
\caption{Typing rules for resource-aware security type system}
\label{fig:securityrules}
\end{figure*}

\subsection*{Proof of Lemma \ref{lem:simplesecurity}}
The proof is done by induction on the structure of the typing derivation.

\paragraph{\textsc{SR:Unit}} There is no variable thus it follows immediately.

\paragraph{\textsc{SR:Bool}} It is similar to the case \textsc{SR:Unit}.

\paragraph{\textsc{SR:Int}} It is similar to the case \textsc{SR:Unit}.

\paragraph{\textsc{SR:Var}} Since $\seccontext(x) = S$, if $S \blacktriangleleft k_1$ then $\seccontext(x) \blacktriangleleft k_1$.

\paragraph{\textsc{SR:B-Op}} If $(\type{bool},k_{x_1} \sqcup k_{x_2}) \blacktriangleleft k_1$ then $\seccontext(x_1) = (\type{bool},k_{x_1}) \blacktriangleleft k_1$ and $\seccontext(x_2) = (\type{bool},k_{x_2}) \blacktriangleleft k_1$.

\paragraph{\textsc{SR:IB-Op}} It is similar to the case \textsc{SR:B-Op}.

\paragraph{\textsc{SR:I-Op}} It is similar to the case \textsc{SR:B-Op}.

\paragraph{\textsc{SR:Gen} - \textsc{SR:C-Gen}} By induction for $e$ in the premise, it follows.

\paragraph{\textsc{SR:Fun}} Because $e$ is well-formed program, there exists a well-typed expression $e_f$ such that $\type{pc'};\sectstack;\seccontext \vdash e_f: S_2$. By induction for $e_f$, for all variables $x$ in $e$, if $S \blacktriangleleft k_1$ then $\seccontext(x) \blacktriangleleft k_1$. It is similar for \textsc{SR:L-Arg} and \textsc{SR:C-Fun}.

\paragraph{\textsc{S:Let}} If $S_2 \blacktriangleleft k_1$ then by induction for $e_2$, $S_1 \blacktriangleleft k_1$. Thus for all variable $x$ in $e$, it is a variable in $e_1$ or $e_2$. By induction for $e_1$ and $e_2$, it follows. It is similar for \textsc{SR:L-Let}.

\paragraph{\textsc{SR:If}} If $S \blacktriangleleft k_1$ then by the hypothesis $(\type{bool},k_x) \blacktriangleleft k_1$. For all variable $y$ in $e$, it is a variable in $e_t$ or $e_f$. By induction for $e_t$ and $e_f$, it follows. It is similar for \textsc{SR:L-If}.

\paragraph{\textsc{SR:Pair}} If $S_1 * S_2 \blacktriangleleft k_1$ then $\seccontext(x_1) = S_1 \blacktriangleleft k_1$ and $\seccontext(x_2) = S_2 \blacktriangleleft k_1$.

\paragraph{\textsc{SR:Match-P}} If $S \blacktriangleleft k_1$ then by induction for $e$, $\seccontext(x_1) \blacktriangleleft k_1$ and $\seccontext(x_2) \blacktriangleleft k_1$. Thus $\seccontext(x) \blacktriangleleft k_1$. For all other variables $y$ in $e$, again by induction for $e$, if $S \blacktriangleleft k_1$ then $\seccontext(y) \blacktriangleleft k_1$. It is similar for \textsc{SR:C-Match-P}.

\paragraph{\textsc{SR:Nil}} It is similar to the case \textsc{SR:Unit}.

\paragraph{\textsc{SR:Cons}} If $(L(S),k_x) \blacktriangleleft k_1$ then $\seccontext(x_h) = S \blacktriangleleft k_1$ and $\seccontext(x_t) = (L(S),k_x) \blacktriangleleft k_1$.

\paragraph{\textsc{SR:Match-L}} If $S_1 \blacktriangleleft k_1$ then by induction for $e_2$, $\seccontext(x_h) = S \blacktriangleleft k_1$ and $\seccontext(x_t) = (L(S),k_x) \blacktriangleleft k_1$. Thus $\seccontext(x) \blacktriangleleft k_1$. For all other variables $y$ in $e$, $y$ is a variable in $e_1$ or $e_2$. Again by induction for $e_1$ and $e_2$, if $S_1 \blacktriangleleft k_1$ then $\seccontext(y) \blacktriangleleft k_1$.

\paragraph{\textsc{SR:SubTyping}} By the subtyping relation, if $S' \blacktriangleleft k_1$ then $S \blacktriangleleft k_1$. Thus by induction for $e$ in the premise, for all variables $x$ in $e$, if $S \blacktriangleleft k_1$ then $\seccontext(x) \blacktriangleleft k_1$. It is similar for \textsc{SR:C-SubTyping}.

\subsection*{Proof of Lemma \ref{lem:noninterference}}
The proof is done by induction on the structure of the evaluation derivation and the typing derivation.

\paragraph{\textsc{SR:Unit}} Suppose the evaluation derivation of $e$ ends with an application of the rule \textsc{E:Unit}, thus $\ev_1 \vdash e \Downarrow \type{()}$ and $\ev_2 \vdash e \Downarrow \type{()}$. Hence, it follows.

\paragraph{\textsc{SR:Bool}} It is similar to the case \textsc{SR:Unit}.

\paragraph{\textsc{SR:Int}} It is similar to the case \textsc{SR:Unit}.

\paragraph{\textsc{SR:Var}} Suppose the evaluation derivation ends with an application of the rule \textsc{E:Var}, thus $\ev_1(x) = v_1$ and $\ev_2(x) = v_2$. The typing derivation ends with an application of the rule \textsc{SR:Var}, thus $\seccontext(x) = S$. If $S \blacktriangleleft k_1$, by the hypothesis $\ev_1(x) = \ev_2(x)$ since $x \in \dom{E_i}, i = \{1,2\}$.

\paragraph{\textsc{SR:B-Op}} Suppose the evaluation derivation ends with an application of the rule \textsc{E:Bin}, thus $\ev_1(x_1) \diamond \ev_1(x_2) = v_1$ and $\ev_2(x_1) \diamond \ev_2(x_2) = v_2$. The typing derivation ends with an application of the rules \textsc{SR:B-Op} or \textsc{SR:Gen}. We have $k_{x_1} \triangleleft S$ and $k_{x_2} \triangleleft S$. If $S \blacktriangleleft k_1$ then $k_{x_1} \sqsubseteq k_1$ and $k_{x_2} \sqsubseteq k_1$. By the hypothesis, we have $\ev_1(x_1) = \ev_2(x_1)$ and $\ev_1(x_2) = \ev_2(x_2)$, thus $v_1 = v_2$.

\paragraph{\textsc{SR:IB-Op}} It is similar to the case \textsc{SR:B-Op}.

\paragraph{\textsc{SR:I-Op}} It is similar to the case \textsc{SR:B-Op}.

\paragraph{\textsc{SR:Gen}-\textsc{SR:C-Gen}} By induction for $e$ in the premise, it follows that if $S \blacktriangleleft k_1$ then $v_1 = v_2$.

\paragraph{\textsc{SR:Fun}} Suppose the evaluation derivation ends with an application of the rule \textsc{E:Fun}, thus $\tstack(g) = T_1 \rightarrow T_2$ and $[y^g \mapsto \ev_i(x)] \vdash e_g \Downarrow v_i$ for $i = \{1,2\}$. The typing derivation ends with an application of the following rules.
\begin{itemize}
	\item Case \textsc{SR:Fun}. Because $e$ is well-formed program, there exists a well-typed expression $e_f$ such that $\type{pc'};\sectstack;\seccontext \vdash e_f: S_2$ and $e_{\atype{f}} = e_g$. By induction for $e_f$, if $S_2 \blacktriangleleft k_1$ then $v_1 = v_2$.
	\item Case \textsc{SR:L-Arg}. It is similar to the case \textsc{SR:Fun}.
	\item Case \textsc{SR:C-Fun}. It is similar to the case \textsc{SR:Fun}.
	\item Case \textsc{SR:Gen} and \textsc{SR:C-Gen}. It follows.
\end{itemize}

\paragraph{\textsc{S:Let}} Suppose the evaluation derivation ends with an application of the rule \textsc{E:Let}, thus $\ev_i \vdash e_1 \Downarrow v^i_1$ and $\ev_i[x \mapsto v^i_1] \vdash e_2 \Downarrow v_i$ for $i = \{1, 2\}$. The typing derivation ends with an application of the following rules.
\begin{itemize}
	\item Case \textsc{SR:L-Let}. If $S_2 \blacktriangleleft k_1$, by the simple security lemma, it holds that $S_1 \blacktriangleleft k_1$. By induction for $e_1$, we have $v^1_1 = v^2_1$, so $\ev_1[x \mapsto v^1_1] \secesim{k} \ev_2[x \mapsto v^2_1]$. Again by induction for $e_2$, we have $v_1 = v_2$.
	\item Case \textsc{SR:Let}. It is similar to the case \textsc{SR:L-Let}.
	\item Case \textsc{SR:Gen} and \textsc{SR:C-Gen}. It follows.
\end{itemize}

\paragraph{\textsc{SR:If}} Suppose $e$ is of the form $\text{if}(x,e_t,e_f)$, the evaluation derivation ends with an application of the rule \textsc{E:If-True} or the rule \textsc{E:If-False}. The typing derivation ends with an application of the following rules.
\begin{itemize}
	\item Case \textsc{SR:L-If}. By the hypothesis we have $k_x \sqsubseteq k_1$, thus $\ev_1(x) = \ev_2(x)$. Assume that $\ev_1(x) = \word{true}$, then $\ev_1 \vdash e_t \Downarrow v_1$ and $\ev_2 \vdash e_t \Downarrow v_2$. By induction for $e_t$ we have $v_1 = v_2$ if $S \blacktriangleleft k_1$. It is similar for $\ev_1(x) = \word{false}$.
	\item Case \textsc{SR:If}. If $k_x \sqsubseteq k_1$ the proof is similar to the case \textsc{SR:L-If}. Otherwise, $k_x \not \sqsubseteq k_1$, thus by the simple security lemma we have $S \not \blacktriangleleft k_1$.
	\item Case \textsc{SR:Gen} and \textsc{SR:C-Gen}. It follows.
\end{itemize}

\paragraph{\textsc{SR:Pair}} Suppose the evaluation derivation ends with an application of the rule \textsc{E:Pair}, thus $(\ev_i(x_1),\ev_i(x_2)) = v_i$ for $i = \{1, 2\}$. The typing derivation ends with an application of the rules \textsc{SR:Pair} or \textsc{SR:Gen}.

If $S_1 * S_2 \blacktriangleleft k$, then by the simple security lemma we have $S_1 \blacktriangleleft k_1$ and $S_2 \blacktriangleleft k_1$. Hence it follows $v_1 = v_2$.

\paragraph{\textsc{SR:Match-P}} Suppose the evaluation derivation ends with an application of the rule \textsc{E:Match-P}, thus $\ev_i(x) = (v^i_1,v^i_2)$ and $\ev_i[x_1 \mapsto v^i_1,x_2 \mapsto v^i_2] \vdash e \Downarrow v_i$ for $i = \{1, 2\}$. The typing derivation ends with an application of the following rules.
\begin{itemize}
	\item Case \textsc{SR:Match-P}. If $S \blacktriangleleft k_1$, then by the simple security lemma we have $S_1 * S_2 \blacktriangleleft k_1$. By the hypothesis, $\ev_1(x) = \ev_2(x)$, thus $v^1_1 = v^2_1$ and $v^1_2 = v^2_2$. Hence, $\ev_1[x_1 \mapsto v^1_1,x_2 \mapsto v^1_2] \secesim{k} \ev_2[x_1 \mapsto v^2_1,x_2 \mapsto v^2_2]$, by induction for $e$ in the premise, it holds that $v_1 = v_2$.
	\item Case \textsc{SR:C-Match-P}. It is similar to the case \textsc{SR:Match-P}.
	\item Case \textsc{SR:Gen} and \textsc{SR:C-Gen}. It follows.
\end{itemize}

\paragraph{\textsc{SR:Nil}} It is similar to the case \textsc{SR:Unit}.

\paragraph{\textsc{SR:Cons}} Suppose the evaluation derivation ends with an application of the rule \textsc{E:Cons}, thus $\ev_i(x_h) = v^i_1$ and $\ev_i(x_t) = [v^i_2, \cdots ,v^i_n]$ for $i = \{1, 2\}$. The typing derivation ends with an application of the rules \textsc{SR:Cons} or \textsc{SR:Gen}. If $(L(S),k_x) \blacktriangleleft k_1$ then by the hypothesis we have $v^1_1 = v^2_1$ and $[v^1_2, \cdots ,v^1_n] = [v^2_2, \cdots ,v^2_n]$. Thus $\ev_1(\word{cons}(x_h,x_t)) = \ev_2(\word{cons}(x_h,x_t))$.

\paragraph{\textsc{SR:Match-L}} Suppose $e$ is of the form $\text{match}(x,e_1,(x_h,x_t).e_2)$, the evaluation derivation ends with an application of the rule \textsc{E:Match-N} or the rule \textsc{E:Match-L}. The typing derivation ends with an application of the following rules.
\begin{itemize}
	\item Case \textsc{SR:Match-L}. If $S_1 \blacktriangleleft k_1$, then by the simple security lemma we have $(L(S),k_x) \blacktriangleleft k_1$. By the hypothesis we have $\ev_1(x) = \ev_2(x)$. Assume that $\ev_1(x) = \ev_2(x) = [v_1, \cdots ,v_n]$, by the rule \textsc{E:Match-L} we have $\ev_i[x_h \mapsto v_1,x_t \mapsto [v_2,...,v_n]] \vdash e_2 \Downarrow v_i$ for $i = \{1, 2\}$. Since $\ev_1[x_h \mapsto v_1,x_t \mapsto [v_2,...,v_n]] \secesim{k} \ev_2[x_h \mapsto v_1,x_t \mapsto [v_2,...,v_n]]$, by induction for $e_2$, it holds that $v_1 = v_2$ if $S_1 \blacktriangleleft k_1$. It is similar for  $\ev_1(x) = \ev_2(x) = \word{nil}$.
	\item Case \textsc{SR:C-Match-L}.It is similar to the case \textsc{SR:Match-L}.
	\item Case \textsc{SR:Gen} and \textsc{SR:C-Gen}. It follows.
\end{itemize}

\paragraph{\textsc{SR:SubTyping}} Suppose the typing derivation ends with the rule \textsc{SR:SubTyping}. If $S' \blacktriangleleft k_1$ then $S \blacktriangleleft k_1$. Thus by induction for $e$ in the premise it follows. It is similar for \textsc{SR:C-SubTyping}.

\subsection*{Proof of Theorem \ref{theo:securitysoundness}}
The proof is done by induction on the structure of the typing derivation and the evaluation derivation. Let $X$ be the set of variables $\hiddenvars{\seccontext}{k_1}$. For all environments $\ev_1$, $\ev_2$ such that $\ev_1 \esim{X} \ev_2$ and $\ev_1 \secesim{k_1} \ev_2$, if $\ev_1 \rtyps{p_1}{p_1'} e \Downarrow v_1$ and $\ev_2 \rtyps{p_2}{p_2'} e \Downarrow v_2$. We then show that $p_1 - p_1' = p_2 - p_2'$ and $v_1 = v_2$ if $S \blacktriangleleft k_1$. By Lemma \ref{lem:noninterference}, $e$ satisfies the noninterference property at security label $k_1$. Thus we need to prove that $p_1 - p_1' = p_2 - p_2'$.

\paragraph{\textsc{SR:Unit}} Suppose the evaluation derivation of $e$ ends with an application of the rule \textsc{E:Unit}, thus  $p_1 - p_1' = p_2 - p_2' = K^{\word{unit}}$.

\paragraph{\textsc{SR:Bool}} It is similar to the case \textsc{SR:Unit}.

\paragraph{\textsc{SR:Int}} It is similar to the case \textsc{SR:Unit}.

\paragraph{\textsc{SR:Var}} It is similar to the case \textsc{SR:Unit}.

\paragraph{\textsc{SR:B-Op}} Suppose the evaluation derivation ends with an application of the rule \textsc{E:Bin}, thus $\ev_1 \rtyps{p_1' + K^{\word{op}}}{p_1'} e \Downarrow v_1$ and $\ev_1 \rtyps{p_2' + K^{\word{op}}}{p_2'} e \Downarrow v_1$. We have $p_1 - p_1' = p_2 - p_2' = K^{\word{op}}$.

\paragraph{\textsc{SR:IB-Op}} It is similar to the case \textsc{SR:B-Op}.

\paragraph{\textsc{SR:I-Op}} It is similar to the case \textsc{SR:B-Op}.

\paragraph{\textsc{SR:C-Gen}} By the hypothesis we have $\type{const}_{X}(e)$, thus it holds that $\reststack;\rescontext \rtyps{q}{q'} e:A$ and $\share(A \mid A,A)$. By the constant-resource theorem, for all $p_1, p_1', p_2, p_2' \in \mathbb{Q}^+_0$ such that $\ev_1 \rtyps{p_1}{p_1'} e \Downarrow v_1$ and $\ev_2 \rtyps{p_2}{p_2'} e \Downarrow v_2$, we have $p_1 - p_1' = q + \poten{\ev_1}(\rescontext) - (q' + \poten{}(v_1:A))$ and $p_2 - p_2' = q + \poten{\ev_2}(\rescontext) - (q' + \poten{}(v_2:A))$.

Since $\ev_1 \esim{X} \ev_2$, $\poten{\ev_1}(X) = \poten{\ev_2}(X)$. For all $y \not \in X$, $\ev_1(y) = \ev_2(y)$ since $\ev_1 \secesim{k_1} \ev_2$, thus $\poten{}(\ev_1(y)) = \poten{}(\ev_2(y))$. Hence, $\poten{\ev_1}(\rescontext) = \poten{\ev_2}(\rescontext)$, it follows $p_1 - p_1' = p_2 - p_2'$.

\paragraph{\textsc{SR:Fun}} Suppose $e$ is of the form $\text{app}(f,x)$, thus the typing derivation ends with an application of either the rule \textsc{SR:L-Arg}, \textsc{SR:C-Fun}, or \textsc{SR:C-Gen}.
\begin{itemize}
	\item Case \textsc{SR:L-Arg}. By the hypothesis we have $\ev_1(x) = \ev_2(x)$, it follows $p_1 - p_1' = p_2 - p_2'$.
	\item Case \textsc{SR:C-Fun}. Because $e$ is well-formed, there exists a well-typed expression $e_f$ such that $\type{pc'};\sectstack;\seccontext \rtyps{\type{const}}{} e_f: S_2$. By induction for $e_f$ which is resource-aware noninterference w.r.t $X$, $p_1 - K^{\word{app}} - p_1' = p_2 - K^{\word{app}} - p_2'$, it follows.
	\item Case \textsc{SR:C-Gen}. By the case \textsc{SR:C-Gen} it follows.
\end{itemize}

\paragraph{\textsc{SR:Let}} Suppose $e$ is of the form $\text{let}(x,e_1,x.e_2)$, thus the typing derivation ends with an application of either the rule \textsc{SR:L-Let} or \textsc{SR:C-Gen}.
\begin{itemize}
	\item Case \textsc{SR:L-Let}. Suppose the evaluations $\ev_1 \rtyps{p_1 - K^{\word{let}}}{p'} e_1 \Downarrow v^1_1$, $\ev_2 \rtyps{p_2 - K^{\word{let}}}{p"} e_1 \Downarrow v^2_1$, $\ev_1[x \mapsto v^1_1] \rtyps{p'}{p_1'} e_2 \Downarrow v_1$, and $\ev_2[x \mapsto v^2_1] \rtyps{p"}{p_2'} e_2 \Downarrow v_2$. By induction for $e_1$ that is resource-aware noninterference w.r.t $X$, $p_1 - K^{\word{let}} - p' = p_2 - K^{\word{let}} - p"$. By the hypothesis $v^1_1 = v^2_1$. Thus $\ev_1[x \mapsto v^1_1] \esim{X} \ev_2[x \mapsto v^2_1]$ and $\ev_1[x \mapsto v^1_1] \secesim{k_1} \ev_2[x \mapsto v^2_1]$, by induction for $e_2$ that is resource-aware noninterference w.r.t $X$, we have $p' - p_1' = p" - p_2'$. Hence, $p_1 - p_1' = p_2 - p_2'$.
	\item Case \textsc{SR:C-Gen}. By the case \textsc{SR:C-Gen} it follows.
\end{itemize}

\paragraph{\textsc{SR:If}} Suppose $e$ is of the form $\text{if}(x,e_t,e_f)$, thus the typing derivation ends with an application of either the rule \textsc{SR:L-If} or \textsc{SR:C-Gen}.
\begin{itemize}
	\item Case \textsc{SR:L-If}. By the hypothesis we have $\ev_1(x) = \ev_2(x)$. Assume that $\ev_1(x) = \ev_2(x) = \word{true}$, by the evaluation rule \textsc{E:If-True}, $\ev_1 \rtyps{p_1-K^{\word{cond}}}{p_1'} e_t \Downarrow v_1$ and $\ev_2 \rtyps{p_2-K^{\word{cond}}}{p_2'} e_t \Downarrow v_2$. By induction for $e_t$ that is resource-aware noninterference w.r.t $X$, we have $p_1 - p_1' = p_2 - p_2'$. It is similar for $\ev_1(x) = \ev_2(x) = \word{false}$.
	\item Case \textsc{SR:C-Gen}. Since $\ev_1 \esim{X} \ev_2$ w.r.t $\seccontext$, we have $\ev_1 \esim{X} \ev_2$ w.r.t $\rescontext$. By the hypothesis we have $\type{const}_{X}(e)$. Thus by the soundness theorem of constant resource type system, it follows $p_1 - p_1' = p_2 - p_2'$.
\end{itemize}

\paragraph{\textsc{SR:Pair}} It is similar to the case \textsc{SR:B-Op}.

\paragraph{\textsc{SR:Match-P}} Suppose $e$ is of the form $\text{match}(x,(x_1,x_2).e)$, thus the typing derivation ends with an application of either the rule \textsc{SR:C-Match-P} or \textsc{SR:C-Gen}.
\begin{itemize}
	\item Case \textsc{SR:C-Match-P}. Let $\ev_1' = \ev_1[x_1 \mapsto v^1_1,x_2 \mapsto v^1_2]$ and $\ev_2' = \ev_2[x_1 \mapsto v^2_1,x_2 \mapsto v^2_2]$. If $x \in X$ then $\sizeval{\ev_1(x)} \esim{} \sizeval{\ev_2(x)}$. Thus $\sizeval{\ev_1'(x_1)} \esim{} \sizeval{\ev_2'(x_1)}$ and $\sizeval{\ev_1'(x_2)} \esim{} \sizeval{\ev_2'(x_2)}$. Hence, $\ev_1' \esim{X \cup \{x_1,x_2\}} \ev_2'$, by induction for $e$ in the premise which is resource-aware noninterference w.r.t $X \cup \{x_1,x_2\}$, $p_1 - K^{\word{matchP}} - p_1' = p_2 - K^{\word{matchP}} - p_2'$, it follows. If $x \not \in X$ then $\ev_1(x) = \ev_2(x)$, it is similar.
	\item Case \textsc{SR:C-Gen}. By the case \textsc{SR:C-Gen} it follows.
\end{itemize}

\paragraph{\textsc{SR:Nil}} It is similar to the case \textsc{SR:Unit}.

\paragraph{\textsc{SR:Cons}} It is similar to the case \textsc{SR:B-Op}.

\paragraph{\textsc{SR:Match-L}} Suppose $e$ is of the form $\text{match}(x,e_1,(x_h,x_t).e_2)$, thus the typing derivation ends with an application of either the rule \textsc{SR:C-Match-L} or \textsc{SR:C-Gen}.
\begin{itemize}
	\item Case \textsc{SR:C-Match-L}. Let $\ev_1' = \ev_1[x_h \mapsto v^1_1,x_t \mapsto v^1_2]$ and $\ev_2' = \ev_2[x_h \mapsto v^2_1,x_t \mapsto v^2_2]$. If $x \in X$ then $\sizeval{\ev_1(x)} \esim{} \sizeval{\ev_2(x)}$. Suppose $\ev_1(x)$ and $\ev_2(x)$ are different from $\word{nil}$, $\sizeval{\ev_1'(x_h)} \esim{} \sizeval{\ev_2'(x_h)}$ and $\sizeval{\ev_1'(x_t)} \esim{} \sizeval{\ev_2'(x_t)}$. Hence, $\ev_1' \esim{X \cup \{x_t,x_h\}} \ev_2'$, by induction for $e_2$ which is resource-aware noninterference w.r.t $X \cup \{x_t,x_h\}$, we have $p_1 - K^{\word{matchL}} - p_1' = p_2 - K^{\word{matchL}} - p_2'$, thus $p_1 - p_1' = p_2 - p_2'$. If $\ev_1(x) = \ev_2(x) = \word{nil}$ then by induction for $e_1$ that is resource-aware noninterference w.r.t $X$, it follows. If $x \not \in X$ then $\ev_1(x) = \ev_2(x)$, it is similar.
	\item Case \textsc{SR:C-Gen}. By the case \textsc{SR:C-Gen} it follows.
\end{itemize}

\paragraph{\textsc{SR:SubTyping}} The typing derivation ends with an application of either the rule \textsc{SR:C-Share} or \textsc{SR:C-Gen}.
\begin{itemize}
	\item Case \textsc{SR:C-SubTyping}. By induction for $e$ in the premise, $p_1 - p_1' = p_2 - p_2'$.
	\item Case \textsc{SR:C-Gen}. By the case \textsc{SR:C-Gen} it follows.
\end{itemize}

\section*{Type systems for lower bounds and constant resource}
The common syntax-directed typing rules for all of three type systems; upper bounds, constant resource, and lower bounds are represented in Fig.~\ref{fig:commonrules}. While the different structural rules are shown in Fig.~\ref{fig:upperbound}, Fig.~\ref{fig:constant}, and Fig.~\ref{fig:lowerbound}. We can see that the relax rules are consistent among these type systems in sense of satisfying the following.
\begin{displaymath}
\begin{array}{ll}
	                & (q \geq p \wedge q - p \leq q' - p') \wedge (q \geq p \wedge q - p \geq q' - p') \\
	\Leftrightarrow & (q \geq p \wedge q - p = q' - p')
\end{array}
\end{displaymath}
That means the constraints for upper bounds and lower bounds imply the constraints for constant resource and vice versa.

The type systems for upper bounds, constant resource, and lower bounds are \emph{affine}, \emph{linear}, and \emph{relevant} sub-structural type systems, respectively. 
\begin{itemize}
	\item Type system for upper bounds allows exchange and weakening, but not contraction properties.
	\item Type system for constant resource allows exchange but not weakening or contraction properties.
	\item Type system for lower bounds allows exchange and contraction, but not weakening properties.
\end{itemize}
%
\begin{figure*}[!ht]
\begin{mathpar}
\small
\RuleToplabel{A:Unit}
{
}
{
\reststack;\emptyset \rtyps{K^{\word{unit}}}{0} ():\type{unit}
}

\RuleToplabel{A:B-Op}
{
	\diamond \in \{\word{and},\word{or}\}
}
{
	\reststack;x_1:\type{bool},x_2:\type{bool} \rtyps{K^{\word{op}}}{0} \text{op}_{\diamond}(x_1,x_2):\type{bool}
}

\RuleToplabel{A:IB-Op}
{
\diamond \in \{=,<>,>,<,<=,>=\}
}
{
\reststack;x_1:\type{int},x_2:\type{int} \rtyps{K^{\word{op}}}{0} \text{op}_{\diamond}(x_1,x_2):\type{bool}
}

\RuleToplabel{A:Int}
{
n \in \mathbb{Z}
}
{
\reststack;\emptyset \rtyps{K^{\word{int}}}{0} n:\type{int}
}

\RuleToplabel{A:I-Op}
{
\diamond \in \{+,-,*,\word{div},\word{mod}\}
}
{
\reststack;x_1:\type{int},x_2:\type{int} \rtyps{K^{\word{op}}}{0} \text{op}_{\diamond}(x_1,x_2):\type{int}
}

\RuleToplabel{A:Fun}
{
	\reststack(f) = A_1 \xrightarrow{q/q'} A_2
}
{
	\reststack;x: A_1 \rtyps{q + K^{\word{app}}}{q'} \text{app}(f,x): A_2
}

\RuleToplabel{A:Let}
{
	\reststack;\rescontext_1 \rtyps{q - K^{\word{let}}}{q'_1} e_1: A_1 \\
	\reststack;\rescontext_2,x: A_1 \rtyps{q'_1}{q'} e_2: A_2
}
{
	\reststack;\rescontext_1,\rescontext_2 \rtyps{q}{q'} \text{let}(x,e_1,x.e_2): A_2
}

\RuleToplabel{A:If}
{
	\reststack;\rescontext \rtyps{q - K^{\word{cond}}}{q'} e_t: A \\
	\reststack;\rescontext \rtyps{q - K^{\word{cond}}}{q'} e_f: A
}
{
	\reststack;\rescontext,x:\type{bool} \rtyps{q}{q'} \text{if}(x,e_t,e_f): A
}

\RuleToplabel{A:Pair}
{
}
{
\reststack;x_1:A_1,x_2:A_2 \rtyps{K^{\word{pair}}}{0} \text{pair}(x_1,x_2):A_1 * A_2
}

\RuleToplabel{A:Match-P}
{
	\reststack;\rescontext,x_1:A_1,x_2:A_2 \rtyps{q - K^{\word{matchP}}}{q'} e: A
}
{
	\reststack;\rescontext,x:A_1 * A_2 \rtyps{q}{q'} \text{match}(x,(x_1,x_2).e): A
}

\RuleToplabel{A:Bool}
{
b \in \{\word{true},\word{false}\}
}
{
\reststack;\emptyset \rtyps{K^{\word{bool}}}{0} b:\type{bool}
}

\RuleToplabel{A:Var}
{
}
{
\reststack;x: A \rtyps{K^{\word{var}}}{0} x: A
}

\RuleToplabel{A:Nil}
{
	A \in \datatypes
}
{
	\reststack;\emptyset \rtyps{K^{\word{nil}}}{0} \word{nil}:L^{p}(A)
}

\RuleToplabel{A:Cons}
{
}
{
	\reststack;x_h:A,x_t:L^{p}(A) \rtyps{p + K^{\word{cons}}}{0} \text{cons}(x_h,x_t):L^{p}(A)
}

\RuleToplabel{A:Match-L}
{
	\reststack;\rescontext \rtyps{q - K^{\word{matchN}}}{q'} e_1: A_1 \\
	\reststack;\rescontext,x_h:A,x_t:L^{p}(A) \rtyps{q + p - K^{\word{matchL}}}{q'} e_2:A_1
}
{
	\reststack;\rescontext,x:L^{p}(A) \rtyps{q}{q'} \text{match}(x,e_1,(x_h,x_t).e_2):A_1
}

\RuleToplabel{A:Share}
{
	\reststack;\rescontext,x_1:A_1,x_2:A_2 \rtyps{q}{q'} e:B \\\\
	\share(A \mid A_1,A_2)
}
{
	\reststack;\rescontext,x:A \rtyps{q}{q'} \text{share}(x,(x_1,x_2).e):B
}
\end{mathpar}
\caption{Common syntax-directed typing rules for upper bounds, constant, and lower bounds}
\label{fig:commonrules}
\end{figure*}

\begin{figure*}[!ht]
\begin{mathpar}
\small
\RuleToplabel{U:Relax}
{
	\reststack;\rescontext \rtyps{p}{p'} e:A \\\\
	q \geq p \\
	q - p \geq q' - p'
}
{
	\reststack;\rescontext \rtyps{q}{q'} e:A
}

\RuleToplabel{U:Weakening}
{
	\reststack;\rescontext \rtyps{q}{q'} e:B
}
{
	\reststack;\rescontext,x:A \rtyps{q}{q'} e:B
}

\RuleToplabel{U:Subtype}
{
	\reststack;\rescontext \rtyps{q}{q'} e:A \\
	B \subtype A
}
{
	\reststack;\rescontext \rtyps{q}{q'} e:B
}

\RuleToplabel{U:Supertype}
{
	\reststack;\rescontext,x:B \rtyps{q}{q'} e:C \\
	B \subtype A
}
{
	\reststack;\rescontext,x:A \rtyps{q}{q'} e:C
}
\end{mathpar}
\caption{Relax and structural typing rules for upper bounds}
\label{fig:upperbound}
\end{figure*}

\begin{figure*}[!ht]
\begin{mathpar}
\small
\RuleToplabel{C:Relax}
{
	\reststack;\rescontext \rtyps{p}{p'} e:A \\
	q \geq p \\
	q - p = q' - p'
}
{
	\reststack;\rescontext \rtyps{q}{q'} e:A
}

\RuleToplabel{C:Weakening}
{
	\reststack;\rescontext \rtyps{q}{q'} e:B \\
	\share(A \mid A,A)
}
{
	\reststack;\rescontext,x:A \rtyps{q}{q'} e:B
}
\end{mathpar}
\caption{Relax and structural typing rules for constant-resource}
\label{fig:constant}
\end{figure*}
\begin{figure*}[!ht]
\begin{mathpar}
\small
\RuleToplabel{L:Relax}
{
	\reststack;\rescontext \rtyps{p}{p'} e:A \\
	q \geq p \\
	q - p \leq q' - p'
}
{
	\reststack;\rescontext \rtyps{q}{q'} e:A
}

\RuleToplabel{L:Weakening}
{
	\reststack;\rescontext \rtyps{q}{q'} e:B \\
	\share(A \mid A,A)
}
{
	\reststack;\rescontext,x:A \rtyps{q}{q'} e:B
}

\RuleToplabel{L:Subtype}
{
	\reststack;\rescontext \rtyps{q}{q'} e:A \\
	A \subtype B
}
{
	\reststack;\rescontext \rtyps{q}{q'} e:B
}

\RuleToplabel{L:Supertype}
{
	\reststack;\rescontext,x:B \rtyps{q}{q'} e:C \\
	A \subtype B
}
{
	\reststack;\rescontext,x:A \rtyps{q}{q'} e:C
}

\RuleToplabel{L:Contraction}
{
	\reststack;\rescontext,x_1:A,x_2:A \rtyps{q}{q'} e:B
}
{
  \reststack;\rescontext,x:A \rtyps{q}{q'} \text{share}(x,(x_1,x_2).e):B
}
\end{mathpar}
\caption{Relax and structural typing rules for lower bounds}
\label{fig:lowerbound}
\end{figure*}

\section*{Proofs of Lemma \ref{lem:environmentpotential}, Theorem \ref{theo:soundness}, Theorem \ref{theo:constant}, Theorem \ref{theo:lowersoundness}, and Theorem \ref{theo:lowerbound}}

\subsection*{Proof of Lemma \ref{lem:environmentpotential}}
The claim is proved by induction on the definitions of potential and size-equivalence, in which $\sizeval{\ev_1(x)} \esim{} \sizeval{\ev_2(x)}$ implies $\poten{}(\ev_1(x):\rescontext(x)) = \poten{}(\ev_2(x):\rescontext(x))$.

\subsection*{Proof of Theorem \ref{theo:soundness}}
The proof is done by induction on the length of the derivation of the evaluation judgement and the typing judgement with lexical order, in which the derivation of the evaluation judgement takes priority over the typing derivation. We need to do induction on the length of both evaluation and typing derivations since on one hand, an induction of only typing derivation would fail for the case of function application, which increases the length of the typing derivation, while the length of the evaluation derivation never increases. On the other hand, if the rule \textsc{C:Weakening} is the final step in the derivation, then the length of typing derivation decreases, while the length of evaluation derivation is unchanged.
\paragraph{\textsc{A:Share}} 
Assume that the typing derivation ends with an application of the rule \textsc{A:Share}, thus $\reststack;\rescontext,x_1:A_1,x_2:A_2 \rtyps{q}{q'} e:B$ and $\share(A|A_1,A_2)$.

Let $\ev_1 = \ev \setminus \{x\} \cup \{[x_1 \mapsto \ev(x),x_2 \mapsto \ev(x)]\}$. Since $\models \ev: \rescontext,x: A$ and following the property of the share relation we have $\models \ev_1: \rescontext,x_1:A_1,x_2:A_2$. By the induction hypothesis for $e$, it holds that for all $p, r \in \mathbb{Q}^+_0$ such that $p = q + \poten{\ev_1}(\rescontext,x_1:A_1,x_2:A_2) + r$, there exists $p' \in \mathbb{Q}^+_0$ satisfying $\ev_1 \rtyps{p}{p'} e \Downarrow v$ and $p' = q' + \poten{}(v:B) + r$.

Because $\poten{}(\ev(x):A) = \poten{}(\ev_1(x_1):A_1) + \poten{}(\ev_1(x_2):A_2)$ and $\poten{\ev}(\rescontext) = \poten{\ev_1}(\rescontext) = \poten{\ev \setminus \{x\}}(\rescontext)$, thus $p = q + \poten{\ev}(\rescontext,x:A) + r$ and there exists $p'$ satisfying $\ev \rtyps{p}{p'} \text{share}(x,(x_1,x_2).e) \Downarrow v$.

\paragraph{\textsc{C:Weakening}} 
Suppose that the typing derivation ends with an application of the rule \textsc{C:Weakening}. Thus we have $\reststack;\rescontext \rtyps{q}{q'} e:B$, in which the data type $A$ satisfies $\share(A \mid A,A)$.

Since $\models \ev: \rescontext,x:A$, it follows $\models \ev:\rescontext$. By the induction hypothesis for $e$, it holds that for all $p, r \in \mathbb{Q}^+_0$ such that $p = q + \poten{\ev}(\rescontext) + r$, there exists $p' \in \mathbb{Q}^+_0$ satisfying $\ev \rtyps{p}{p'} e \Downarrow v$ and $p' = q' + \poten{}(v:B) + r$. By the property of the share relation, $\poten{}(a:A) = 0$, then we have $p = q + \poten{\ev}(\rescontext,x:A) + r$, $\ev \rtyps{p}{p'} e \Downarrow v$ and $p' = q' + \poten{}(v:B) + r$ as required.

\paragraph{\textsc{C:Relax}} 
Suppose that the typing derivation ends with an application of the rule \textsc{C:Relax}, thus we have
$\reststack;\rescontext \rtyps{q_1}{q_1'} e:A$, $q \geq q_1$, and $q - q_1 = q' - q_1'$.

For all $p, r \in \mathbb{Q}^+_0$ such that $p = q + \poten{\ev}(\rescontext) + r = q_1 + \poten{\ev}(\rescontext) + (q - q_1) + r$, we have $\models \ev: \rescontext$. By the induction hypothesis for $e$ in the premise, there exists $p' \in \mathbb{Q}^+_0$ satisfying $\ev \rtyps{p}{p'} e \Downarrow v$ and $p' = q_1' + \poten{}(v:A) + (q - q_1) + r = q' + \poten{}(v:A) + r$.

\paragraph{\textsc{A:Var}} 
Assume that $e$ is a variable $x$. If $\reststack;x:A \rtyps{K^{\word{var}}}{0} x:A$. Thus for all $p, r \in \mathbb{Q}^+_0$ such that $p = K^{\word{var}} + \poten{}(v:A) + r$, there exists $p' = \poten{}(v:A) + r$ satisfying $\ev \rtyps{p}{p'} e \Downarrow v$.

\paragraph{\textsc{A:Unit}} It is similar to the case \textsc{A:Var}.

\paragraph{\textsc{A:Bool}} It is similar to the case \textsc{A:Var}.

\paragraph{\textsc{A:Int}} It is similar to the case \textsc{A:Var}.

\paragraph{\textsc{A:B-Op}} Assume that $e$ is an expression of the form $\text{op}_\diamond(x_1,x_2)$, where $\diamond = \{\word{and},\word{or}\}$. Thus $\reststack;x_1:\type{bool},x_2:\type{bool} \rtyps{K^{\word{op}}}{0} e:\type{bool}$ and $\models \ev: \{x_1:\type{bool},x_2:\type{bool}\}$. We have $\ev \rtyps{K^{\word{op}}}{0} e \Downarrow v$, thus for all $p, r \in \mathbb{Q}^+_0$ such that $p = K^{\word{op}} + r= K^{\word{op}} + \poten{\ev}(x_1:\type{bool},x_2:\type{bool}) + r$, there exists $p' = \poten{}(v:\type{bool}) + r = r$ satisfying $\ev \rtyps{p}{p'} e \Downarrow v$.

\paragraph{\textsc{A:I-Op}} It is similar to the case \textsc{A:B-Op}.

\paragraph{\textsc{A:IB-Op}} It is similar to the case \textsc{A:B-Op}.

\paragraph{\textsc{A:Cons}} If $e$ is of the form $\text{cons}(x_1,x_2)$, then the type derivation ends with an application of the rule \textsc{A:Cons} and the evaluation ends with the application of the rule \textsc{E:Cons}. Thus
$\reststack;x_1:A,x_2:L^{p_1}(A) \rtyps{p_1 + K^{\word{cons}}}{0} e:L^{p_1}(A)$ and $\models \ev: \{x_1:A,x_2:L^{p_1}(A)\}$.

We have $\ev \rtyps{K^{\word{cons}}}{0} e \Downarrow [v_1,...,v_n]$, where $\ev(x_1) = v_1$ and $\ev(x_2) = [v_2,\cdots,v_n]$. Let $\rescontext = x_h:A,x_t:L^{p_1}(A)$, for all $p, r \in \mathbb{Q}^+_0$ such that $p = p_1 + K^{\word{cons}} + \poten{\ev}(\rescontext) +r$, there exists $p' \in \mathbb{Q}^+_0$ satisfying $p' = \poten{}([v_1,...,v_n]:L^{p_1}(A)) + r = \poten{\ev}(\rescontext) + p_1 + r$
and $\ev \rtyps{p}{p'} e \Downarrow [v_1,...,v_n]$.

\paragraph{\textsc{A:Pair}} It is similar to the case \textsc{A:Cons}.

\paragraph{\textsc{A:Nil}} It is similar to the case \textsc{A:Cons}.

\paragraph{\textsc{A:Match-P}} Suppose that the typing derivation $\reststack;\rescontext,x:A_1 * A_2 \rtyps{q}{q'} \text{match}(x,(x_1,x_2).e): A$ ends with an application of the rule \textsc{A:Match-P}. Thus $\reststack;\rescontext,x_1:A_1,x_2:A_2 \rtyps{q - K^{\word{matchP}}}{q'} e: A$ and $\models \ev: \rescontext,x:A_1 * A_2$.

Let $\ev_1 = \ev[x_1 \mapsto v_1,x_2 \mapsto v_2]$ and $\rescontext_1 = \rescontext,x_1:A_1,x_2:A_2$. Since $\models v_1: A_1$, $\models v_2: A_2$, and $\models \ev: \rescontext$ it holds that $\models \ev_1: \rescontext_1$. For all $p, r \in \mathbb{Q}^+_0$ such that $p = q + \poten{\ev}(\rescontext,x: A_1 * A_2) + r$, thus $p - K^{\word{matchP}} = q - K^{\word{matchP}} + \poten{\ev_1}(\rescontext_1) + r$, by the induction hypothesis for $e$, there exists $p' \in \mathbb{Q}^+_0$ satisfying $p' = q' + \poten{}(v:A) + r$ and $\ev_1 \rtyps{p - K^{\word{matchP}}}{p'} e \Downarrow v$. Hence, by the rule \textsc{E:Match-P}, there exists $p' = q' + \poten{}(v:A) + r$ satisfying $\ev \rtyps{p}{p'} \text{match}(x,(x_1,x_2).e) \Downarrow v$.

\paragraph{\textsc{A:Fun}} Assume that $e$ is a function application of the form $\text{app}(f,x)$. Thus $\reststack;x: A_1 \rtyps{q + K^{\word{app}}}{q'} e: A_2$ and $\reststack(f) = A_1 \xrightarrow{q/q'} A_2$. Because the considering program is well-formed, there exists a well-typed expression $e_f$ under the typing context $\rescontext_1 = y^{\atype{f}}:A_1$ and the signature $\reststack$, or $\reststack;\rescontext_1 \rtyps{q}{q'} e_f:A_2$.

Let $\rescontext = x: A_1$, $\ev(x) = v_1$ and $\ev_1 = [y^{\atype{f}} \mapsto v_1]$, since $\models \ev:\rescontext$, it follows that $\models \ev_1:\rescontext_1$. For all $p, r \in \mathbb{Q}^+_0$ such that $p = q + K^{\word{app}} + \poten{\ev}(\rescontext) + r$, since $\poten{\ev_1}(\rescontext_1) = \poten{}(\ev_1(y^{\atype{f}}):A_1) = \poten{\ev}(\rescontext) = \poten{}(\ev(x):A_1)$, it holds that $p - K^{\word{app}} = q + \poten{\ev_1}(\rescontext_1) + r$. By the induction hypothesis for $e_f$, there exists $p' \in \mathbb{Q}^+_0$ satisfying $p' = q' + \poten{}(v:A_2) + r$ and $\ev_1 \rtyps{p_1}{p_1'} e_{\atype{f}} \Downarrow v$. Hence, $\ev \rtyps{p}{p'} e \Downarrow v$.

\paragraph{\textsc{A:If}} Suppose that $e$ is an expression of the form $\text{if}(x,e_t,e_f)$. Then one of the rules \textsc{E:If-True}
and \textsc{E:If-False} has been applied in the evaluation derivation depending on the value of $x$.

Assume that the variable $x$ is assigned the value \word{true} in $\ev$, or $\ev(x) = \word{true}$. The typing rule for $e$ has been derived by an application of the rule \textsc{A:If} using the premise on the left thus $\reststack;\rescontext \rtyps{q-K^{\word{cond}}}{q'} e_t:A$.

Let $\rescontext_1 = \rescontext,x:\type{bool}$, since $\models \ev: \rescontext_1$, it follows that $\models \ev: \rescontext$. For all $p, r \in \mathbb{Q}^+_0$ such that $p = q + \poten{\ev}(\rescontext_1) + r$, since $\poten{\ev}(\rescontext) = \poten{\ev}(\rescontext_1)$ thus $p_1 = p - K^{\word{cond}} = q - K^{\word{cond}} + \poten{\ev}(\rescontext) + r$. By the induction hypothesis for $e_t$, there exists $p_1' \in \mathbb{Q}^+_0$ satisfying $\ev \rtyps{p_1}{p_1'} e_t \Downarrow v$ and $p_1' = q' + \poten{}(v:A)$. Hence, by the rule \textsc{E:If-True}, there exists $p' = p_1'$ satisfying $\ev \rtyps{p}{p'} e \Downarrow v$ and $p' = q' + \poten{}(v:A)$. If $x$ is assigned the value \word{false} in $\ev$ then it is similar to the case $\ev(x) = \word{true}$.

\paragraph{\textsc{A:Match-L}} It is the same as the case of a conditional expression. The evaluation derivation applies one of the rules \textsc{E:Match-N} and \textsc{E:Match-L} depending on the value of $x$.

Assume that $x$ is assigned the value $[v_1,....,v_n]$ under $\ev$, or $\ev(x) = [v_1,...,v_n]$. Then, the evaluation derivation ends with an application of the rule \textsc{E:Match-L}. Let $\ev_1 = \ev[x_h \mapsto v_1,x_t \mapsto [v_2,...,v_n]]$ and $\rescontext_1 = \rescontext,x_h:A,x_t:L^{p_1}(A)$, the typing derivation ends with an application of the rule \textsc{A:Match-L}, thus $\reststack;\rescontext_1 \rtyps{q + p_1 - K^{\word{matchL}}}{q'} e_2:A_1$.

Since $\models [v_1,...,v_n]: L^{p_1}(A)$, we have $\models v_i: A, \forall i = 1,...,n$. Hence, it holds that $\models v_1: A$ and $\models [v_2,...,v_n]: L^{p_1}(A)$. Finally, we have $\models \ev_1: \rescontext_1$ (since $\models \ev: \rescontext$ implies $\models \ev_1: \rescontext$).

For all $p, r \in \mathbb{Q}^+_0$ such that $p = q + \poten{\ev}(\rescontext,x:L^{p_1}(A)) + r$, because $\poten{\ev}(\rescontext,x:L^{p_1}(A)) = \poten{\ev}(\rescontext) + n.p_1 + \Sigma^{n}_{i=1}\poten{}(v_i:A)$, $\poten{\ev_1}(\rescontext_1) = \poten{\ev_1}(\rescontext) + (n-1).p_1 + \Sigma^{n}_{i=1}\poten{}(v_i:A)$ and $\poten{\ev_1}(\rescontext) = \poten{\ev}(\rescontext)$, thus we have $\poten{\ev_1}(\rescontext_1) = \poten{\ev}(\rescontext,x:L^{p_1}(A)) - p_1$. Thus $p_2 = p - K^{\word{matchL}} = q + p_1 - K^{\word{matchL}} + \poten{\ev_1}(\rescontext_1) + r$. By the induction hypothesis for $e_2$, there exists $p_2' \in \mathbb{Q}^+_0$ satisfying $\ev_1 \rtyps{p_2}{p_2'} e_2 \Downarrow v$ and $p_2' = q' + \poten{}(v:A_1)$.  Hence, there exists $p' = p_2'$ such that $\ev \rtyps{p}{p'} e \Downarrow v$. If
$\ev(x) = \word{nil}$ then it is similar to the case \textsc{A:Match-P}.

\paragraph{\textsc{A:Let}} Assume that $e$ is an expression of the form $\text{let}(x,e_1,x.e_2)$. Hence, the evaluation derivation ends with an application of the rule \textsc{E:Let}. Let $\ev_1 = \ev[x \mapsto v_1]$ and $\rescontext = \rescontext_1,\rescontext_2$. The typing derivation ends with an application of the rule \textsc{A:Let}, thus $\reststack;\rescontext_1 \rtyps{q - K^{\word{let}}}{q_1'} e_1:A_1$ and
$\reststack;\rescontext_2,x:A_1 \rtyps{q_1'}{q'} e_2:A_2$.

For all $p, r \in \mathbb{Q}^+_0$ such that $p = q + \poten{\ev}(\rescontext) + r$, thus $p_1 = p - K^{\word{let}} = q - K^{\word{let}} + \poten{\ev}(\rescontext_1) + \poten{\ev}(\rescontext_2) + r$. Since $\models \ev: \rescontext$, we have $\models \ev: \rescontext_1$. By the induction hypothesis for $e_1$, there exists $p_1' \in \mathbb{Q}^+_0$ satisfying $\ev \rtyps{p_1}{p_1'} e_1 \Downarrow v_1$ and $p_1' = q_1' + \poten{}(v_1:A_1) + \poten{\ev}(\rescontext_2) + r$.

We have $\models \ev: \rescontext_2$, thus $\models \ev_1: \rescontext_2,x:A_1$. Again by the induction hypothesis for $e_2$, with $p_2 = p - K^{\word{let}} - (p_1 - p_1') = p_1' = q_1' + \poten{\ev_1}(\rescontext_2,x:A_1) + r$, there exists $p_2' \in \mathbb{Q}^+_0$ satisfying $\ev_1 \rtyps{p_2}{p_2'} e_2 \Downarrow v$ and $p_2' = q' + \poten{}(v:A_2) + r$. Hence, by the rule \textsc{E:Let}, there exists $p' = p_2'$ satisfying $\ev \rtyps{p}{p'} e \Downarrow v$ and $p' = q' + \poten{}(v:A_2)$.

\subsection*{Proof of Theorem \ref{theo:constant}}
First, we prove that if $\ev \rtyps{p}{p'} e \Downarrow v$ then $p - p' = q + \poten{\ev}(\rescontext) - (q' + \poten{}(v:A))$. Suppose $p - p' \not = q + \poten{\ev}(\rescontext) - (q' + \poten{}(v:A))$, there exists always some $r_1, r_2 \in \mathbb{Q}^+_0$ such that $p + r_1 = q + \poten{\ev}(\rescontext) + r_2$. Since $\ev \rtyps{p}{p'} e \Downarrow v$, we have $\ev \rtyps{p+r_1}{p'+r_1} e \Downarrow v$. By Theorem \ref{theo:soundness}, $p' + r_1 = q' + \poten{}(v:A) + r_2$, thus the assumption is contradictory.

Consider any $\ev_1$ and $\ev_2$ such that $\ev_1 \esim{X} \ev_2$, hence $\ev_1 \vdash e \Downarrow v_1$ and $\ev_2 \vdash e \Downarrow v_2$. For all $p_1, p_1' \in \mathbb{Q}^+_0$ such that $\ev_1 \rtyps{p_1}{p_1'} e \Downarrow v_1$, we have $p_1 - p_1' = q + \poten{\ev_1}(\rescontext) - (q' + \poten{}(v_1:A))$. Similarly, for all $p_2, p_2' \in \mathbb{Q}^+_0$ such that $\ev_2 \rtyps{p_2}{p_2'} e \Downarrow v_2$, $p_2 - p_2' = q + \poten{\ev_2}(\rescontext) - (q' + \poten{}(v_2:A))$. Since $\poten{\ev_1}(X) = \poten{\ev_2}(X)$ by Lemma \ref{lem:environmentpotential}, $\forall x \in \dom{\rescontext} \setminus X. \poten{}(\ev_i(x):\rescontext(x)) = 0$, and $\poten{}(v_i:A) = 0$, $i = 1,2$. Thus $p_1 - p_1' = p_2 - p_2'$.

\subsection*{Proof of Theorem \ref{theo:lowersoundness}}
The proof is relied on Theorem \ref{theo:lowerbound}. For all $p, r \in \mathbb{Q}^+_0$ such that $p < q + \poten{\ev}(\rescontext) + r$, assume that there exists some $p' \in \mathbb{Q}^+_0$ such that $\ev \rtyps{p}{p'} e \Downarrow v$ and $p' \geq q' + \poten{}(v:A) + r$. Thus we have $p - p' < q + \poten{\ev}(\rescontext) - (q' + \poten{}(v:A))$. On the other hand, it holds that $q + \poten{\ev}(\rescontext) - (q' + \poten{}(v:A)) \leq p - p'$. The assumption is contradictory.

\subsection*{Proof of Theorem \ref{theo:lowerbound}}
The proof is done by induction on the length of the derivation of the evaluation judgement $\ev \rtyps{p}{p'} e \Downarrow v$ and the typing judgement $\tstack;\context \rtyps{q}{q'} e:A$ with lexical order, in which the derivation of the evaluation judgement takes priority over the typing derivation. We need to do induction on the length of both evaluation and typing derivations since on one hand, an induction of only typing derivation would fail for the case of function application, which increases the length of the typing derivation, while the length of the evaluation derivation never increases. On the other hand, if the rules \textsc{L:Weakening} and \textsc{A:Share} are final step in the derivation, then the length of typing derivation decreases, while the length of evaluation derivation is unchanged.

\paragraph{\textsc{A:Share}} Assume that the typing derivation ends with an application
of the rule \textsc{A:Share}, thus $\reststack;\rescontext,x_1:A_1,x_2:A_2 \rtyps{q}{q'} e:B$
and $\share(A|A_1,A_2)$. Let $\ev_1 = \ev \setminus \{x\} \cup \{[x_1 \mapsto \ev(x),x_2 \mapsto \ev(x)]\}$. Since $\models \ev: \rescontext,x: A$ and following the property of the share relation we have $\models \ev_1: \rescontext,x_1:A_1,x_2:A_2$.

For all $p, p' \in \mathbb{Q}^+_0$ such that $\ev \rtyps{p}{p'} \text{share}(x,(x_1,x_2).e \Downarrow v$, by the rule
\textsc{E:Share} we have $\ev_1 \rtyps{p}{p'} e \Downarrow v$. Hence, by the induction hypothesis for $e$ in the premise,
it holds that
$q + \poten{\ev_1}(\rescontext,x_1:A_1,x_2:A_2) - (q' + \poten{}(v:B)) \leq p - p'$.

Because $\poten{}(\ev(x):A) = \poten{}(\ev_1(x_1):A_1) + \poten{}(\ev_1(x_2):A_2)$ and
$\poten{\ev}(\rescontext) = \poten{\ev_1}(\rescontext) = \poten{\ev \setminus \{x\}}(\rescontext)$,
we have $q + \poten{\ev}(\rescontext,x:A) - (q' + \poten{}(v:B)) \leq p - p'$.

\paragraph{\textsc{L:Weakening}} Suppose that the typing derivation $\reststack;\rescontext,x:A \rtyps{q}{q'} e:B$ ends with an application of the rule \textsc{L:Weakening}. Thus we have $\reststack;\rescontext \rtyps{q}{q'} e:B$, in which the data type $A$ satisfies $\share(A \mid A,A)$. Since $\models \ev: \rescontext,x:A$, it follows that $\models \ev:\rescontext$.

For all $p, p' \in \mathbb{Q}^+_0$ such that $\ev \rtyps{p}{p'} e \Downarrow v$, by the induction hypothesis for $e$ in the premise, it holds that $q + \poten{\ev}(\rescontext) - (q' + \poten{}(v:B)) \leq p - p'$. By the property of the share relation, $\poten{}(a:A) = 0$, hence we have $q + \poten{\ev}(\rescontext,x:A) - (q' + \poten{}(v:B)) \leq p - p'$.

\paragraph{\textsc{L:Relax}} Suppose that the typing derivation ends with an application of the rule \textsc{L:Relax}, thus we have $\reststack;\rescontext \rtyps{q_1}{q_1'} e:A$, $q \geq q_1$, and $q - q_1 \leq q' - q_1'$.

For all $p, p' \in \mathbb{Q}^+_0$ such that $\ev \rtyps{p}{p'} e \Downarrow v$, we have $\models \ev: \rescontext$, hence by the induction hypothesis for $e$ in the premise, it holds that $q_1 + \poten{\ev}(\rescontext) - (q_1' + \poten{}(v:A)) \leq p - p'$. We have
$q + \poten{\ev}(\rescontext) - (q' + \poten{}(v:A)) = q_1 + \poten{\ev}(\rescontext) - (q_1' + \poten{}(v:A)) + ((q - q_1) - (q' - q_1'))$. Since $q - q_1 \leq q' - q_1'$, it holds that $q + \poten{\ev}(\rescontext) - (q' + \poten{}(v:A)) \leq q_1 + \poten{\ev}(\rescontext) - (q_1' + \poten{}(v:A)) \leq p - p'$.

\paragraph{\textsc{A:Var}} Assume that $e$ is a variable $x$. If $\reststack;x:A \rtyps{K^{\word{var}}}{0} x:A$. Thus for all $p, p' \in \mathbb{Q}^+_0$ such that $\ev \rtyps{p}{p'} e \Downarrow v$, we have $p = p' + K^{\word{var}}$, hence $K^{\word{var}} + \poten{}(\ev(x):A) - \poten{}(v:A) \leq p - p' = K^{\word{var}}$.

\paragraph{\textsc{A:Unit}} It is similar to the case \textsc{A:Var}.

\paragraph{\textsc{A:Bool}} It is similar to the case \textsc{A:Var}.

\paragraph{\textsc{A:Int}} It is similar to the case \textsc{A:Var}.

\paragraph{\textsc{A:B-Op}} Assume that $e$ is an expression of the form $\text{op}_\diamond(x_1,x_2)$, where $\diamond = \{\word{and},\word{or}\}$. Thus $\reststack;x_1:\type{bool},x_2:\type{bool} \rtyps{K^{\word{op}}}{0} e:\type{bool}$ and $\models \ev: \{x_1:\type{bool},x_2:\type{bool}\}$.

For all $p, p' \in \mathbb{Q}^+_0$ such that $\ev \rtyps{p}{p'} e \Downarrow v$, we have $K^{\word{op}} + \poten{\ev}(x_1:\type{bool},x_2:\type{bool}) - \poten{}(v:\type{bool}) = K^{\word{op}} \leq p - p' = K^{\word{op}}$.

\paragraph{\textsc{A:IB-Op}} It is similar to the case \textsc{A:B-Op}.

\paragraph{\textsc{A:I-Op}} It is similar to the case \textsc{A:B-Op}.

\paragraph{\textsc{A:Cons}} If $e$ is of the form $\text{cons}(x_1,x_2)$, then the typing derivation ends with an application of the rule \textsc{A:Cons} and the evaluation derivation ends with the application of the rule \textsc{E:Cons}. Thus $\reststack;x_1:A,x_2:L^{p_1}(A) \rtyps{p_1 + K^{\word{cons}}}{0} e:L^{p_1}(A)$ and $\models \ev: \{x_1:A,x_2:L^{p_1}(A)\}$.

For all $p, p' \in \mathbb{Q}^+_0$ such that $\ev \rtyps{p}{p'} e \Downarrow v$, we have $p - p' = K^{\word{cons}}$, $\ev(x_1) = v_1$ and $\ev(x_2) = [v_2,\cdots,v_n]$. Let $\rescontext = x_h:A,x_t:L^{p_1}(A)$, it holds that $p_1 + K^{\word{cons}} + \poten{\ev}(\rescontext) - (\poten{}([v_1,...,v_n]:L^{p_1}(A))) = K^{\word{cons}} \leq p - p'$.

\paragraph{\textsc{A:Pair}} It is similar to the case \textsc{A:Cons}.

\paragraph{\textsc{A:Nil}} It is similar to the case \textsc{A:Cons}.

\paragraph{\textsc{A:Match-P}} Suppose that the typing derivation $\reststack;\rescontext,x:A_1 * A_2 \rtyps{q}{q'} \text{match}(x,(x_1,x_2).e): A$ ends with an application of the rule \textsc{A:Match-P}. Thus $\reststack;\rescontext,x_1:A_1,x_2:A_2 \rtyps{q - K^{\word{matchP}}}{q'} e: A$ and $\models \ev: \context,x:A_1 * A_2$.

Let $\ev_1 = \ev[x_1 \mapsto v_1,x_2 \mapsto v_2]$ and $\rescontext_1 = \rescontext,x_1:A_1,x_2:A_2$, since $\models v_1: A_1$, $\models v_2: A_2$, and $\models \ev: \rescontext$ it holds that $\models \ev_1: \rescontext_1$.

For all $p, p' \in \mathbb{Q}^+_0$ such that $\ev \rtyps{p}{p'} e \Downarrow v$, by the rule \textsc{E:Match-P} we have $\ev_1 \rtyps{p - K^{\word{matchP}}}{p'} e \Downarrow v$. Hence, by the induction hypothesis for $e$ in the premise, it holds that $q - K^{\word{matchP}} + \poten{\ev_1}(\rescontext_1) - (q' + \poten{}(v:A)) \leq p - K^{\word{matchP}} - p'$.

Since $\poten{\ev}(\rescontext,x: A_1 * A_2) = \poten{\ev_1}(\rescontext_1)$, it follows that $q + \poten{\ev}(\rescontext,x: A_1 * A_2) - (q' + \poten{}(v:A)) \leq p - p'$.

\paragraph{\textsc{A:Fun}} Assume that $e$ is a function application of the form $\text{app}(f,x)$. Thus $\reststack;x: A_1 \rtyps{q + K^{\word{app}}}{q'} e: A_2$ and $\reststack(f) = A_1 \xrightarrow{q/q'} A_2$. Because the considering program is well-formed, there exists a well-typed expression $e_f$ under the typing context $\rescontext_1 = y^{\atype{f}}:A_1$ and the signature $\reststack$, or $\reststack;\rescontext_1 \rtyps{q}{q'} e_f:A_2$.

Let $\rescontext = x: A_1$, $\ev(x) = v_1$ and $\ev_1 = [y^{\atype{f}} \mapsto v_1]$, since $\models \ev:\rescontext$, it follows that $\models \ev_1:\rescontext_1$. For all $p, p' \in \mathbb{Q}^+_0$ such that $\ev \rtyps{p}{p'} e \Downarrow v$, we have $\ev_1 \rtyps{p - K^{\word{app}}}{p'} e_{\atype{f}} \Downarrow v$. Hence, by the induction hypothesis for $e_f$, it holds that $q + \poten{\ev_1}(\rescontext_1) - (q' + \poten{}(v:A_2)) \leq p - K^{\word{app}} - p'$.

Since $\poten{\ev_1}(\rescontext_1) = \poten{}(\ev_1(y^{\atype{f}}):A_1) = \poten{\ev}(\rescontext) = \poten{}(\ev(x):A_1)$, it follows that
$q + K^{\word{app}} + \poten{}(\ev(x):A_1) - (q' + \poten{}(v:A_2)) \leq p - p'$.

\paragraph{\textsc{A:If}} Suppose that $e$ is an expression of the form $\text{if}(x,e_t,e_f)$. Then one of the rules \textsc{E:If-True}
and \textsc{E:If-False} has been applied in the evaluation derivation depending on the value of $x$.

Assume that the variable $x$ is assigned the value \word{true} in $\ev$, or $\ev(x) = \word{true}$. The typing rule for $e$ has been derived by an application of the rule \textsc{A:If} using the premise on the left thus $\reststack;\rescontext \rtyps{q-K^{\word{cond}}}{q'} e_t:A$. Let
$\rescontext_1 = \rescontext,x:\type{bool}$, since $\models \ev: \rescontext_1$, it follows that $\models \ev: \rescontext$.

For all $p, p' \in \mathbb{Q}^+_0$ such that $\ev \rtyps{p}{p'} e \Downarrow v$, by the rule \textsc{E:If-True} we have $\ev \rtyps{p - K^{\word{cond}}}{p'} e_t \Downarrow v$. Hence, by the induction hypothesis for $e_t$, it holds that $q - K^{\word{cond}} + \poten{\ev}(\rescontext) - (q' + \poten{}(v:A)) \leq p - K^{\word{cond}} - p'$.

Because $\poten{\ev}(\rescontext) = \poten{\ev}(\rescontext_1)$, it follows $q + \poten{\ev}(\rescontext_1) - (q' + \poten{}(v:A)) \leq p - p'$. If $\ev(x) = \word{false}$ then the proof is similar.

\paragraph{\textsc{A:Match-L}} It is the same as the case of a conditional expression. The evaluation derivation applies one of the rules \textsc{E:Match-N} and \textsc{E:Match-L} depending on the value of $x$.

Assume that $x$ is assigned the value $[v_1,....,v_n]$ under $\ev$, or $\ev(x) = [v_1,...,v_n]$. Then, the evaluation derivation ends with an application of the rule \textsc{E:Match-L}. Let $\ev_1 = \ev[x_h \mapsto v_1,x_t \mapsto [v_2,...,v_n]]$ and $\rescontext_1 = \rescontext,x_h:A,x_t:L^{p_1}(A)$, the typing derivation ends with an application of the rule \textsc{A:Match-L}, thus $\reststack;\rescontext_1 \rtyps{q + p_1 - K^{\word{matchL}}}{q'} e_2:A_1$.

Since $\models [v_1,...,v_n]: L^{p_1}(A)$, we have $\models v_i: A, \forall i = 1,...,n$. Hence, it holds that $\models v_1: A$ and $\models [v_2,...,v_n]: L^{p_1}(A)$. Finally, we have $\models \ev_1: \rescontext_1$ (since $\models \ev: \rescontext$ implies $\models \ev_1: \rescontext$).

For all $p, p' \in \mathbb{Q}^+_0$ such that $\ev \rtyps{p}{p'} e \Downarrow v$, by the rule \textsc{E:Match-L} we have $\ev_1 \rtyps{p - K^{\word{matchL}}}{p'} e_2 \Downarrow v$. By the induction hypothesis for $e_2$, it holds that $q + p_1 - K^{\word{matchL}} + \poten{\ev_1}(\rescontext_1) - (q' + \poten{}(v:A_1)) \leq p - K^{\word{matchL}} - p'$.

Because $\poten{\ev}(\rescontext,x:L^{p}(A)) = \poten{\ev}(\rescontext) + n.p_1 + \Sigma^{n}_{i=1}\poten{}(v_i:A)$, $\poten{\ev_1}(\rescontext_1) = \poten{\ev_1}(\rescontext) + (n-1).p_1 + \Sigma^{n}_{i=1}\poten{}(v_i:A)$ and $\poten{\ev_1}(\rescontext) = \poten{\ev}(\rescontext)$, thus we have $\poten{\ev_1}(\rescontext_1) = \poten{\ev}(\rescontext,x:L^{p_1}(A)) - p_1$. Therefore, $q + \poten{\ev}(\rescontext,x:L^{p_1}(A)) - (q' + \poten{}(v:A_1)) \leq p - p'$. If $\ev(x) = \word{nil}$ then it is similar to the case \textsc{A:Match-P}.

\paragraph{\textsc{A:Let}} Assume that $e$ is an expression of the form $\text{let}(x,e_1,x.e_2)$. Hence, the evaluation derivation ends with an application of the rule \textsc{E:Let}. Let $\ev_1 = \ev[x \mapsto v_1]$ and $\rescontext = \rescontext_1,\rescontext_2$. The typing derivation ends with an application of the rule \textsc{A:Let}, thus $\reststack;\rescontext_1 \rtyps{q - K^{\word{let}}}{q_1'} e_1:A_1$ and $\reststack;\rescontext_2,x:A_1 \rtyps{q'_1}{q'} e_2:A_2$.

For all $p, p' \in \mathbb{Q}^+_0$ such that $\ev \rtyps{p}{p'} e \Downarrow v$, by the rule \textsc{E:Let} we have $\ev \rtyps{p - K^{\word{let}}}{p_1'} e_1 \Downarrow v_1$ and $\ev_1 \rtyps{p_1'}{p'} e_2 \Downarrow v$. Since $\models \ev: \rescontext$, we have $\models \ev: \rescontext_1$. By the induction hypothesis for $e_1$, it holds that $q - K^{\word{let}} + \poten{\ev}(\rescontext_1) - (q_1' + \poten{}(v_1:A_1)) \leq p - K^{\word{let}} - p_1'$.

We have $\models \ev: \rescontext_2$, thus $\models \ev_1: \rescontext_2,x:A_1$. Again by the induction hypothesis for $e_2$, we derive that
$q_1' + \poten{\ev_1}(\rescontext_2,x:A_1) - (q' + \poten{}(v:A_2)) \leq p_1' - p'$.

Sum two in-equations above, it follows $q + \poten{\ev}(\rescontext_1) - \poten{}(v_1:A_1) + \poten{\ev_1}(\rescontext_2,x:A_1) - (q' + \poten{}(v:A_2)) = q + \poten{\ev}(\rescontext_1,\rescontext_2) - (q' + \poten{}(v:A_2)) \leq p - p'$.

\paragraph{\textsc{L:Subtype}} Assume that the typing derivation ends with an application of the rule \textsc{L:Subtype}, thus $\reststack;\rescontext \rtyps{q}{q'} e:A$ and $A \subtype B$.

By the induction hypothesis for $e$ in the premise, for all $p, p' \in \mathbb{Q}^+_0$ such that $\ev \rtyps{p}{p'} e \Downarrow v$ it holds that
$q + \poten{\ev}(\rescontext) - (q' + \poten{}(v:A)) \leq p - p'$.

Because $\poten{}(\ev(x):A) \leq \poten{}(\ev(x):B)$ we have $q + \poten{\ev}(\rescontext) - (q' + \poten{}(v:B)) \leq p - p'$.

\paragraph{\textsc{L:Supertype}} Assume that the typing derivation ends with an application of the rule \textsc{L:Supertype}, thus $\reststack;\rescontext,x:B \rtyps{q}{q'} e:C$ and $A \subtype B$. Since $\models \ev: \rescontext,x: A$ and following the property of the subtyping relation we have $\models \ev: \rescontext,x: B$.

By the induction hypothesis for $e$ in the premise, for all $p, p' \in \mathbb{Q}^+_0$ such that $\ev \rtyps{p}{p'} e \Downarrow v$ it holds that
$q + \poten{\ev}(\rescontext,x:B) - (q' + \poten{}(v:C)) \leq p - p'$.

Because $\poten{}(\ev(x):A) \leq \poten{}(\ev(x):B)$ we have $q + \poten{\ev}(\rescontext,x:A) - (q' + \poten{}(v:C)) \leq p - p'$.

\end{document}